\documentclass[english,onecolumn]{IEEEtran}
\usepackage[T1]{fontenc}
\usepackage[latin9]{inputenc}
\usepackage{amsmath}
\usepackage{graphicx}
\usepackage{amssymb}
\usepackage{esint}

\makeatletter
\newtheorem{definitn}{Definition}
\newtheorem{thm}{Theorem}
\newtheorem{lemma}{Lemma}
\newtheorem{remrk}{Remark}
\newtheorem{cor}{Corollary}

\usepackage{cite}
\allowdisplaybreaks

\@ifundefined{showcaptionsetup}{}{%
 \PassOptionsToPackage{caption=false}{subfig}}
\usepackage{subfig}
\makeatother

\usepackage{babel}

\begin{document}

\title{Towards a Better Understanding of Large Scale Network Models}

\author{Guoqiang Mao, \emph{Senior Member, IEEE}%
\thanks{G. Mao is with the School of Electrical and Information Engineering,
the University of Sydney and National ICT Australia. Email: guoqiang.mao@sydney.edu.au. %
}, and Brian D.O. Anderson, \emph{Life Fellow, IEEE}%
\thanks{B.D.O. Anderson is with the Research School of Information Sciences
and Engineering, Australian National University and National ICT Australia.
Email: brian.anderson@anu.edu.au.%
}\emph{}%
\thanks{This research is funded by ARC Discovery project: DP0877562.%
}}
\maketitle
\begin{abstract}
Connectivity and capacity are two fundamental properties of wireless
multi-hop networks. The scalability of these properties has been a
primary concern for which asymptotic analysis is a useful tool. Three
related but logically distinct network models are often considered
in asymptotic analyses, viz. the dense network model, the extended
network model and the infinite network model, which consider respectively
a network deployed in a fixed finite area with a sufficiently large
node density, a network deployed in a sufficiently large area with
a fixed node density, and a network deployed in $\Re^{2}$ with a
sufficiently large node density. The infinite network model originated
from continuum percolation theory and asymptotic results obtained
from the infinite network model have often been applied to the dense
and extended networks. In this paper, through two case studies related
to network connectivity on the expected number of isolated nodes and
on the vanishing of components of finite order $k>1$ respectively,
we demonstrate some subtle but important differences between the infinite
network model and the dense and extended network models. Therefore
extra scrutiny has to be used in order for the results obtained from
the infinite network model to be applicable to the dense and extended
network models. Asymptotic results are also obtained on the expected
number of isolated nodes, the vanishingly small impact of the boundary
effect on the number of isolated nodes and the vanishing of components
of finite order $k>1$ in the dense and extended network models using
a generic random connection model. \end{abstract}

\section{Introduction\label{sec:Introduction}}

Wireless multi-hop networks in various forms, e.g. wireless ad hoc
networks, sensor networks, mesh networks and vehicular networks, have
been the subject of intense research in the recent decades (see \cite{Haenggi09Stochastic}
and references therein). Connectivity and capacity are two fundamental
properties of these networks. The scalability of these properties
as the number of nodes in the network becomes sufficiently large has
been a primary concern. Asymptotic analysis, valid when the number
of nodes in the network is large enough, has been useful for understanding
the characteristics of these networks.

Three related but logically distinct network models have been widely
used in the asymptotic analysis of large scale multi-hop networks.
The first model, often referred to as the \emph{dense network model,}
considers that the network is deployed in a finite area with a sufficiently
large node density. The second model, often referred to as the \emph{extended
network model}, considers that the node density is fixed and the network
area is sufficiently large. The third model, referred to as the \emph{infinite
network model}, has its origin in continuum percolation theory \cite{Meester96Continuum}.
It considers a network deployed in an infinite area, i.e. $\Re^{2}$
in 2D, and analyzes the properties of the network as the node density
becomes sufficiently large. Due to the relatively longer history of
research into continuum percolation theory and relatively abundant
results in that area, and the close connections between the infinite
network model and the dense and extended network models, results obtained
in the infinite network model are often applied straightforwardly
to the first and second models \cite{Gupta98Critical,Ta09On,Kong08Connectivity,Goeckel09Asymptotic,Ammari08Integrated,Li09Asymptotic}.

In this paper, through two case studies on key events related to the
network connectivity, i.e. the expected number of isolated nodes and
the vanishing of components of fixed and finite order $k>1$ (the
order of a component refers to the number of nodes in the component),
using a random connection model, we demonstrate some subtle but important
differences between the infinite network model and the dense and extended
network models due to the \emph{truncation effect}, to be explained
in the following paragraphs. Therefore results obtained from an infinite
network model \emph{cannot be directly applied} to the dense and extended
networks. Instead some careful analysis of the impact of the truncation
effect is required.

Here we give a detailed explanation of the above comments using a
\emph{unit disk connection model} as an example%
\footnote{In the paper, we have omitted some trivial discussions on the difference
between Poisson and uniform distributions and consider Poisson node
distribution only.%
}. Under the \emph{unit disk connection model}, two nodes are directly
connected if and only if (iff) their Euclidean distance is smaller
than or equal to a given threshold $r\left(\rho\right)$, a parameter
which is often taken as a function of a further parameter $\rho$,
to be defined shortly, under the dense and extended network models;
the parameter $r\left(\rho\right)$ is termed the \emph{transmission
range}. The dense and extended network models that are often considered
assume respectively a) nodes are Poissonly distributed in a unit area,
say a square, with density $\rho$ and $r\left(\rho\right)=\sqrt{\frac{\log\rho+c}{\pi\rho}}$
(the dense network model); b) nodes are Poissonly distributed on a
square $\sqrt{\rho}\times\sqrt{\rho}$ with density 1 and $r\left(\rho\right)=\sqrt{\frac{\log\rho+c}{\pi}}$
(the extended network model). The parameter $c$ may be either a constant;
or it can depend on $\rho$, in which case $c=o\left(\log\rho\right)$.
The corresponding infinite network model considers nodes Poissonly
distributed in $\Re^{2}$ with density $\rho$ and a pair of nodes
are directly connected iff their Euclidean distance is smaller than
or equal to $r$, which \emph{does not} depend on $\rho$. The dense
network model can be converted into the extended network model by
scaling the Euclidean distances between \emph{all pairs} of nodes
by a factor of $\sqrt{\rho}$ while maintaining their connections,
and conversely. Therefore the dense network model and the extended
network model are equivalent in the analysis of connectivity. In the
extended network model, as $\rho\rightarrow\infty$, the network area
approaches $\Re^{2}$ and the average node degree approaches infinity
following $\Theta\left(\log\rho\right)$, i.e. a node has more and
more connections as $\rho\rightarrow\infty$. This resembles the situation
that occurs in the infinite network model as $\rho\rightarrow\infty$.
This close connection between the infinite network model and the dense
and extended network models creates the \emph{illusion} that as $\rho\rightarrow\infty$
results obtained in the infinite network model can also be applied
directly to the dense and extended models, e.g. those dealing with
the vanishing of isolated nodes, the uniqueness of the component of
infinite order, the vanishing of components of finite order $k>1$
\cite{Gupta98Critical,Ta09On,Kong08Connectivity,Goeckel09Asymptotic,Ammari08Integrated,Li09Asymptotic}. 

Starting from the dense network model however, if we scale the Euclidean
distances between all pairs of nodes by a factor $1/\sqrt{\frac{\log\rho+c}{\pi\rho}}$,
there results a network on a square $1/\sqrt{\frac{\log\rho+c}{\pi\rho}}\times1/\sqrt{\frac{\log\rho+c}{\pi\rho}}$
with node density $\frac{\log\rho+c}{\pi}$, where $\frac{\log\rho+c}{\pi}\rightarrow\infty$
as $\rho\rightarrow\infty$, and a pair of nodes are directly connected
iff their Euclidean distance is equal to or smaller than $r=1$, \emph{independently}
of the node density. This latter network model is also equivalent
to the dense and extended network models in connectivity. On the other
hand, this latter network can also be obtained from an infinite network
on $\Re^{2}$ with node density $\frac{\log\rho+c}{\pi}$ and $r=1$
by removing all nodes and the associated connections outside a square
of $1/\sqrt{\frac{\log\rho+c}{\pi\rho}}\times1/\sqrt{\frac{\log\rho+c}{\pi\rho}}$
in $\Re^{2}$. We term the effect associated with the above removal
procedure as the \emph{truncation effect}. From the above discussion,
it is clear that a prerequisite for the results obtained in the infinite
network model to be applicable to the dense or extended network models
is that the impact of the truncation effect on the property concerned
must be vanishingly small as $\rho\rightarrow\infty$. 

The main contributions of this paper are:
\begin{itemize}
\item Through two case studies, one on the expected number of isolated nodes
and the other on the vanishing of components of fixed and finite order
$k>1$, using a random connection model, we show however that ensuring
the impact of the truncation effect is vanishingly small either requires
imposing a stronger requirement on the connection function or needs
some non-trivial analysis to rule out the possibility of occurrence
of some events associated with the truncation effect. Therefore results
obtained assuming an infinite network model \emph{cannot} be applied
directly to the dense and extended network models.
\item In particular, we show that in order for the impact of the truncation
effect on the number of isolated nodes to be vanishingly small, a
stronger requirement on the connection function (than the usual requirements
of rotational invariance, integral boundedness and non-increasing
monotonicity) needs to be imposed. 
\item We show that some non-trivial analysis is required to rule out the
possibility of occurrence of some events associated with the truncation
effect in order to establish the result on the vanishing of components
of components of fixed and finite order $k>1$ in the dense and extended
network models. For example, an infinite component in $\Re^{2}$ may,
after truncation, yield multiple \emph{components of extremely large
order}%
\footnote{It is trivial to show that for any finite $\rho$, \emph{almost surely}
there is no infinite component in a network whose nodes are Poissonly
distributed with density $\frac{\log\rho+c}{\pi}$ on a square of
$1/\sqrt{\frac{\log\rho+c}{\pi\rho}}\times1/\sqrt{\frac{\log\rho+c}{\pi\rho}}$.
Therefore we use the term \emph{components of extremely large order}
to refer to those components whose order may become \emph{asymptotically
infinite} as $\rho\rightarrow\infty$.%
}, finite components of fixed order $k>1$ and isolated nodes in $1/\sqrt{\frac{\log\rho+c}{\pi\rho}}\times1/\sqrt{\frac{\log\rho+c}{\pi\rho}}$,
where these components are only connected via nodes and associated
connections in the infinite component but outside $1/\sqrt{\frac{\log\rho+c}{\pi\rho}}\times1/\sqrt{\frac{\log\rho+c}{\pi\rho}}$.
Thus the dense and extended networks may still possibly have finite
components of order $k>1$ even though the infinite network can be
shown to \emph{asymptotically almost surely} have no such finite components
as $\rho\rightarrow\infty$. 
\item Asymptotic results are established on the expected number of isolated
nodes, the vanishingly small impact of the boundary effect on the
number of isolated nodes and the vanishing of components of finite
order $k>1$ in the dense and extended network models using a generic
random connection model. These results form key steps in extending
asymptotic results on network connectivity from the unit disk model
to the more generic random connection model. 
\end{itemize}
To our knowledge, this is the first paper that has provided solid
theoretical analysis to explain the difference between the infinite
network model and the dense and extended network models and the cause
of this difference, i.e. it is attributable to the truncation effect,
which is different from the boundary effect that has been widely studied.

The rest of the paper is organized as follows. Section \ref{sec:Related-Work}
reviews related work. Section \ref{sec:Network-Models} gives a formal
definition of the network models, symbols and notations considered
in the paper. Section \ref{sec:expected number of isolated nodes}
comparatively studies the expected number of isolated nodes in a dense
(or extended) network and in its counterpart infinite network model.
Through the study, it shows that under certain conditions the impact
of the truncation effect on the expected number of isolated nodes
is non-negligible or may even be the dominant factor. Section \ref{sec:Components of finite size}
first gives an example to show that asymptotic vanishing of components
of fixed and finite order $k>1$ in an infinite network does not carry
straightforwardly the conclusion that components of fixed and finite
order $k>1$ also vanish asymptotically in the dense and extended
networks. Then to fill this theoretical gap and with a supplementary
condition holding, a result is presented on the asymptotic vanishing
of components of fixed and finite order $k>1$ in the dense and extended
network models under a random connection model. Finally Section \ref{sec:Conclusion}
summarizes conclusions and future work.

\section{Related Work\label{sec:Related-Work}}

Extensive research has been done on connectivity problems using the
well-known random geometric graph and the unit disk connection model,
which is usually obtained by randomly and uniformly distributing $n$
vertices in a given area and connecting any two vertices iff their
 distance is smaller than or equal to a given threshold $r(n)$ \cite{Penrose99On,Penrose03Random}.
Significant outcomes have been obtained \cite{Gupta98Critical,Xue04The,Philips89Connectivity,Ravelomanana04Extremal,Balister05Connectivity,Wan04Asymptotic,Penrose03Random,Balister09A}. 

Penrose \cite{Penrose97The,Penrose99A} and Gupta et al. \cite{Gupta98Critical}
proved using different techniques that if the transmission range is
set to $r\left(n\right)=\sqrt{\frac{\log n+c\left(n\right)}{\pi n}}$,
a random network formed by uniformly placing $n$ nodes in a unit-area
disk in $\Re^{2}$ is asymptotically almost surely connected as $n\rightarrow\infty$
iff $c\left(n\right)\rightarrow\infty$. Specifically, Penrose's result
is based on the fact that in the above random network as $\rho\rightarrow\infty$
the longest edge of the minimum spanning tree converges in probability
to the minimum transmission range required for the above random network
to have no isolated nodes (or equivalently the longest edge of the
nearest neighbor graph of the above network) \cite{Penrose97The,Penrose99A,Penrose03Random}.
Gupta and Kumar's result is based on a key finding in continuum percolation
theory \cite[Chapter 6]{Meester96Continuum}: Consider an \emph{infinite
network} with nodes distributed on $\Re^{2}$ following a Poisson
distribution with density $\rho$; and suppose that\textbf{ }a\textbf{
}pair of nodes separated by a Euclidean distance $x$ are directly
connected with probability $g\left(x\right)$, independent of the
event that another distinct pair of nodes are directly connected.
Here, $g:\Re^{+}\rightarrow\left[0,1\right]$ satisfies the conditions
of rotational invariance, non-increasing monotonicity and integral
boundedness \cite[pp. 151-152]{Meester96Continuum}. As $\rho\rightarrow\infty$
\emph{asymptotically almost surely} the above network on $\Re^{2}$
has only a unique infinite component and isolated nodes.

In \cite{Philips89Connectivity}, Philips et al. proved that the average
node degree, i.e. the expected number of neighbors of an arbitrary
node, must grow logarithmically with the area of the network to ensure
that the network is connected, where nodes are placed randomly on
a square according to a Poisson point process with a known density
in $\Re^{2}$. This result by Philips et al. actually provides a necessary
condition on the average node degree required for connectivity. In
\cite{Xue04The}, Xue et al. showed that in a network with a total
of $n$ nodes randomly and uniformly distributed in a unit square
in $\Re^{2}$, if each node is connected to $c\log n$ nearest neighbors
with $c\leq0.074$ then the resulting random network is asymptotically
almost surely disconnected as $n\rightarrow\infty$; and if each node
is connected to $c\log n$ nearest neighbors with $c\geq5.1774$ then
the network is asymptotically almost surely connected as $n\rightarrow\infty$.
In \cite{Balister05Connectivity}, Balister et al. advanced the results
in \cite{Xue04The} and improved the lower and upper bounds to $0.3043\log n$
and $0.5139\log n$ respectively. In a more recent paper \cite{Balister09A},
Balister et al. achieved much improved results by showing that there
exists a constant $c_{crit}$ such that if each node is connected
to $\left\lfloor c\log n\right\rfloor $ nearest neighbors with $c<c_{crit}$
then the network is asymptotically almost surely disconnected as $n\rightarrow\infty$,
and if each node is connected to $\left\lfloor c\log n\right\rfloor $
nearest neighbors with $c>c_{crit}$ then the network is asymptotically
almost surely connected as $n\rightarrow\infty$. In both \cite{Balister05Connectivity}
and \cite{Balister09A}, the authors considered nodes randomly distributed
following a Poisson process of intensity one in a square of area $n$
in $\Re^{2}$. In \cite{Ravelomanana04Extremal}, Ravelomanana investigated
the critical transmission range for connectivity in 3-dimensional
wireless sensor networks and derived similar results to the 2-dimensional
results in \cite{Gupta98Critical}. 

All the above work is based on the unit disk connection model. The
unit disk connection model may simplify analysis but no real antenna
has an antenna pattern similar to it. The log-normal shadowing connection
model, which is more realistic than the unit disk connection model,
has accordingly been considered for investigating network connectivity
in \cite{Hekmat06Connectivity,Orriss03Probability,Miorandi05Coverage,Miorandi08The,Bettstetter04failure,Bettstetter05Connectivity}.
Under the log-normal shadowing connection model, two nodes are directly
connected if the received power at one node from the other node, whose
attenuation follows the log-normal model \cite{Rappaport02Wireless},
is greater than a given threshold. In \cite{Hekmat06Connectivity,Orriss03Probability,Miorandi05Coverage,Miorandi08The,Bettstetter04failure,Bettstetter05Connectivity},
the authors investigated from different perspectives the necessary
condition for a network with nodes uniformly or Poissonly distributed
in a bounded area in $\Re^{2}$ and a pair of nodes are directly connected
following the log-normal connection model to be connected. Most of
the above work is based on the observation that a necessary condition
for a connected network is that the network has no isolated nodes.
Their analysis \cite{Hekmat06Connectivity,Orriss03Probability,Miorandi05Coverage,Miorandi08The,Bettstetter04failure,Bettstetter05Connectivity}
also relies on the assumption that under the log-normal connection
model, the node isolation events are independent, an assumption yet
to be validated analytically.

Other work in the area include \cite{Dousse05Impact,Goeckel09Asymptotic,Li09Asymptotic,Kong08Connectivity},
which studies from the percolation perspective, the impact of mutual
interference caused by simultaneous transmissions, the impact of physical
layer cooperative transmissions, the impact of directional antennas
and the impact of unreliable links on connectivity respectively.

In this paper we discuss the relation between three widely used network
models in the above studies, i.e. the dense network model, the extended
network model and the infinite network model which originated from
continuum percolation theory. We examine mainly from the connectivity
perspective the similarities and differences between these models
and demonstrate that results obtained from continuum percolation theory
assuming an infinite network model \emph{cannot} be directly applied
to the dense and extended network models. We also establish some results
that form key steps in extending asymptotic results on network connectivity
from the unit disk model to the more generic random connection model.

\section{Network Models\label{sec:Network-Models}}

In this section we give a formal definition of network models considered
in the paper. Let $g:\Re^{+}\rightarrow\left[0,1\right]$ be a function
satisfying the conditions of non-increasing monotonicity and integral
boundedness %
\footnote{Throughout this paper, we use the non-bold symbol, e.g. $x$, to denote
a scalar and the bold symbol, e.g. $\boldsymbol{x}$, to denote a
vector.%
}, %
\footnote{We refer readers to \cite[Chapter 6]{Franceschetti07Random,Meester96Continuum}
for detailed discussions on the random connection model.%
}:\begin{eqnarray}
g\left(x\right)\leq g\left(y\right) &  & \textrm{whenever}\;\; x\geq y\label{eq:conditions on g(x) - non-increasing}\\
0<\int_{\Re^{2}}g\left(\left\Vert \boldsymbol{x}\right\Vert \right)d\boldsymbol{x}<\infty\label{eq:conditions on g(x) - integral boundness}\end{eqnarray}
where $\left\Vert \boldsymbol{x}\right\Vert $ denotes the Euclidean
norm of $\boldsymbol{x}$. The function $g$ is the connection function
that has been widely considered in the random connection model \cite[Chapter 6]{Franceschetti07Random,Meester96Continuum}.
Further the requirement of rotational invariance on the connection
function in the random connection model \cite[Chapter 6]{Franceschetti07Random,Meester96Continuum}
has been met implicitly by letting $g$ be a function of a scalar,
typically representing the Euclidean distance between two nodes being
considered. 

The following notations and definitions are used throughout the paper:
\begin{itemize}
\item $f\left(z\right)=o_{z}\left(h\left(z\right)\right)$ iff $\lim_{z\rightarrow\infty}\frac{f\left(z\right)}{h\left(z\right)}=0$;
\item $f\left(z\right)=\omega_{z}\left(h\left(z\right)\right)$ iff $h\left(z\right)=o_{z}\left(f\left(z\right)\right)$;
\item $f\left(z\right)=\Theta_{z}\left(h\left(z\right)\right)$ iff there
exist a sufficiently large $z_{0}$ and two positive constants $c_{1}$
and $c_{2}$ such that for any $z>z_{0}$, $c_{1}h\left(z\right)\geq f\left(z\right)\geq c_{2}h\left(z\right)$;
\item $f\left(z\right)\sim_{z}h\left(z\right)$ iff $\lim_{z\rightarrow\infty}\frac{f\left(z\right)}{h\left(z\right)}=1$;
\item An event $\xi$ is said to occur almost surely if its probability
equals to one;
\item An event $\xi_{z}$ depending on $z$ is said to occur asymptotically
almost surely (a.a.s.) if its probability tends to one as $z\rightarrow\infty$.
\end{itemize}
The above definition applies whether the argument $z$ is continuous
or discrete, e.g. assuming integer values.

Using the integral boundedness condition on $g$ and the non-increasing
property of $g$, it can be shown that\begin{eqnarray*}
\int_{\Re^{2}}g\left(\left\Vert \boldsymbol{x}\right\Vert \right)d\boldsymbol{x} & = & \lim_{z\rightarrow\infty}\int_{0}^{z}2\pi xg\left(x\right)dx\end{eqnarray*}
and\[
\lim_{z\rightarrow\infty}\int_{z}^{\infty}2\pi xg\left(x\right)dx=0\]
The above equation, together with the following derivations

\begin{eqnarray*}
 &  & \lim_{z\rightarrow\infty}\int_{z}^{\infty}2\pi xg\left(x\right)dx\\
 & \geq & \lim_{z\rightarrow\infty}\int_{z}^{2z}2\pi xg\left(x\right)dx\\
 & \geq & \lim_{z\rightarrow\infty}\int_{z}^{2z}2\pi xg\left(2z\right)dx\\
 & = & \lim_{z\rightarrow\infty}3\pi z^{2}g\left(2z\right)\end{eqnarray*}
allow us to conclude that \begin{equation}
g\left(x\right)=o_{x}\left(\frac{1}{x^{2}}\right)\label{eq:scaling property of g(x)}\end{equation}

From time to time, we may require $g$ to satisfy the more restrictive
requirement that \begin{equation}
g\left(x\right)=o_{x}\left(\frac{1}{x^{2}\log^{2}x}\right)\label{eq:Condition on g(x) requirement 2}\end{equation}
and \eqref{eq:conditions on g(x) - non-increasing}. When we do impose
such additional constraint, we will specify it clearly. It is obvious
that conditions \eqref{eq:conditions on g(x) - non-increasing} and
\eqref{eq:conditions on g(x) - integral boundness} imply \eqref{eq:scaling property of g(x)}
while condition \eqref{eq:Condition on g(x) requirement 2} implies
\eqref{eq:conditions on g(x) - integral boundness} and \eqref{eq:scaling property of g(x)}. 

In the following analysis, we will only use \eqref{eq:conditions on g(x) - non-increasing}
and \eqref{eq:Condition on g(x) requirement 2} (instead of \eqref{eq:conditions on g(x) - non-increasing}
and \eqref{eq:conditions on g(x) - integral boundness}) when necessary.
This helps to identify which part of the analysis relies on the more
restrictive requirement on $g$. In our analysis, we assume that $g$
has infinite support when necessary. Our results however apply to
the situation when $g$ has bounded support, which forms a special
case  and only makes the analysis easier.

Further, define \begin{equation}
r_{\rho}\triangleq\sqrt{\frac{\log\rho+b}{C\rho}}\label{eq:definition of r_rho}\end{equation}
 for some non-negative value $\rho$, where \begin{equation}
0<C=\int_{\Re^{2}}g\left(\left\Vert \boldsymbol{x}\right\Vert \right)d\boldsymbol{x}<\infty\label{eq:definition of C}\end{equation}
 and $b$ is a constant ($+\infty$ is allowed).

In the following, we give the formal definitions of four network models
discussed in the paper. The motivation for defining a new model in
Definition  \ref{def:network model discussed in the paper} appears
later after all models are defined.
\begin{definitn}
\label{def:(dense-network-model)}(dense network model) Let $\mathcal{G}\left(\mathcal{X}_{\rho},g_{r_{\rho}},A\right)$
be a network with nodes Poissonly distributed on a unit square $A\triangleq\left[-\frac{1}{2},\frac{1}{2}\right]^{2}$
with density $\rho$ and a pair of nodes separated by a Euclidean
distance $x$ are directly connected with probability $g_{r_{\rho}}\left(x\right)\triangleq g\left(\frac{x}{r_{\rho}}\right)$,
independent of the event that another distinct pair of nodes are directly
connected. $\mathcal{X}_{\rho}$ denotes the vertex set in $\mathcal{G}\left(\mathcal{X}_{\rho},g_{r_{\rho}},A\right)$.
\end{definitn}

\begin{definitn}
\label{def:(extended-network-model)}(extended network model) Let
$\mathcal{G}\left(\mathcal{X}_{1},g_{\sqrt{\frac{\log\rho+b}{C}}},A_{\sqrt{\rho}}\right)$
be a network with nodes Poissonly distributed on a square $A_{\sqrt{\rho}}\triangleq\left[-\frac{\sqrt{\rho}}{2},\frac{\sqrt{\rho}}{2}\right]^{2}$
with density 1 and a pair of nodes separated by a Euclidean distance
$x$ are directly connected with probability $g_{\sqrt{\frac{\log\rho+b}{C}}}\left(x\right)\triangleq g\left(\frac{x}{\sqrt{\frac{\log\rho+b}{C}}}\right)$,
independent of the event that another distinct pair of nodes are directly
connected. $\mathcal{X}_{1}$ denotes the vertex set in $\mathcal{G}\left(\mathcal{X}_{1},g_{\sqrt{\frac{\log\rho+b}{C}}},A_{\sqrt{\rho}}\right)$.
\end{definitn}

\begin{definitn}
\label{def:network model discussed in the paper}Let $\mathcal{G}\left(\mathcal{X}_{\frac{\log\rho+b}{C}},g,A_{\frac{1}{r_{\rho}}}\right)$
be a network with nodes Poissonly distributed on a square $A_{\frac{1}{r_{\rho}}}\triangleq\left[-\frac{1}{2r_{\rho}},\frac{1}{2r_{\rho}}\right]^{2}$
with density $\frac{\log\rho+b}{C}$ and a pair of nodes separated
by a Euclidean distance $x$ are directly connected with probability
$g\left(x\right)$, independent of the event that another distinct
pair of nodes are directly connected. $\mathcal{X}_{\frac{\log\rho+b}{C}}$
denotes the vertex set in $\mathcal{G}\left(\mathcal{X}_{\frac{\log\rho+b}{C}},g,A_{\frac{1}{r_{\rho}}}\right)$.
\end{definitn}

\begin{definitn}
\label{def:(infinite-network-model)}(infinite network model) Let
$\mathcal{G}\left(\mathcal{X}_{\rho},g,\Re^{2}\right)$ be a network
with nodes Poissonly distributed on $\Re^{2}$ \textbf{with density
$\rho$} and a pair of nodes separated by a Euclidean distance $x$
are directly connected with probability $g\left(x\right)$, independent
of the event that another distinct pair of nodes are directly connected.
$\mathcal{X}_{\rho}$ denotes the vertex set in $\mathcal{G}\left(\mathcal{X}_{\rho},g,\Re^{2}\right)$.
\end{definitn}
With minor abuse of the terminology, we use $A$ (respectively $A_{\sqrt{\rho}}$,
$A_{\frac{1}{r_{\rho}}}$) to denote both the square itself and the
area of the square, and in the latter case, $A=1$ (respectively $A_{\sqrt{\rho}}=\rho$,
$A_{\frac{1}{r_{\rho}}}=\frac{1}{r_{\rho}^{2}}$).

The reason for choosing this particular form of $r_{\rho}$ and the
above network models is to avoid triviality in the analysis and to
make the analysis compatible with existing results obtained under
a unit disk connection model. Particularly $ $when $g$ takes the
form that $g(x)=1$ for $x\leq1$ and $g(x)=0$ for $x>1$, it can
be shown that $\mathcal{G}\left(\mathcal{X}_{\rho},g_{r_{\rho}},A\right)$
reduces to the dense network model under a unit disk connection model
discussed in \cite{Gupta98Critical,Penrose03Random,Franceschetti07Random}
where $C=\pi$ and $r_{\rho}$ corresponds to the critical transmission
range for connectivity; $\mathcal{G}\left(\mathcal{X}_{1},g_{\sqrt{\frac{\log\rho+b}{C}}},A_{\sqrt{\rho}}\right)$
reduces to the extended network model under a unit disk connection
model considered in \cite[Chapter 3.3.2]{Franceschetti07Closing,Franceschetti07Random}.
Thus the above model easily incorporates the unit disk connection
model as a special case. A similar conclusion can also be drawn for
the log-normal connection model.

Now we establish the relationship between the three network models
in Definitions \ref{def:(dense-network-model)}, \ref{def:(extended-network-model)},
\ref{def:network model discussed in the paper} on finite and then
asymptotically infinite regions respectively using the scaling and
coupling technique \cite{Meester96Continuum}. Given an instance of
$\mathcal{G}\left(\mathcal{X}_{\rho},g_{r_{\rho}},A\right)$, if we
scale the Euclidean distances between \emph{all pairs} of nodes by
a factor of $\sqrt{\rho}$ while maintaining their connections, there
results a random network where nodes are Poissonly distributed on
a square $A_{\sqrt{\rho}}$ with density $1$ and a pair of nodes
separated by a Euclidean distance $x$ are directly connected with
probability $g_{\sqrt{\frac{\log\rho+b}{\rho}}}\left(x\right)$, i.e.
an instance of $\mathcal{G}\left(\mathcal{X}_{1},g_{\sqrt{\frac{\log\rho+b}{C}}},A_{\sqrt{\rho}}\right)$.
All connectivity properties, e.g. connectivity, number of isolated
nodes, number of components of a specified order, that hold in the
instance of $\mathcal{G}\left(\mathcal{X}_{\rho},g_{r_{\rho}},A\right)$
are also valid for the associated instance in $\mathcal{G}\left(\mathcal{X}_{1},g_{\sqrt{\frac{\log\rho+b}{C}}},A_{\sqrt{\rho}}\right)$
(To be more precise, the \emph{underlying graphs }of these two network
instances are \emph{isomorphic} \cite{Mao07Wireless,Gross04Handbook}).
Similarly if we shrink the Euclidean distances between all pairs of
nodes in a network, which is an instance of $\mathcal{G}\left(\mathcal{X}_{1},g_{\sqrt{\frac{\log\rho+b}{C}}},A_{\sqrt{\rho}}\right)$,
by a factor of $\frac{1}{\sqrt{\rho}}$, there results an instance
of $\mathcal{G}\left(\mathcal{X}_{\rho},g_{r_{\rho}},A\right)$ and
the two networks again have the same connectivity property. Therefore
$\mathcal{G}\left(\mathcal{X}_{\rho},g_{r_{\rho}},A\right)$ and $\mathcal{G}\left(\mathcal{X}_{1},g_{\sqrt{\frac{\log\rho+b}{C}}},A_{\sqrt{\rho}}\right)$
are \emph{equivalent} in that any connectivity property that holds
in one model will necessarily hold in the other. Similarly, it can
also be shown that $\mathcal{G}\left(\mathcal{X}_{\rho},g_{r_{\rho}},A\right)$
and $\mathcal{G}\left(\mathcal{X}_{\frac{\log\rho+b}{C}},g,A_{\frac{1}{r_{\rho}}}\right)$
are equivalent in their connectivity properties. Thus in this paper
we only chose one model, i.e. $\mathcal{G}\left(\mathcal{X}_{\frac{\log\rho+b}{C}},g,A_{\frac{1}{r_{\rho}}}\right)$,
to discuss the connectivity properties of finite and asymptotically
infinite networks. The reason for choosing this network model is that
under the model, a pair of nodes are directly connected following
$g$, in the same way as nodes in the infinite network model $\mathcal{G}\left(\mathcal{X}_{\rho},g,\Re^{2}\right)$
are directly connected. This facilitates the discussion and comparison
between the finite (asymptotically infinite) network model and the
infinite network model, which is a key focus of the paper. 

Further, we point out that the above discussion on the equivalence
of network models $\mathcal{G}\left(\mathcal{X}_{\rho},g_{r_{\rho}},A\right)$,
$\mathcal{G}\left(\mathcal{X}_{1},g_{\sqrt{\frac{\log\rho+b}{C}}},A_{\sqrt{\rho}}\right)$
and $\mathcal{G}\left(\mathcal{X}_{\frac{\log\rho+b}{C}},g,A_{\frac{1}{r_{\rho}}}\right)$
is only valid for the random connection model. For the other widely
used model, i.e. the SINR model, under some special circumstances,
e.g. the background noise is negligible \cite{Haenggi09Stochastic}
and the attenuation function is a power law function, the three network
models are equivalent; otherwise under more general conditions, \emph{the
three models are not equivalent} (see e.g. \cite{Dousse05Impact,Dousse03Impact}).
However the key observation revealed in our analysis, i.e. results
obtained from an infinite network model do not necessarily apply to
the dense and extended network models, also holds for the SINR model.

\section{A Comparative Study of The Expected Number of Isolated Nodes \label{sec:expected number of isolated nodes}}

In this section we comparatively study the expected number of isolated
nodes in $\mathcal{G}\left(\mathcal{X}_{\frac{\log\rho+b}{C}},g,A_{\frac{1}{r_{\rho}}}\right)$
and the expected number of isolated nodes in its counterpart in an
infinite network, i.e. a region with the same area as $A_{\frac{1}{r_{\rho}}}$
in an infinite network on $\Re^{2}$ with the same node density $\frac{\log\rho+b}{C}$
and connection function $g$. The number of isolated nodes is a key
parameter in the analysis of network connectivity. A necessary condition
for a network to be connected is that the network has no isolated
node. Such a necessary condition has been shown to be also a sufficient
condition for a connected network as $\rho\rightarrow\infty$ under
a unit disk connection model \cite{Penrose03Random} and this may
also be possibly true for a random connection model.

\subsection{Expected Number of Isolated Nodes in an Asymptotically Infinite Network}

In this subsection we analyze the expected number of isolated nodes
in $\mathcal{G}\left(\mathcal{X}_{\frac{\log\rho+b}{C}},g,A_{\frac{1}{r_{\rho}}}\right)$.
For an arbitrary node in $\mathcal{G}\left(\mathcal{X}_{\frac{\log\rho+b}{C}},g,A_{\frac{1}{r_{\rho}}}\right)$
at location $\boldsymbol{y}$, it can be shown that the probability
that the node is isolated is given by \cite{Ta09On}:

\begin{equation}
\Pr\left(I_{\boldsymbol{y}}=1\right)=e^{-\int_{A_{\frac{1}{r_{\rho}}}}\frac{\log\rho+b}{C}g\left(\left\Vert \boldsymbol{x}-\boldsymbol{y}\right\Vert \right)d\boldsymbol{x}}\label{eq:probability a random node being isolated finite square}\end{equation}
where $I_{\boldsymbol{y}}$ is an indicator random variable: $I_{\boldsymbol{y}}=1$
if the node at $\boldsymbol{y}$ is isolated and $I_{\boldsymbol{y}}=0$
otherwise. Denote by $W$ the number of isolated nodes in an instance
of $\mathcal{G}\left(\mathcal{X}_{\frac{\log\rho+b}{C}},g,A_{\frac{1}{r_{\rho}}}\right)$.
It then follows that the expected number of isolated nodes in $\mathcal{G}\left(\mathcal{X}_{\frac{\log\rho+b}{C}},g,A_{\frac{1}{r_{\rho}}}\right)$
is given by

\begin{equation}
E\left(W\right)=\int_{A_{\frac{1}{r_{\rho}}}}\frac{\log\rho+b}{C}e^{-\int_{A_{\frac{1}{r_{\rho}}}}\frac{\log\rho+b}{C}g\left(\left\Vert \boldsymbol{x}-\boldsymbol{y}\right\Vert \right)d\boldsymbol{x}}d\boldsymbol{y}\label{eq:expected number of isolated nodes in a square finite}\end{equation}

On the basis of \eqref{eq:expected number of isolated nodes in a square finite},
the following theorem can be obtained.
\begin{thm}
\label{thm:expected isolated nodes asymptotically infinite square}The
expected number of isolated nodes in $\mathcal{G}\left(\mathcal{X}_{\frac{\log\rho+b}{C}},g,A_{\frac{1}{r_{\rho}}}\right)$
is $\int_{A_{\frac{1}{r_{\rho}}}}\frac{\log\rho+b}{C}e^{-\int_{A_{\frac{1}{r_{\rho}}}}\frac{\log\rho+b}{C}g\left(\left\Vert \boldsymbol{x}-\boldsymbol{y}\right\Vert \right)d\boldsymbol{x}}d\boldsymbol{y}$.
For $g$ satisfying both \eqref{eq:conditions on g(x) - non-increasing}
and \eqref{eq:Condition on g(x) requirement 2}, the expected number
of isolated nodes in $\mathcal{G}\left(\mathcal{X}_{\frac{\log\rho+b}{C}},g,A_{\frac{1}{r_{\rho}}}\right)$
converges asymptotically to $e^{-b}$ as $\rho\rightarrow\infty$.\end{thm}
\begin{proof}
See Appendix I
\end{proof}

\subsubsection{Impact of Boundary Effect on the Number of Isolated Nodes}

Before we proceed to the comparison of the expected number of isolated
nodes in $\mathcal{G}\left(\mathcal{X}_{\frac{\log\rho+b}{C}},g,A_{\frac{1}{r_{\rho}}}\right)$
and the expected number in its counterpart in an infinite network,
we first examine the impact of boundary effect on the number of isolated
nodes in $\mathcal{G}\left(\mathcal{X}_{\frac{\log\rho+b}{C}},g,A_{\frac{1}{r_{\rho}}}\right)$.
Boundary effect is a common concern in the analysis of network connectivity.
The analysis of the impact of the boundary effect is done by comparing
the number of isolated nodes in $\mathcal{G}\left(\mathcal{X}_{\frac{\log\rho+b}{C}},g,A_{\frac{1}{r_{\rho}}}\right)$
and the number in a network with nodes Poissonly distributed on a
torus $A_{\frac{1}{r_{\rho}}}^{T}\triangleq\left[-\frac{1}{2r_{\rho}},\frac{1}{2r_{\rho}}\right]^{2}$
with node density $\frac{\log\rho+b}{C}$ and where a pair of nodes
separated by a \emph{toroidal distance} $x^{T}$ \cite[p. 13]{Penrose03Random}
are directly connected with probability $g\left(x^{T}\right)$, independent
of the event that another distinct pair of nodes are directly connected.
Denote the network on a torus by $\mathcal{G}^{T}\left(\mathcal{X}_{\frac{\log\rho+b}{C}},g,A_{\frac{1}{r_{\rho}}}^{T}\right)$.
The following lemma can be established.
\begin{lemma}
\label{lem:Expected Isolated nodes torus}The expected number of isolated
nodes in $\mathcal{G}^{T}\left(\mathcal{X}_{\frac{\log\rho+b}{C}},g,A_{\frac{1}{r_{\rho}}}^{T}\right)$
is $\rho e^{-\int_{A_{\frac{1}{r_{\rho}}}}\frac{\log\rho+b}{C}g\left(\left\Vert \boldsymbol{x}\right\Vert \right)d\boldsymbol{x}}$.
For $g$ satisfying both \eqref{eq:conditions on g(x) - non-increasing}
and \eqref{eq:Condition on g(x) requirement 2}, the expected number
of isolated nodes in $\mathcal{G}^{T}\left(\mathcal{X}_{\frac{\log\rho+b}{C}},g,A_{\frac{1}{r_{\rho}}}\right)$
converges to $e^{-b}$ as $\rho\rightarrow\infty$.\end{lemma}
\begin{proof}
See Appendix II
\end{proof}
On the basis of Theorem \ref{thm:expected isolated nodes asymptotically infinite square}
and Lemma \ref{lem:Expected Isolated nodes torus}, and using the
coupling technique, the following lemma can be obtained.
\begin{lemma}
\label{lem:Isolated nodes due to boundary effect}For $g$ satisfying
both \eqref{eq:conditions on g(x) - non-increasing} and \eqref{eq:Condition on g(x) requirement 2},
the number of isolated nodes in $\mathcal{G}\left(\mathcal{X}_{\frac{\log\rho+b}{C}},g,A_{\frac{1}{r_{\rho}}}\right)$
due to the boundary effect is a.a.s. $0$ as $\rho\rightarrow\infty$.\end{lemma}
\begin{proof}
Comparing Theorem \ref{thm:expected isolated nodes asymptotically infinite square}
and Lemma \ref{lem:Expected Isolated nodes torus}, it is noted that
the expected numbers of isolated nodes on a torus and on a square
respectively asymptotically converge to the same \emph{non-zero finite
constant} $e^{-b}$ as $\rho\rightarrow\infty$. Now we use the coupling
technique \cite{Meester96Continuum} to construct the connection between
$W$ and $W^{T}$, the number of isolated nodes in the corresponding
instance of $\mathcal{G}^{T}\left(\mathcal{X}_{\frac{\log\rho+b}{C}},g,A_{\frac{1}{r_{\rho}}}\right)$.
Consider an instance of $\mathcal{G}^{T}\left(\mathcal{X}_{\frac{\log\rho+b}{C}},g,A_{\frac{1}{r_{\rho}}}^{T}\right)$.
The number of isolated nodes in that network is $W^{T}$, which depends
on $\rho$. Remove each connection of the above network with probability
$1-\frac{g\left(x\right)}{g\left(x^{T}\right)}$, independent of the
event that another connection is removed, where $x$ is the Euclidean
distance between the two endpoints of the connection and $x^{T}$
is the corresponding toroidal distance. Due to $x^{T}\leq x$ (see
\eqref{eq:property of toroidal distance 1} in Appendix II) and the
non-increasing property of $g$, $0\leq1-\frac{g\left(x\right)}{g\left(x^{T}\right)}\leq1$.
Further note that only connections between nodes near the boundary
with $x^{T}<x$ will be affected, i.e. when $x=x^{T}$ the removal
probability is zero. Denote the number of \emph{newly }appearing isolated
nodes by $W^{E}$. $W^{E}$ has the meaning of being \emph{the number
of isolated nodes due to the boundary effect}. It is straightforward
to show that $W^{E}$ is a \emph{non-negative random integer}, depending
on $\rho$. Further, such a connection removal process results a random
network with nodes Poissonly distributed with density $\frac{\log\rho+b}{C}$
where a pair of nodes separated by a \emph{Euclidean} distance $x$
are directly connected with probability $g\left(x\right)$, i.e. a
random network on a square with the boundary effect included. The
following equation results as a consequence of the above discussion:
\[
W=W^{E}+W^{T}\]
Using Theorem \ref{thm:expected isolated nodes asymptotically infinite square},
Lemma \ref{lem:Expected Isolated nodes torus} and the above equation,
it can be shown that \[
\lim_{\rho\rightarrow\infty}E\left(W^{E}\right)=\lim_{\rho\rightarrow\infty}E\left(W-W^{T}\right)=0\]
Due to the non-negativity of $W^{E}$: \[
\lim_{\rho\rightarrow\infty}\Pr\left(W^{E}=0\right)=1\]
\end{proof}
\begin{remrk}
Note that for $g$ not satisfying \eqref{eq:Condition on g(x) requirement 2},
$E\left(W\right)$ and $E\left(W^{T}\right)$ are not necessarily
convergent as $\rho\rightarrow\infty$. Particularly using the same
procedure in Appendix I and II (see also \eqref{eq:necessity of second requirement on g}
in Section \ref{sub:Comparison isolate nodes with infinite network}
below), it can be shown that when \textbf{$g\left(x\right)=\omega_{x}\left(\frac{1}{x^{2}\log^{2}x}\right)$,
}both $\lim_{\rho\rightarrow\infty}E\left(W\right)$ and $\lim_{\rho\rightarrow\infty}E\left(W^{T}\right)$
are unbounded. When \textbf{$g\left(x\right)=\Theta_{x}\left(\frac{1}{x^{2}\log^{2}x}\right)$,
}$\lim_{\rho\rightarrow\infty}E\left(W\right)$ and $\lim_{\rho\rightarrow\infty}E\left(W^{T}\right)$
start to depend on the asymptotic behavior of $g$ and is only convergent
when $\lim_{x\rightarrow\infty}g\left(x\right)x^{2}\log^{2}x=a$,
where $0<a<\infty$ is a positive constant. In that case, it can be
shown that $\lim_{\rho\rightarrow\infty}E\left(W\right)$ and $\lim_{\rho\rightarrow\infty}E\left(W^{T}\right)$
converge to $e^{-b+\frac{4\pi}{C}a}$. For $\lim_{\rho\rightarrow\infty}E\left(W^{T}\right)$
the above result can be established by first choosing a small positive
constant $\triangle\varepsilon$ and then letting $\rho$ be sufficiently
large such that $D\left(\boldsymbol{0},\frac{1}{2}r_{\rho}^{-1-\triangle\varepsilon}\right)$
contains $A_{\frac{1}{r_{\rho}}}$, where $D\left(\boldsymbol{x},r\right)$
denotes a disk centered at $\boldsymbol{x}$ and with a radius $r$.
An upper and lower bound on $E\left(W^{T}\right)$ can then be established
by noting that \begin{eqnarray*}
 &  & \lim_{\rho\rightarrow\infty}\rho e^{-\int_{D\left(\boldsymbol{0},\frac{1}{2}r_{\rho}^{-1-\triangle\varepsilon}\right)}\frac{\log\rho+b}{C}g\left(\left\Vert \boldsymbol{x}\right\Vert \right)d\boldsymbol{x}}\\
 & \leq & \lim_{\rho\rightarrow\infty}E\left(W^{T}\right)=\rho e^{-\int_{A_{\frac{1}{r_{\rho}}}}\frac{\log\rho+b}{C}g\left(\left\Vert \boldsymbol{x}\right\Vert \right)d\boldsymbol{x}}\\
 & \leq & \lim_{\rho\rightarrow\infty}\rho e^{-\int_{D\left(\boldsymbol{0},\frac{1}{2}r_{\rho}^{-1}\right)}\frac{\log\rho+b}{C}g\left(\left\Vert \boldsymbol{x}\right\Vert \right)d\boldsymbol{x}}\end{eqnarray*}
Following the exactly same procedure as that in \eqref{eq:an analysis of the trunction effect}
and \eqref{eq:little's rule trunction} (in Appendix II) and finally
letting $ $$\triangle\varepsilon\rightarrow0$, the result for \textbf{$\lim_{\rho\rightarrow\infty}E\left(W^{T}\right)$}
can be obtained. The result for $\lim_{\rho\rightarrow\infty}E\left(W\right)$
can be obtained following a similar procedure as that in Appendix
I.
\end{remrk}

\subsection{The Number of Isolated Nodes in a Region $A_{\frac{1}{r_{\rho}}}$
of an Infinite Network with Node Density $\frac{\log\rho+b}{C}$}

In this subsection, we consider the number of isolated nodes in the
counterpart of $\mathcal{G}\left(\mathcal{X}_{\frac{\log\rho+b}{C}},g,A_{\frac{1}{r_{\rho}}}\right)$
in an infinite network. Specifically, for a meaningful comparison
with the number of isolated nodes in $\mathcal{G}\left(\mathcal{X}_{\frac{\log\rho+b}{C}},g,A_{\frac{1}{r_{\rho}}}\right)$,
we consider the number of isolated nodes, denoted by $W^{\infty}$
(with superscript $^{\infty}$ marking the parameter in an infinite
network), in a square $A_{\frac{1}{r_{\rho}}}$ of an infinite network
on $\Re^{2}$ with Poissonly distributed node at density $\frac{\log\rho+b}{C}$.
Denote the infinite network by $\mathcal{G}\left(\mathcal{X}_{\frac{\log\rho+b}{C}},g,\Re^{2}\right)$.
For $g$ satisfying \eqref{eq:conditions on g(x) - integral boundness},
a randomly chosen node in $\mathcal{G}\left(\mathcal{X}_{\frac{\log\rho+b}{C}},g,\Re^{2}\right)$,
at location $\boldsymbol{y}\in A_{\frac{1}{r_{\rho}}}$, is isolated
with probability\begin{equation}
\Pr\left(I_{\boldsymbol{y}}^{\infty}=1\right)=e^{-\int_{\Re^{2}}\frac{\log\rho+b}{C}g\left(\left\Vert \boldsymbol{x}-\boldsymbol{y}\right\Vert \right)d\boldsymbol{x}}=\frac{1}{\rho}e^{-b}\label{eq:probability of node being isolated - infinite}\end{equation}
 where \eqref{eq:conditions on g(x) - integral boundness} is used
in the above equation. Therefore \begin{eqnarray}
E\left(W^{\infty}\right) & = & \int_{A_{\frac{1}{r_{\rho}}}}\frac{\log\rho+b}{C}\times\frac{1}{\rho}e^{-b}d\boldsymbol{y}\nonumber \\
 & = & \frac{\log\rho+b}{C}\times\frac{1}{\rho}e^{-b}\times\left(\frac{1}{r_{\rho}}\right)^{2}\nonumber \\
 & = & e^{-b}\label{eq:expected number of isolated nodes in asymptotically infinite region of infinite network}\end{eqnarray}

The last line follows by \eqref{eq:definition of r_rho}.

The above result is summarized in the following lemma:
\begin{lemma}
\label{lem:expected isolated nodes in an infinite counterpart}For
$g$ satisfying \eqref{eq:conditions on g(x) - integral boundness},
the expected number of isolated nodes in a region $A_{\frac{1}{r_{\rho}}}$
of $\mathcal{G}\left(\mathcal{X}_{\frac{\log\rho+b}{C}},g,\Re^{2}\right)$
is $e^{-b}$.

\end{lemma}

\subsection{A Comparison of the Expected Number of Isolated Nodes in $\mathcal{G}\left(\mathcal{X}_{\frac{\log\rho+b}{C}},g,A_{\frac{1}{r_{\rho}}}\right)$
and In Its Counterpart in An Infinite Network\label{sub:Comparison isolate nodes with infinite network}}

Comparing Theorem \ref{thm:expected isolated nodes asymptotically infinite square}
and Lemma \ref{lem:expected isolated nodes in an infinite counterpart},
we note that:
\begin{enumerate}
\item The expected number of isolated nodes in $\mathcal{G}\left(\mathcal{X}_{\frac{\log\rho+b}{C}},g,A_{\frac{1}{r_{\rho}}}\right)$
only \emph{converges asymptotically} to $e^{-b}$ as $\rho\rightarrow\infty$
whereas the expected number of isolated nodes in an area of the same
size in $\mathcal{G}\left(\mathcal{X}_{\frac{\log\rho+b}{C}},g,\Re^{2}\right)$
\emph{is always} $e^{-b}$ no matter which value $\rho$ takes.
\item The expected number of isolated nodes in $\mathcal{G}\left(\mathcal{X}_{\frac{\log\rho+b}{C}},g,A_{\frac{1}{r_{\rho}}}\right)$
converges asymptotically to $e^{-b}$ for $g$ \emph{satisfying both
\eqref{eq:conditions on g(x) - non-increasing} and \eqref{eq:Condition on g(x) requirement 2}}
whereas the expected number of isolated nodes in an area of the same
size in $\mathcal{G}\left(\mathcal{X}_{\frac{\log\rho+b}{C}},g,\Re^{2}\right)$
is $e^{-b}$ for $g$ \emph{satisfying \eqref{eq:conditions on g(x) - integral boundness}
only}.
\end{enumerate}
In the following we examine the reason behind the differences. 

Using \eqref{eq:probability a random node being isolated finite square},
\eqref{eq:expected number of isolated nodes in a square finite},
\eqref{eq:probability of node being isolated - infinite} and \eqref{eq:expected number of isolated nodes in asymptotically infinite region of infinite network},
it can be shown that

\begin{align}
 & \frac{E\left(W\right)}{E\left(W^{\infty}\right)}\nonumber \\
= & e^{b}\int_{A_{\frac{1}{r_{\rho}}}}\frac{\log\rho+b}{C}e^{-\int_{A_{\frac{1}{r_{\rho}}}}\frac{\log\rho+b}{C}g\left(\left\Vert \boldsymbol{x}-\boldsymbol{y}\right\Vert \right)d\boldsymbol{x}}d\boldsymbol{y}\nonumber \\
= & e^{b}\int_{A_{\frac{1}{r_{\rho}}}}\frac{\log\rho+b}{C}\exp\left(-\int_{\Re^{2}}\frac{\log\rho+b}{C}g\left(\left\Vert \boldsymbol{x}-\boldsymbol{y}\right\Vert \right)d\boldsymbol{x}\right.\nonumber \\
\times & \left.\int_{\Re^{2}\backslash A_{\frac{1}{r_{\rho}}}}\frac{\log\rho+b}{C}g\left(\left\Vert \boldsymbol{x}-\boldsymbol{y}\right\Vert \right)d\boldsymbol{x}\right)d\boldsymbol{y}\nonumber \\
= & \int_{A_{\frac{1}{r_{\rho}}}}\frac{\log\rho+b}{C\rho}e^{\int_{\Re^{2}\backslash A_{\frac{1}{r_{\rho}}}}\frac{\log\rho+b}{C}g\left(\left\Vert \boldsymbol{x}-\boldsymbol{y}\right\Vert \right)d\boldsymbol{x}}d\boldsymbol{y}\label{eq:ratio of expected numbers of isolated nodes}\end{align}
It is trivial to show that the value in \eqref{eq:ratio of expected numbers of isolated nodes}
is always greater than $1$ for $g$ with infinite support. That is,
for any $g$ with infinite support, the expected number of isolated
nodes in $\mathcal{G}\left(\mathcal{X}_{\frac{\log\rho+b}{C}},g,A_{\frac{1}{r_{\rho}}}\right)$
is strictly larger than the expected number of isolated nodes in an
area $A_{\frac{1}{r_{\rho}}}$ of $\mathcal{G}\left(\mathcal{X}_{\frac{\log\rho+b}{C}},g,\Re^{2}\right)$.
Further, it can be shown that the value in \eqref{eq:ratio of expected numbers of isolated nodes}
accounts for the cumulative effect of nodes outside $A_{\frac{1}{r_{\rho}}}$
in $\mathcal{G}\left(\mathcal{X}_{\frac{\log\rho+b}{C}},g,\Re^{2}\right)$
and the associated connections between these nodes and nodes inside
$A_{\frac{1}{r_{\rho}}}$ on decreasing the expected number of isolated
nodes in $A_{\frac{1}{r_{\rho}}}$ respectively. Because $\mathcal{G}\left(\mathcal{X}_{\frac{\log\rho+b}{C}},g,A_{\frac{1}{r_{\rho}}}\right)$
can be obtained from $\mathcal{G}\left(\mathcal{X}_{\frac{\log\rho+b}{C}},g,\Re^{2}\right)$
by removing all these nodes and associated connections outside an
area of $A_{\frac{1}{r_{\rho}}}$ in $\mathcal{G}\left(\mathcal{X}_{\frac{\log\rho+b}{C}},g,\Re^{2}\right)$,
we term the this distinction \emph{the truncation effect}. Theorem
\ref{thm:expected isolated nodes asymptotically infinite square}
and Lemma \ref{lem:expected isolated nodes in an infinite counterpart}
shows that when $g$ satisfies both \eqref{eq:conditions on g(x) - non-increasing}
and \eqref{eq:Condition on g(x) requirement 2} (i.e. $g$ has to
decrease fast enough), the impact of the truncation effect on isolated
nodes becomes vanishingly small as $\rho\rightarrow\infty$.

Based on the above discussion, the following theorem can be established:
\begin{thm}
\label{thm:conclusion on the impact of truncation effect sufficient condition}For
$g$ satisfying \eqref{eq:conditions on g(x) - integral boundness},
the expected number of isolated nodes in an area of $A_{\frac{1}{r_{\rho}}}$
in $\mathcal{G}\left(\mathcal{X}_{\frac{\log\rho+b}{C}},g,\Re^{2}\right)$
is $e^{-b}$. Removing all nodes of $\mathcal{G}\left(\mathcal{X}_{\frac{\log\rho+b}{C}},g,\Re^{2}\right)$
outside $A_{\frac{1}{r_{\rho}}}$ and the associated connections,
there results $\mathcal{G}\left(\mathcal{X}_{\frac{\log\rho+b}{C}},g,A_{\frac{1}{r_{\rho}}}\right)$.
The expected number of isolated nodes in $\mathcal{G}\left(\mathcal{X}_{\frac{\log\rho+b}{C}},g,A_{\frac{1}{r_{\rho}}}\right)$
converges to $e^{-b}$ if $g$ satisfies both \eqref{eq:conditions on g(x) - non-increasing}
and \eqref{eq:Condition on g(x) requirement 2}. The more restrictive
requirement on $g$ is a sufficient condition for the impact of the
\emph{truncation effect} associated with the above removal operations
on the number of isolated nodes in $\mathcal{G}\left(\mathcal{X}_{\frac{\log\rho+b}{C}},g,A_{\frac{1}{r_{\rho}}}\right)$
to be vanishingly small as $\rho\rightarrow\infty$. 
\end{thm}
In the following, we show that the more restrictive requirement on
$g$ in \eqref{eq:Condition on g(x) requirement 2} (compared with
\eqref{eq:conditions on g(x) - non-increasing} and \eqref{eq:conditions on g(x) - integral boundness})
is also necessary for the impact of the truncation effect to become
vanishingly small as $\rho\rightarrow\infty$. Specifically, consider
the case when \eqref{eq:Condition on g(x) requirement 2} is not satisfied.
Let \begin{eqnarray}
f\left(x\right) & \triangleq & g\left(x\right)x^{2}\log^{2}x\label{eq:definition of f(x)}\end{eqnarray}
Condition \eqref{eq:Condition on g(x) requirement 2} not being satisfied
means

\begin{equation}
\lim_{x\rightarrow\infty}f\left(x\right)\neq0\label{eq:a contradictory case for f(x)}\end{equation}
i.e. $\lim_{x\rightarrow\infty}f\left(x\right)$ may equal to a positive
constant, $\infty$, or does not exist (e.g. $f\left(x\right)$ is
a periodic function of $x$). 

It can be shown that (following the equation, detailed explanations
are given and see also \eqref{eq:expected number of isolated nodes in a torus finite}
in Appendix II)

\begin{eqnarray}
 &  & \lim_{\rho\rightarrow\infty}E\left(W\right)\nonumber \\
 & \geq & \lim_{\rho\rightarrow\infty}E\left(W^{T}\right)\nonumber \\
 & = & \lim_{\rho\rightarrow\infty}\rho e^{-\int_{A_{\frac{1}{r_{\rho}}}}\frac{\log\rho+b}{C}g\left(\left\Vert \boldsymbol{x}\right\Vert \right)d\boldsymbol{x}}\nonumber \\
 & \geq & \lim_{\rho\rightarrow\infty}\rho e^{-\int_{D\left(\boldsymbol{0},\frac{1}{2}r_{\rho}^{-1}\right)}\frac{\log\rho+b}{C}g\left(\left\Vert \boldsymbol{x}\right\Vert \right)d\boldsymbol{x}}\nonumber \\
 & = & e^{-b}\lim_{\rho\rightarrow\infty}e^{\int_{\Re^{2}\backslash D\left(\boldsymbol{0},\frac{1}{2}r_{\rho}^{-1}\right)}\frac{\log\rho+b}{C}g\left(\left\Vert \boldsymbol{x}\right\Vert \right)d\boldsymbol{x}}\nonumber \\
 & = & e^{-b+\frac{4\pi}{C}\lim_{x\rightarrow\infty}f\left(x\right)}\label{eq:necessity of second requirement on g}\end{eqnarray}
where the last step results because of the following equation: \begin{eqnarray*}
 &  & \int_{\Re^{2}\backslash D\left(\boldsymbol{0},\frac{1}{2}r_{\rho}^{-1}\right)}\frac{\log\rho+b}{C}g\left(\left\Vert \boldsymbol{x}\right\Vert \right)d\boldsymbol{x}\\
 & = & \lim_{\rho\rightarrow\infty}\int_{\frac{1}{2}r_{\rho}^{-1}}^{\infty}\frac{\log\rho+b}{C}2\pi xg\left(x\right)dx\\
 & = & \lim_{\rho\rightarrow\infty}\frac{\frac{\pi}{2}r_{\rho}^{-4}g\left(\frac{1}{2}r_{\rho}^{-1}\right)\frac{\log\rho+b-1}{C\rho^{2}}}{\frac{C}{\rho\left(\log\rho+b\right)^{2}}}\\
 & = & \lim_{\rho\rightarrow\infty}\frac{\pi}{2C}\left(\log\rho+b\right)^{2}r_{\rho}^{-2}g\left(\frac{1}{2}r_{\rho}^{-1}\right)\\
 & = & \lim_{\rho\rightarrow\infty}\frac{\pi}{2C}\left(\log\rho+b\right)^{2}r_{\rho}^{-2}\frac{f\left(\frac{1}{2}r_{\rho}^{-1}\right)}{\frac{1}{4}r_{\rho}^{-2}\log^{2}\left(\frac{1}{2}r_{\rho}^{-1}\right)}\\
 & = & \lim_{\rho\rightarrow\infty}\frac{2\pi\left(\log\rho+b\right)^{2}f\left(\frac{1}{2}r_{\rho}^{-1}\right)}{C\left(\log\frac{1}{2}-\frac{1}{2}\log\left(\log\rho+b\right)+\frac{1}{2}\log\rho+\frac{1}{2}\log C\right)^{2}}\\
 & = & \frac{4\pi}{C}\lim_{\rho\rightarrow\infty}f\left(\frac{1}{2}r_{\rho}^{-1}\right)\\
 & = & \frac{4\pi}{C}\lim_{x\rightarrow\infty}f\left(x\right)\end{eqnarray*}
where in the second step, L'Hôpital's rule with $\frac{C}{\log\rho+b}$
being the denominator and $\int_{\frac{1}{2}r_{\rho}^{-1}}^{\infty}2\pi xg\left(x\right)dx$
being the numerator is used; in the third step, \eqref{eq:definition of f(x)}
is used. 
\begin{remrk}
Equation \eqref{eq:necessity of second requirement on g} shows also
that $\lim_{\rho\rightarrow\infty}E\left(W^{T}\right)\geq e^{-b+\frac{4\pi}{C}\lim_{x\rightarrow\infty}f\left(x\right)}$
where $E\left(W^{T}\right)$ is the expected number of isolated nodes
on a torus, which does not include the contribution of the boundary
effect on the number of isolated nodes. Note also that\textbf{ }the
expected number of isolated nodes in an area of $A_{\frac{1}{r_{\rho}}}$
in $\mathcal{G}\left(\mathcal{X}_{\frac{\log\rho+b}{C}},g,\Re^{2}\right)$
is $e^{-b}$. Therefore the term $e^{\frac{4\pi}{C}\lim_{x\rightarrow\infty}f\left(x\right)}$
is entirely attributable to the truncation effect.
\end{remrk}
Note that $f\left(x\right)$ is a non-negative function for $x>1$.
It is obvious from \eqref{eq:necessity of second requirement on g}
that \emph{unless} $\lim_{x\rightarrow\infty}f\left(x\right)=0$,
i.e. \eqref{eq:Condition on g(x) requirement 2} is satisfied, the
expected number of isolated node in $\mathcal{G}\left(\mathcal{X}_{\frac{\log\rho+b}{C}},g,A_{\frac{1}{r_{\rho}}}\right)$
will be larger than the expected number of isolated nodes in an area
of $A_{\frac{1}{r_{\rho}}}$ in $\mathcal{G}\left(\mathcal{X}_{\frac{\log\rho+b}{C}},g,\Re^{2}\right)$.
That is, the\textbf{ }impact of the \emph{truncation effect} on the
number of isolated nodes in $\mathcal{G}\left(\mathcal{X}_{\frac{\log\rho+b}{C}},g,A_{\frac{1}{r_{\rho}}}\right)$
will \emph{not be} vanishingly small as $\rho\rightarrow\infty$.
In particular, it can be shown that for $g\left(x\right)=\Theta_{x}\left(\frac{1}{x^{2}\log^{2}x}\right)$,
the impact of the truncation effect is non-negligible or even dominant
in determining the number of isolated nodes in $\mathcal{G}\left(\mathcal{X}_{\frac{\log\rho+b}{C}},g,A_{\frac{1}{r_{\rho}}}\right)$.
Using \eqref{eq:necessity of second requirement on g}, it can also
be shown that for $g\left(x\right)=\omega_{x}\left(\frac{1}{x^{2}\log^{2}x}\right)$,
$\lim_{\rho\rightarrow\infty}E\left(W\right)$ is unbounded, i.e.
connectivity cannot be achieved for $g\left(x\right)=\omega_{x}\left(\frac{1}{x^{2}\log^{2}x}\right)$
even if \eqref{eq:conditions on g(x) - non-increasing} and \eqref{eq:conditions on g(x) - integral boundness}
are satisfied.

The above discussion leads to the following conclusion:
\begin{thm}
\label{thm:conclusion on the impact of truncation effect necessary condition}The
more restrictive requirement on $g$ that it satisfies\textbf{ }\eqref{eq:Condition on g(x) requirement 2}
is a necessary condition for the impact of the \emph{truncation effect}
on the number of isolated nodes in $\mathcal{G}\left(\mathcal{X}_{\frac{\log\rho+b}{C}},g,A_{\frac{1}{r_{\rho}}}\right)$
to be vanishingly small as $\rho\rightarrow\infty$. Further for $g\left(x\right)=\Theta_{x}\left(\frac{1}{x^{2}\log^{2}x}\right)$,
the impact of the truncation effect is non-negligible or even dominant
in determining the number of isolated nodes in $\mathcal{G}\left(\mathcal{X}_{\frac{\log\rho+b}{C}},g,A_{\frac{1}{r_{\rho}}}\right)$;
and for $g\left(x\right)=\omega_{x}\left(\frac{1}{x^{2}\log^{2}x}\right)$,
the truncation effect is the dominant factor in determining the number
of isolated nodes in $\mathcal{G}\left(\mathcal{X}_{\frac{\log\rho+b}{C}},g,A_{\frac{1}{r_{\rho}}}\right)$.
\end{thm}
Noting that the number of isolated nodes in a network is a non-negative
integer, the following result can be obtained as an easy consequence
of Theorem \ref{thm:conclusion on the impact of truncation effect sufficient condition}
(see also \cite{Mao11On}). Notice that in formulating this result,
we drop the assumption that $b$, originally introduced in \eqref{eq:definition of r_rho},
is a constant, and allow it instead to be $\rho$-dependent. 
\begin{cor}
\label{cor:necessary condition for connectivity}For $g$ satisfying
both \eqref{eq:conditions on g(x) - non-increasing} and \eqref{eq:Condition on g(x) requirement 2},
a necessary condition for $\mathcal{G}\left(\mathcal{X}_{\frac{\log\rho+b}{C}},g,A_{\frac{1}{r_{\rho}}}\right)$
to be \emph{a.a.s.} (as $\rho\rightarrow\infty$) connected is $b\rightarrow\infty$. \end{cor}
\begin{remrk}
As pointed out in \cite[p. 151]{Meester96Continuum}, the three requirements
on $g$ in the random connection model, i.e. rotational invariance,
non-increasing monotonicity and integral boundedness, are not equally
important. Particularly, rotational invariance and non-increasing
monotonicity are required only to simply the analysis such that {}``the
notation and formulae will be somewhat simpler''. Similarly, we expect
the results obtained in this section and in the next section requiring
non-increasing monotonicity in \eqref{eq:conditions on g(x) - non-increasing}
are also valid when the condition in \eqref{eq:conditions on g(x) - non-increasing}
is removed. These however require more complicated handling of $g\left(x\right)$,
particularly when $x$ is sufficiently large.
\end{remrk}

\section{Vanishing of Components of Finite Order\label{sec:Components of finite size}}

In this section we consider the events of the asymptotic vanishing
of components of fixed and finite order $k>1$ in the infinite network
$\mathcal{G}\left(\mathcal{X}_{\frac{\log\rho+b}{C}},g,\Re^{2}\right)$
and in $\mathcal{G}\left(\mathcal{X}_{\frac{\log\rho+b}{C}},g,A_{\frac{1}{r_{\rho}}}\right)$
respectively as $\rho\rightarrow\infty$. 

In \cite[Theorem 6.4]{Meester96Continuum} it was shown that as $\rho\rightarrow\infty$
(and $\frac{\log\rho+b}{C}\rightarrow\infty$) the probability for
a node to be isolated given that its component is finite converges
to $1$. In other words, as $\rho\rightarrow\infty$ a.a.s. $\mathcal{G}\left(\mathcal{X}_{\frac{\log\rho+b}{C}},g,\Re^{2}\right)$
has only isolated nodes and components of infinite order, and components
of fixed and finite order $k>1$ asymptotically vanish. In the following
we show that due to the truncation effect, the above result obtained
in $\mathcal{G}\left(\mathcal{X}_{\frac{\log\rho+b}{C}},g,\Re^{2}\right)$
does \emph{not} carry over to the conclusion that as $\rho\rightarrow\infty$
a.a.s. $\mathcal{G}\left(\mathcal{X}_{\frac{\log\rho+b}{C}},g,A_{\frac{1}{r_{\rho}}}\right)$
has only isolated nodes and infinite components too, without further
analysis on the impact of the truncation effect\textbf{.} Specifically,
an infinite component in $\mathcal{G}\left(\mathcal{X}_{\frac{\log\rho+b}{C}},g,\Re^{2}\right)$
may \emph{possibly} consist of components of extremely large order,
components of fixed and finite order $k>1$ and isolated nodes involving
nodes and connections entirely contained inside $A_{\frac{1}{r_{\rho}}}$,
where these components are only connected to each other via nodes
and connections outside $A_{\frac{1}{r_{\rho}}}$. Note that for any
finite $\rho$, almost surely there is no infinite component in $\mathcal{G}\left(\mathcal{X}_{\frac{\log\rho+b}{C}},g,A_{\frac{1}{r_{\rho}}}\right)$.
Therefore we use the term \emph{component of extremely large order}
to refer to a component whose order may become asymptotically infinite
as $\rho\rightarrow\infty$. As the nodes and associated connections
outside $A_{\frac{1}{r_{\rho}}}$ are removed, the infinite component
in $\Re^{2}$ may \emph{possibly} leave components of extremely large
order, components of finite order $k>1$ and isolated nodes in $A_{\frac{1}{r_{\rho}}}$.
As such, vanishing of components of finite order $k>1$ in $\mathcal{G}\left(\mathcal{X}_{\frac{\log\rho+b}{C}},g,\Re^{2}\right)$
as $\rho\rightarrow\infty$ does not \emph{necessarily} carry the
conclusion that components of finite order $k>1$ in $\mathcal{G}\left(\mathcal{X}_{\frac{\log\rho+b}{C}},g,A_{\frac{1}{r_{\rho}}}\right)$
also vanish as $\rho\rightarrow\infty$, even when $A_{\frac{1}{r_{\rho}}}$
approaches $\Re^{2}$ as $\rho\rightarrow\infty$. An example is illustrated
in Fig. \ref{fig:An-illustration-of-non-varnish-finite-component}. 

\begin{figure}
\begin{centering}
\includegraphics[width=0.3\columnwidth]{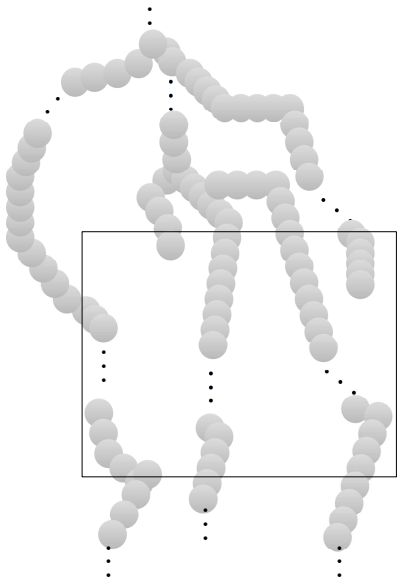}
\par\end{centering}

\caption{An illustration that an infinite component in $\Re^{2}$ may leave
components of extremely large order, components of finite order $k>1$
and isolated nodes in a finite (or asymptotically infinite) region
in $\Re^{2}$ when nodes and connections outside the finite (asymptotically
infinite) region is removed. The figure uses the unit disk connection
model as a special case for easy illustration. Each ball has a radius
of half of the transmission range and is centered at a node. Two adjacent
balls overlap iff the associated nodes are directly connected. The
figure shows an infinite component with nodes organized in a tree
structure. The square area represents the finite (asymptotically infinite)
region. Even as the square grows to include more and more nodes of
the infinite component, it is still possible for the square to have
components of finite order $k>1$ when nodes and connections outside
the square are removed.\label{fig:An-illustration-of-non-varnish-finite-component}}

\end{figure}

We further point out that many other topologies, particularly under
a random connection model where even a pair of nodes separated by
a large distance may have a non-zero probability to be directly connected,
can be drawn for an infinite component in $ $$\Re^{2}$, where after
removing all nodes and associated connections of the infinite component
outside $A_{\frac{1}{r_{\rho}}}$, the infinite component leaves components
of finite order $k>1$ inside $A_{\frac{1}{r_{\rho}}}$, even when
$A_{\frac{1}{r_{\rho}}}$ grows as $\rho\rightarrow\infty$. We emphasize
that we are not hinting that the topology of the infinite component
shown in Fig. \ref{fig:An-illustration-of-non-varnish-finite-component}
is likely to occur in $\mathcal{G}\left(\mathcal{X}_{\frac{\log\rho+b}{C}},g,\Re^{2}\right)$
as $\rho\rightarrow\infty$, but neither can such a possibility be
precluded using \cite[Theorem 6.4]{Meester96Continuum}. Therefore
a conclusion cannot be drawn straightforwardly from \cite[Theorem 6.4]{Meester96Continuum}
that a.a.s. components of finite order $k>1$ in $\mathcal{G}\left(\mathcal{X}_{\frac{\log\rho+b}{C}},g,A_{\frac{1}{r_{\rho}}}\right)$
vanish as $\rho\rightarrow\infty$. Instead some non-trivial analysis
is required to establish such a conclusion in $\mathcal{G}\left(\mathcal{X}_{\frac{\log\rho+b}{C}},g,A_{\frac{1}{r_{\rho}}}\right)$. 

We present such a result for the vanishing of components of finite
order $k>1$ in $\mathcal{G}\left(\mathcal{X}_{\frac{\log\rho+b}{C}},g,A_{\frac{1}{r_{\rho}}}\right)$
as $\rho\rightarrow\infty$ to fill this theoretical gap:
\begin{thm}
\label{thm:varnish of finite components}For $g$ satisfying \eqref{eq:conditions on g(x) - non-increasing}
and \eqref{eq:Condition on g(x) requirement 2}, a.a.s. there is no
component of finite order $k>1$ in $\mathcal{G}\left(\mathcal{X}_{\frac{\log\rho+b}{C}},g,A_{\frac{1}{r_{\rho}}}\right)$. \end{thm}
\begin{proof}
See Appendix III.\end{proof}
\begin{remrk}
Theorem \ref{thm:varnish of finite components} gives a sufficient
condition on $g$ required for the number of components of fixed and
finite order $k>1$ in $\mathcal{G}\left(\mathcal{X}_{\frac{\log\rho+b}{C}},g,A_{\frac{1}{r_{\rho}}}\right)$
to be vanishingly small as $\rho\rightarrow\infty$. It is also interesting
to obtain a necessary condition on $g$ required for the number of
components of fixed and finite order $k>1$ in $\mathcal{G}\left(\mathcal{X}_{\frac{\log\rho+b}{C}},g,A_{\frac{1}{r_{\rho}}}\right)$
to be vanishingly small. The technique used in the proof of Theorem
\ref{thm:varnish of finite components} however cannot answer the
above question on a necessary condition on $g$. More specifically,
denote by $\xi_{k}$ the (random) number of components of order $k$
in an instance of $\mathcal{G}\left(\mathcal{X}_{\lambda},g,A_{\frac{1}{r_{\rho}}}\right)$
and let $M$ be an arbitrarily large positive integer $M$. The proof
of Theorem \ref{thm:varnish of finite components} is based on an
analysis of $E\left(\sum_{k=2}^{M}\xi_{k}\right)$. By showing that
$\lim_{\rho\rightarrow\infty}E\left(\sum_{k=2}^{M}\xi_{k}\right)=0$,
it follows that $\lim_{\rho\rightarrow\infty}\Pr\left(\sum_{k=2}^{M}\xi_{k}=0\right)=1$.
However $\lim_{\rho\rightarrow\infty}E\left(\sum_{k=2}^{M}\xi_{k}\right)=0$
is only a sufficient condition for $\lim_{\rho\rightarrow\infty}\Pr\left(\sum_{k=2}^{M}\xi_{k}=0\right)=1$,
\emph{not} a necessary condition. It would be interesting to develop
a technique to obtain a tight necessary condition on $g$ required
for the number of components of fixed and finite order $k>1$ in $\mathcal{G}\left(\mathcal{X}_{\frac{\log\rho+b}{C}},g,A_{\frac{1}{r_{\rho}}}\right)$
to be vanishingly small.
\end{remrk}

\section{Conclusion\label{sec:Conclusion}}

In this paper, we discussed the connectivity of several network models
including the widely used dense network model $\mathcal{G}\left(\mathcal{X}_{\rho},g_{r_{\rho}},A\right)$,
extended network model $\mathcal{G}\left(\mathcal{X}_{1},g_{\sqrt{\frac{\log\rho+b}{C}}},A_{\sqrt{\rho}}\right)$
and infinite network model $\mathcal{G}\left(\mathcal{X}_{\rho},g,\Re^{2}\right)$.
Using the scaling and coupling technique, it is shown that the dense
network model and the extended network model are equivalent in their
connectivity properties and they are also equivalent to the network
model $\mathcal{G}\left(\mathcal{X}_{\frac{\log\rho+b}{C}},g,A_{\frac{1}{r_{\rho}}}\right)$,
which can be obtained from the infinite network model $\mathcal{G}\left(\mathcal{X}_{\frac{\log\rho+b}{C}},g,\Re^{2}\right)$
by removing all nodes and associated connections outside the area
$A_{\frac{1}{r_{\rho}}}$ of $\mathcal{G}\left(\mathcal{X}_{\rho},g,\Re^{2}\right)$.
Define the effect associated with the above removal operation as the
truncation effect. A prerequisite for any (asymptotic) conclusion
obtained in the infinite network model to be applicable to the dense
and extended network models is that the impact of the truncation effect
must be vanishingly small on the parameter concerned as $\rho\rightarrow\infty$
- a conclusion that often needs non-trivial analysis to establish.
We then conducted two case studies using a random connection model,
on the expected number of isolated nodes and on the vanishing of components
of fixed and finite order $k>1$ respectively, with a focus on examining
the impact of the truncation effect and showed that the connection
function $g$ has to decrease sufficiently fast in order for the truncation
effect to have a vanishingly small impact. 

In the first case study, we showed that for $g$ satisfying both \eqref{eq:conditions on g(x) - non-increasing}
and \eqref{eq:Condition on g(x) requirement 2}, the impact of the
truncation effect on the number of isolated nodes in $\mathcal{G}\left(\mathcal{X}_{\frac{\log\rho+b}{C}},g,A_{\frac{1}{r_{\rho}}}\right)$
is vanishingly small as $\rho\rightarrow\infty$. However for $g$
satisfying \eqref{eq:conditions on g(x) - non-increasing} and \eqref{eq:conditions on g(x) - integral boundness}
only, the impact of the truncation effect on the number of isolated
nodes in $\mathcal{G}\left(\mathcal{X}_{\frac{\log\rho+b}{C}},g,A_{\frac{1}{r_{\rho}}}\right)$
is non-negligible and may even be the dominant factor in determining
the number of isolated nodes.

In the second case study, we first showed using an example that due
to the truncation effect, asymptotic vanishing of components of fixed
and finite order $k>1$ in an infinite network does not carry over
straightforwardly to the conclusion that components of fixed and finite
order $k>1$ also vanish asymptotically in the dense and extended
networks. Then to fill this theoretical gap, a result is presented
on the asymptotic vanishing of components of finite order $k>1$ in
the dense and extended network models under a random connection model. 

Some interesting results useful for the analysis of connectivity under
a random connection model in the dense and extended networks were
also established. These include the expected number of isolated nodes,
which resulted in a necessary condition for a dense (or extended)
network to be connected, the vanishingly small impact of the boundary
effect on the number of isolated nodes, and the asymptotic vanishing
of components of finite order $k>1$.

Many results in the paper were given in the form of \emph{sufficient}
conditions on the connection function $g$ required for the impact
of the truncation effect to be vanishingly small. It will be interesting
and important to examine \emph{necessary} conditions on $g$ required
for the impact of the truncation effect to be vanishingly small.

\section*{Appendix I Proof of Theorem \ref{thm:expected isolated nodes asymptotically infinite square}}

In this Appendix, we give a proof of Theorem \ref{thm:expected isolated nodes asymptotically infinite square}.

We analyze $E\left(W\right)$ as $\rho\rightarrow\infty$. Denote
by $D\left(\boldsymbol{y},r_{\rho}^{-\varepsilon}\right)$ a disk
centered at $\boldsymbol{y}$ and with a radius $r_{\rho}^{-\varepsilon}$,
where $\varepsilon$ is a small positive constant and $\varepsilon<\frac{1}{4}$.
 Denote by $B\left(A_{\frac{1}{r_{\rho}}}\right)\subset A_{\frac{1}{r_{\rho}}}$
an area within $r_{\rho}^{-\varepsilon}$ of the border of $A_{\frac{1}{r_{\rho}}}$;
denote by $\ell A_{\frac{1}{r_{\rho}}}\subset A_{\frac{1}{r_{\rho}}}$
a rectangular area of size $r_{\rho}^{-\varepsilon}\times\left(r_{\rho}^{-1}-2r_{\rho}^{-\varepsilon}\right)$
within $r_{\rho}^{-\varepsilon}$ of one side of $A_{\frac{1}{r_{\rho}}}$,
away from the corners of $A_{\frac{1}{r_{\rho}}}$ by $r_{\rho}^{-\varepsilon}$,
and there are four such areas; let $\angle A_{\frac{1}{r_{\rho}}}\subset A_{\frac{1}{r_{\rho}}}$
denote a square of size $r_{\rho}^{-\varepsilon}\times r_{\rho}^{-\varepsilon}$
at the four corners of $A_{\frac{1}{r_{\rho}}}$. Fig. \ref{fig:An-Illustration-of-the boundary areas}
illustrates these areas.

\begin{figure}

\begin{centering}
\includegraphics[width=0.35\columnwidth]{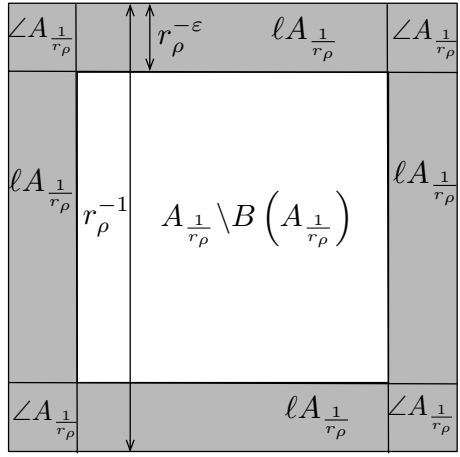}
\par\end{centering}

\caption{An Illustration of the boundary areas of $A_{\frac{1}{r_{\rho}}}$.
The areas $\angle A_{\frac{1}{r_{\rho}}}$, $\ell A_{\frac{1}{r_{\rho}}}$
are self-explanatory and $B\left(A_{\frac{1}{r_{\rho}}}\right)$ is
the shaded area in the figure.\label{fig:An-Illustration-of-the boundary areas}}

\end{figure}

It follows from \eqref{eq:expected number of isolated nodes in a square finite}
that

\begin{eqnarray}
 &  & \lim_{\rho\rightarrow\infty}E\left(W\right)\nonumber \\
 & = & \lim_{\rho\rightarrow\infty}\int_{A_{\frac{1}{r_{\rho}}}}\frac{\log\rho+b}{C}e^{-\int_{A_{\frac{1}{r_{\rho}}}}\frac{\log\rho+b}{C}g\left(\left\Vert \boldsymbol{x}-\boldsymbol{y}\right\Vert \right)d\boldsymbol{x}}d\boldsymbol{y}\nonumber \\
 & = & \lim_{\rho\rightarrow\infty}\rho r_{\rho}^{2}\int_{A_{\frac{1}{r_{\rho}}}\backslash B\left(A_{\frac{1}{r_{\rho}}}\right)}e^{-\rho r_{\rho}^{2}\int_{A_{\frac{1}{r_{\rho}}}}g\left(\left\Vert \boldsymbol{x}-\boldsymbol{y}\right\Vert \right)d\boldsymbol{x}}d\boldsymbol{y}\nonumber \\
 & + & \lim_{\rho\rightarrow\infty}4\rho r_{\rho}^{2}\int_{\ell A_{\frac{1}{r_{\rho}}}}e^{-\rho r_{\rho}^{2}\int_{A_{\frac{1}{r_{\rho}}}}g\left(\left\Vert \boldsymbol{x}-\boldsymbol{y}\right\Vert \right)d\boldsymbol{x}}d\boldsymbol{y}\nonumber \\
 & + & \lim_{\rho\rightarrow\infty}4\rho r_{\rho}^{2}\int_{\angle A_{\frac{1}{r_{\rho}}}}e^{-\rho r_{\rho}^{2}\int_{A_{\frac{1}{r_{\rho}}}}g\left(\left\Vert \boldsymbol{x}-\boldsymbol{y}\right\Vert \right)d\boldsymbol{x}}d\boldsymbol{y}\label{eq:asymptotic expected isolated nodes boundary effect}\end{eqnarray}
 The three summands in \eqref{eq:asymptotic expected isolated nodes boundary effect}
represent respectively the expected number of isolated nodes in the
central area $A_{\frac{1}{r_{\rho}}}\backslash B\left(A_{\frac{1}{r_{\rho}}}\right)$,
in the boundary area along the four sides of $A_{\frac{1}{r_{\rho}}}$
and in the four corners of $A_{\frac{1}{r_{\rho}}}$. In the following
analysis, we will show that for $g$ satisfying both \eqref{eq:conditions on g(x) - non-increasing}
and \eqref{eq:Condition on g(x) requirement 2}, the first term approaches
$e^{-b}$ as $\rho\rightarrow\infty$, and the second and the third
terms approach $0$ as $\rho\rightarrow\infty$.

Consider the first summand in \eqref{eq:asymptotic expected isolated nodes boundary effect}.
Using the definition of $r_{\rho}$ in \eqref{eq:definition of r_rho},
first it can be shown that for any $\boldsymbol{y}$ and therefore
$\boldsymbol{y}\in A_{\frac{1}{r_{\rho}}}\backslash B\left(A_{\frac{1}{r_{\rho}}}\right)$
(see Fig. \ref{fig:An-Illustration-of-the boundary areas} for the
region $A_{\frac{1}{r_{\rho}}}\backslash B\left(A_{\frac{1}{r_{\rho}}}\right)$):

\begin{eqnarray}
 &  & \lim_{\rho\rightarrow\infty}\rho e^{-\rho r_{\rho}^{2}\int_{D\left(\boldsymbol{y},r_{\rho}^{-\varepsilon}\right)}g\left(\left\Vert \boldsymbol{x}-\boldsymbol{y}\right\Vert \right)d\boldsymbol{x}}\nonumber \\
 & = & \lim_{\rho\rightarrow\infty}\rho e^{-\rho r_{\rho}^{2}\left(\int_{\Re^{2}}g\left(\left\Vert \boldsymbol{x}-\boldsymbol{y}\right\Vert \right)d\boldsymbol{x}-\int_{\Re^{2}\backslash D\left(\boldsymbol{y},r_{\rho}^{-\varepsilon}\right)}g\left(\left\Vert \boldsymbol{x}-\boldsymbol{y}\right\Vert \right)d\boldsymbol{x}\right)}\nonumber \\
 & = & \lim_{\rho\rightarrow\infty}\rho e^{-\rho r_{\rho}^{2}\left(C-\int_{\Re^{2}\backslash D\left(\boldsymbol{y},r_{\rho}^{-\varepsilon}\right)}g\left(\left\Vert \boldsymbol{x}-\boldsymbol{y}\right\Vert \right)d\boldsymbol{x}\right)}\nonumber \\
 & = & \lim_{\rho\rightarrow\infty}e^{-b}e^{\rho r_{\rho}^{2}\int_{\Re^{2}\backslash D\left(\boldsymbol{y},r_{\rho}^{-\varepsilon}\right)}g\left(\left\Vert \boldsymbol{x}-\boldsymbol{y}\right\Vert \right)d\boldsymbol{x}}\nonumber \\
 & = & e^{-b}\lim_{\rho\rightarrow\infty}e^{\frac{\log\rho+b}{C}\int_{r_{\rho}^{-\varepsilon}}^{\infty}2\pi rg\left(r\right)dr}\label{eq:analysis on the first term}\end{eqnarray}

It can be shown further using \eqref{eq:definition of r_rho} that
(following the equation, detailed explanations are given):

\begin{align}
 & \lim_{\rho\rightarrow\infty}\frac{\log\rho+b}{C}\int_{r_{\rho}^{-\varepsilon}}^{\infty}2\pi rg\left(r\right)dr\nonumber \\
= & \lim_{\rho\rightarrow\infty}\frac{\int_{r_{\rho}^{-\varepsilon}}^{\infty}2\pi rg\left(r\right)dr}{\frac{C}{\log\rho+b}}\nonumber \\
= & \lim_{\rho\rightarrow\infty}\frac{-2\pi r_{\rho}^{-\varepsilon}g\left(r_{\rho}^{-\varepsilon}\right)\left(-\frac{\varepsilon}{2}r_{\rho}^{-\varepsilon-2}\frac{1-\left(\log\rho+b\right)}{C\rho^{2}}\right)}{-\frac{C}{\rho\left(\log\rho+b\right)^{2}}}\nonumber \\
= & \lim_{\rho\rightarrow\infty}\pi\varepsilon\left(\log\rho+b\right)^{2}r_{\rho}^{-2\varepsilon-2}g\left(r_{\rho}^{-\varepsilon}\right)\frac{\log\rho+b-1}{C\rho}\nonumber \\
= & \lim_{\rho\rightarrow\infty}\pi\varepsilon\left(\log\rho+b\right)^{2}r_{\rho}^{-2\varepsilon}g\left(r_{\rho}^{-\varepsilon}\right)\label{eq:requirements on g - intermediate results}\\
= & \lim_{\rho\rightarrow\infty}\pi\varepsilon\left(\log\rho+b\right)^{2}r_{\rho}^{-2\varepsilon}o_{\rho}\left(\frac{1}{r_{\rho}^{-2\varepsilon}\log^{2}\left(r_{\rho}^{-2\varepsilon}\right)}\right)\label{eq:first term intermediate step x}\\
= & \lim_{\rho\rightarrow\infty}\left(\pi\varepsilon\left(\log\rho+b\right)^{2}\right.\nonumber \\
 & \left.o_{\rho}\left(\frac{1}{2\varepsilon^{2}\left(\log\left(\log\rho+b\right)-\log C-\log\rho\right)^{2}}\right)\right)\nonumber \\
= & 0\label{eq:first term intermediate step 1}\end{align}
where L'Hôpital's rule is used in the second step of the above equation,
and $g\left(x\right)=o_{x}\left(\frac{1}{x^{2}\log^{2}x}\right)$
is used from \eqref{eq:requirements on g - intermediate results}
to \eqref{eq:first term intermediate step x}. As a result of \eqref{eq:analysis on the first term}
and \eqref{eq:first term intermediate step 1}

\begin{equation}
\lim_{\rho\rightarrow\infty}\rho e^{-\rho r_{\rho}^{2}\int_{D\left(\boldsymbol{y},r_{\rho}^{-\varepsilon}\right)}g\left(\left\Vert \boldsymbol{x}-\boldsymbol{y}\right\Vert \right)d\boldsymbol{x}}=e^{-b}\label{eq:expected number of isolated nodes d_y}\end{equation}

It follows that (see Fig. \ref{fig:An-Illustration-of-the boundary areas}
for an illustration of the region $A_{\frac{1}{r_{\rho}}}\backslash B\left(A_{\frac{1}{r_{\rho}}}\right)$,
which is unshaded in the figure.)

\begin{align*}
 & \lim_{\rho\rightarrow\infty}\rho r_{\rho}^{2}\int_{A_{\frac{1}{r_{\rho}}}\backslash B\left(A_{\frac{1}{r_{\rho}}}\right)}e^{-\rho r_{\rho}^{2}\int_{A_{\frac{1}{r_{\rho}}}}g\left(\left\Vert \boldsymbol{x}-\boldsymbol{y}\right\Vert \right)d\boldsymbol{x}}d\boldsymbol{y}\\
\leq & \lim_{\rho\rightarrow\infty}\rho r_{\rho}^{2}\int_{A_{\frac{1}{r_{\rho}}}\backslash B\left(A_{\frac{1}{r_{\rho}}}\right)}e^{-\rho r_{\rho}^{2}\int_{D\left(\boldsymbol{y},r_{\rho}^{-\varepsilon}\right)}g\left(\left\Vert \boldsymbol{x}-\boldsymbol{y}\right\Vert \right)d\boldsymbol{x}}d\boldsymbol{y}\\
= & \lim_{\rho\rightarrow\infty}\left(\rho e^{-\rho r_{\rho}^{2}\int_{D\left(\boldsymbol{0},r_{\rho}^{-\varepsilon}\right)}g\left(\left\Vert \boldsymbol{x}\right\Vert \right)d\boldsymbol{x}}\right)\left(r_{\rho}^{2}\int_{A_{\frac{1}{r_{\rho}}}\backslash B\left(A_{\frac{1}{r_{\rho}}}\right)}d\boldsymbol{y}\right)\\
= & e^{-b}\end{align*}
and \begin{eqnarray*}
 &  & \lim_{\rho\rightarrow\infty}\rho r_{\rho}^{2}\int_{A_{\frac{1}{r_{\rho}}}\backslash B\left(A_{\frac{1}{r_{\rho}}}\right)}e^{-\rho r_{\rho}^{2}\int_{A_{\frac{1}{r_{\rho}}}}g\left(\left\Vert \boldsymbol{x}-\boldsymbol{y}\right\Vert \right)d\boldsymbol{x}}d\boldsymbol{y}\\
 & \geq & \lim_{\rho\rightarrow\infty}\rho r_{\rho}^{2}\int_{A_{\frac{1}{r_{\rho}}}\backslash B\left(A_{\frac{1}{r_{\rho}}}\right)}e^{-\rho r_{\rho}^{2}\int_{\Re^{2}}g\left(\left\Vert \boldsymbol{x}-\boldsymbol{y}\right\Vert \right)d\boldsymbol{x}}d\boldsymbol{y}\\
 & = & e^{-b}\end{eqnarray*}
Therefore

\begin{equation}
\lim_{\rho\rightarrow\infty}\rho r_{\rho}^{2}\int_{A_{\frac{1}{r_{\rho}}}/B\left(A_{\frac{1}{r_{\rho}}}\right)}e^{-\rho r_{\rho}^{2}\int_{A_{\frac{1}{r_{\rho}}}}g\left(\left\Vert \boldsymbol{x}-\boldsymbol{y}\right\Vert \right)d\boldsymbol{x}}d\boldsymbol{y}=e^{-b}\label{eq:Analysis on the first term - final result}\end{equation}

For the second term in \eqref{eq:asymptotic expected isolated nodes boundary effect},
an illustration of the boundary area for $\boldsymbol{y}\in\ell A_{\frac{1}{r_{\rho}}}$
is shown in Fig. \ref{fig:An-illustration-of-boundary-effect-left-border}. 

\begin{figure}
\noindent \begin{centering}
\subfloat[ \label{fig:An-illustration-of-boundary-effect-definition-c_y-etc}]{\noindent \begin{centering}
\includegraphics[width=0.235\columnwidth]{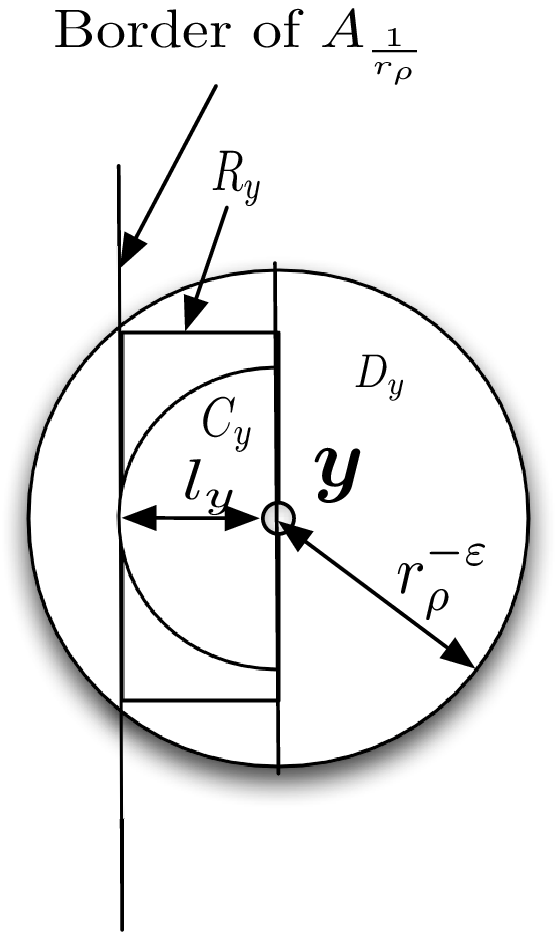}
\par\end{centering}

}\subfloat[\label{fig:An-illustration-of-boundary-effect-definition of Ly}]{\noindent \begin{centering}
\includegraphics[width=0.215\columnwidth]{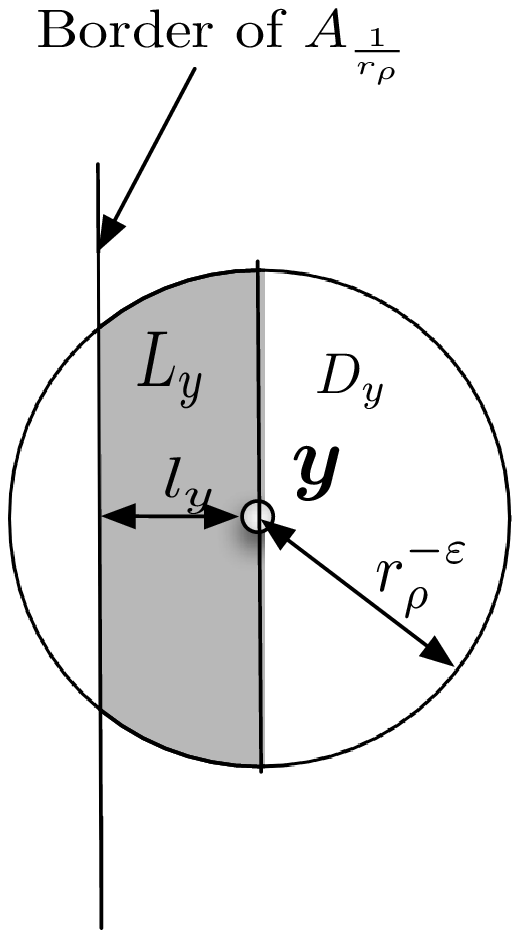}
\par\end{centering}

}
\par\end{centering}

\caption{An illustration of the boundary area for $\boldsymbol{y}\in\ell A_{\frac{1}{r_{\rho}}}$.
The figure is drawn for $\boldsymbol{y}$ located near the left border
of $A_{\frac{1}{r_{\rho}}}$. The situations for $\boldsymbol{y}$
near the top, bottom and right borders of $A_{\frac{1}{r_{\rho}}}$
can be drawn analogously. $D\left(\boldsymbol{y},r_{\rho}^{-\varepsilon}\right)$
is a disk centered at $\boldsymbol{y}$ and has a radius $r_{\rho}^{-\varepsilon}$.
$l_{y}$ is the distance between $\boldsymbol{y}$ and the border
of $A_{\frac{1}{r_{\rho}}}$. $D_{y}$ is a half disk centered at
$\boldsymbol{y}$, with a radius $r_{\rho}^{-\varepsilon}$ and on
the right side of $\boldsymbol{y}$. $C_{y}$ is a half disk centered
at $\boldsymbol{y}$, with a radius $l_{y}$ and on the left side
of $\boldsymbol{y}$. $R_{y}\subset A_{\frac{1}{r_{\rho}}}\cap D\left(\boldsymbol{y},r_{\rho}^{-\varepsilon}\right)$
is a rectangle of $l_{y}\times2\sqrt{r_{\rho}^{-2\varepsilon}-l_{y}^{2}}$
on the left side of $\boldsymbol{y}$. $L_{y}=\left(A_{\frac{1}{r_{\rho}}}\cap D\left(\boldsymbol{y},r_{\rho}^{-\varepsilon}\right)\right)\backslash D_{y}$
is the shaded area in sub-figure b.\label{fig:An-illustration-of-boundary-effect-left-border}}

\end{figure}

Define $L_{y}\triangleq\left(A_{\frac{1}{r_{\rho}}}\cap D\left(\boldsymbol{y},r_{\rho}^{-\varepsilon}\right)\right)\backslash D_{y}$
(i.e. the shaded area in Fig. \ref{fig:An-illustration-of-boundary-effect-definition of Ly}).
The symbols $D_{y}$, $C_{y}$, $l_{y}$ and $R_{y}$ are defined
in Fig. \ref{fig:An-illustration-of-boundary-effect-left-border}.
It can be shown that 

\begin{align}
 & 4\lim_{\rho\rightarrow\infty}\rho r_{\rho}^{2}\int_{\ell A_{\frac{1}{r_{\rho}}}}e^{-\rho r_{\rho}^{2}\int_{A_{\frac{1}{r_{\rho}}}}g\left(\left\Vert \boldsymbol{x}-\boldsymbol{y}\right\Vert \right)d\boldsymbol{x}}d\boldsymbol{y}\nonumber \\
\leq & 4\lim_{\rho\rightarrow\infty}r_{\rho}^{2}\int_{\ell A_{\frac{1}{r_{\rho}}}}\rho e^{-\rho r_{\rho}^{2}\int_{A_{\frac{1}{r_{\rho}}}\cap D\left(\boldsymbol{y},r_{\rho}^{-\varepsilon}\right)}g\left(\left\Vert \boldsymbol{x}-\boldsymbol{y}\right\Vert \right)d\boldsymbol{x}}d\boldsymbol{y}\nonumber \\
= & 4\lim_{\rho\rightarrow\infty}r_{\rho}^{2}\int_{\ell A_{\frac{1}{r_{\rho}}}}\rho e^{-\rho r_{\rho}^{2}\left(\int_{D_{y}}g\left(\left\Vert \boldsymbol{x}-\boldsymbol{y}\right\Vert \right)d\boldsymbol{x}+\int_{L_{y}}g\left(\left\Vert \boldsymbol{x}-\boldsymbol{y}\right\Vert \right)d\boldsymbol{x}\right)}d\boldsymbol{y}\nonumber \\
= & 4\lim_{\rho\rightarrow\infty}\left(\left(\rho^{\frac{1}{2}}e^{-\frac{1}{2}\rho r_{\rho}^{2}\int_{D\left(\boldsymbol{0},r_{\rho}^{-\varepsilon}\right)}g\left(\left\Vert \boldsymbol{x}\right\Vert \right)d\boldsymbol{x}}\right)\right.\nonumber \\
 & \left.\left(\rho^{\frac{1}{2}}r_{\rho}^{2}\int_{\ell A_{\frac{1}{r_{\rho}}}}e^{-\rho r_{\rho}^{2}\int_{L_{y}}g\left(\left\Vert \boldsymbol{x}-\boldsymbol{y}\right\Vert \right)d\boldsymbol{x}}d\boldsymbol{y}\right)\right)\label{eq:second term intermediate step}\end{align}

For the first term $\rho^{\frac{1}{2}}e^{-\frac{1}{2}\rho r_{\rho}^{2}\int_{D\left(\boldsymbol{0},r_{\rho}^{-\varepsilon}\right)}g\left(\left\Vert \boldsymbol{x}\right\Vert \right)d\boldsymbol{x}}$
in \eqref{eq:second term intermediate step}, it can be shown that

\begin{eqnarray}
 &  & \lim_{\rho\rightarrow\infty}\rho^{\frac{1}{2}}e^{-\frac{1}{2}\rho r_{\rho}^{2}\int_{D\left(\boldsymbol{0},r_{\rho}^{-\varepsilon}\right)}g\left(\left\Vert \boldsymbol{x}\right\Vert \right)d\boldsymbol{x}}\nonumber \\
 & = & \lim_{\rho\rightarrow\infty}\rho^{\frac{1}{2}}e^{-\frac{1}{2}\rho r_{\rho}^{2}\left(\int_{\Re^{2}}g\left(\left\Vert \boldsymbol{x}\right\Vert \right)d\boldsymbol{x}-\int_{\Re^{2}\backslash D\left(\boldsymbol{0},r_{\rho}^{-\varepsilon}\right)}g\left(\left\Vert \boldsymbol{x}\right\Vert \right)d\boldsymbol{x}\right)}\nonumber \\
 & = & \lim_{\rho\rightarrow\infty}\rho^{\frac{1}{2}}e^{-\frac{1}{2}\rho r_{\rho}^{2}C}e^{\frac{1}{2}\rho r_{\rho}^{2}\int_{\Re^{2}\backslash D\left(\boldsymbol{0},r_{\rho}^{-\varepsilon}\right)}g\left(\left\Vert \boldsymbol{x}\right\Vert \right)d\boldsymbol{x}}\nonumber \\
 & = & e^{-\frac{b}{2}}\label{eq:first iterm in the second term step}\end{eqnarray}
where \eqref{eq:first term intermediate step 1} is used in reaching
\eqref{eq:first iterm in the second term step}. 

Let $\gamma$ be a positive constant and $\frac{1}{2}>\gamma>\frac{\varepsilon}{2}$.
Let $\triangle$ be a positive constant such that \begin{equation}
\int_{0}^{\triangle}2\pi xg\left(x\right)dx=\gamma2C\label{eq:second term construction of the delta}\end{equation}
The existence of such a positive constant $\triangle$ can be shown
by using \eqref{eq:definition of C} and noting that $2\gamma<1$.
Using the non-increasing property of $g$, it can also be shown that
$g\left(\triangle\right)>0$; otherwise it can be shown that $\int_{0}^{\triangle}2\pi xg\left(x\right)dx=C$
which implies $\gamma=\frac{1}{2}$. This constitutes a contradiction
with the requirement that $\frac{1}{2}>\gamma>\frac{\varepsilon}{2}$.
Therefore $g\left(\triangle\right)>0$. In the following analysis,
it is assumed that $\rho$ is sufficiently large such that $r_{\rho}^{-\varepsilon}\geq2\triangle$.

For the second term in \eqref{eq:second term intermediate step},
it can be shown that 

\begin{eqnarray}
 &  & \lim_{\rho\rightarrow\infty}\rho^{\frac{1}{2}}r_{\rho}^{2}\int_{\ell A_{\frac{1}{r_{\rho}}}}e^{-\rho r_{\rho}^{2}\int_{L_{y}}g\left(\left\Vert \boldsymbol{x}-\boldsymbol{y}\right\Vert \right)d\boldsymbol{x}}d\boldsymbol{y}\nonumber \\
 & = & \lim_{\rho\rightarrow\infty}\left(\rho^{\frac{1}{2}}r_{\rho}^{2}\left(r_{\rho}^{-1}-2r_{\rho}^{-\varepsilon}\right)\right.\nonumber \\
 &  & \times\left.\int_{0}^{r_{\rho}^{-\varepsilon}}e^{-\rho r_{\rho}^{2}\int_{L_{y}}g\left(\left\Vert \boldsymbol{x}-\boldsymbol{y}\right\Vert \right)d\boldsymbol{x}}dy\right)\label{eq:an interpretation of y}\\
 & \leq & \lim_{\rho\rightarrow\infty}\rho^{\frac{1}{2}}r_{\rho}\int_{0}^{r_{\rho}^{-\varepsilon}}e^{-\rho r_{\rho}^{2}\int_{L_{y}}g\left(\left\Vert \boldsymbol{x}-\boldsymbol{y}\right\Vert \right)d\boldsymbol{x}}dy\nonumber \\
 & = & \lim_{\rho\rightarrow\infty}\sqrt{\frac{\log\rho+b}{C}}\int_{0}^{r_{\rho}^{-\varepsilon}}e^{-\rho r_{\rho}^{2}\int_{L_{y}}g\left(\left\Vert \boldsymbol{x}-\boldsymbol{y}\right\Vert \right)d\boldsymbol{x}}dy\nonumber \\
 & = & \lim_{\rho\rightarrow\infty}\sqrt{\frac{\log\rho+b}{C}}\left(\int_{0}^{\triangle}e^{-\rho r_{\rho}^{2}\int_{L_{y}}g\left(\left\Vert \boldsymbol{x}-\boldsymbol{y}\right\Vert \right)d\boldsymbol{x}}dy\right.\nonumber \\
 & + & \left.\int_{\triangle}^{r_{\rho}^{-\varepsilon}}e^{-\rho r_{\rho}^{2}\int_{L_{y}}g\left(\left\Vert \boldsymbol{x}-\boldsymbol{y}\right\Vert \right)d\boldsymbol{x}}dy\right)\label{eq:second term intermediate step-1}\end{eqnarray}
where in \eqref{eq:an interpretation of y} $\boldsymbol{y}$ represents
a (any) point in $\ell A_{\rho}$ at a Euclidean distance $y\in[0,r_{\rho}^{-\varepsilon}]$
apart from the border of $A_{\rho}$. Define $\lambda\triangleq\frac{\log\rho+b}{C}$
for convenience, it can be further shown that in \eqref{eq:second term intermediate step-1}

\begin{eqnarray}
 &  & \lim_{\rho\rightarrow\infty}\sqrt{\lambda}\int_{\triangle}^{r_{\rho}^{-\varepsilon}}e^{-\rho r_{\rho}^{2}\int_{L_{y}}g\left(\left\Vert \boldsymbol{x}-\boldsymbol{y}\right\Vert \right)d\boldsymbol{x}}dy\nonumber \\
 & \leq & \lim_{\rho\rightarrow\infty}\sqrt{\lambda}\int_{\triangle}^{r_{\rho}^{-\varepsilon}}e^{-\rho r_{\rho}^{2}\int_{C_{y}}g\left(\left\Vert \boldsymbol{x}-\boldsymbol{y}\right\Vert \right)d\boldsymbol{x}}dy\nonumber \\
 & = & \lim_{\rho\rightarrow\infty}\sqrt{\lambda}\int_{\triangle}^{r_{\rho}^{-\varepsilon}}e^{-\frac{1}{2}\rho r_{\rho}^{2}\int_{0}^{y}2\pi xg\left(x\right)dx}dy\nonumber \\
 & = & \lim_{\rho\rightarrow\infty}\sqrt{\lambda}\int_{\triangle}^{r_{\rho}^{-\varepsilon}}e^{-\frac{1}{2}\rho r_{\rho}^{2}\left(\int_{0}^{\triangle}2\pi xg\left(x\right)dx+\int_{\triangle}^{y}2\pi xg\left(x\right)dx\right)}dy\nonumber \\
 & \leq & \lim_{\rho\rightarrow\infty}\sqrt{\lambda}\int_{\triangle}^{r_{\rho}^{-\varepsilon}}e^{-\frac{1}{2}\rho r_{\rho}^{2}\int_{0}^{\triangle}2\pi xg\left(x\right)dx}dy\nonumber \\
 & = & \lim_{\rho\rightarrow\infty}\sqrt{\lambda}\int_{\triangle}^{r_{\rho}^{-\varepsilon}}e^{-\gamma\left(\log\rho+b\right)}dy\label{eq:second term dertivation of the half disk step 1}\\
 & = & \lim_{\rho\rightarrow\infty}\sqrt{\lambda}\left(e^{-\gamma b}\rho^{-\gamma}\left(\left(\frac{\log\rho+b}{C\rho}\right)^{-\frac{\varepsilon}{2}}-\triangle\right)\right)\nonumber \\
 & = & 0\label{eq:second term derivation of the hald disk step 2}\end{eqnarray}
where \eqref{eq:second term construction of the delta} is used in
reaching \eqref{eq:second term dertivation of the half disk step 1},
and $\gamma>\frac{\varepsilon}{2}$ is used in reaching \eqref{eq:second term derivation of the hald disk step 2}.
It can also be shown that for the other term in \eqref{eq:second term intermediate step-1},

\begin{eqnarray}
 &  & \lim_{\rho\rightarrow\infty}\sqrt{\frac{\log\rho+b}{C}}\int_{0}^{\triangle}e^{-\rho r_{\rho}^{2}\int_{L_{y}}g\left(\left\Vert \boldsymbol{x}-\boldsymbol{y}\right\Vert \right)d\boldsymbol{x}}dy\nonumber \\
 & \leq & \lim_{\rho\rightarrow\infty}\sqrt{\lambda}\int_{0}^{\triangle}e^{-\rho r_{\rho}^{2}\int_{R_{y}}g\left(\left\Vert \boldsymbol{x}-\boldsymbol{y}\right\Vert \right)d\boldsymbol{x}}dy\label{eq:second term rectangular analysis step2-2}\\
 & = & \lim_{\rho\rightarrow\infty}\sqrt{\lambda}\int_{0}^{\triangle}e^{-\rho r_{\rho}^{2}2\int_{0}^{y}\int_{0}^{\sqrt{r_{\rho}^{-2\varepsilon}-x^{2}}}g\left(\sqrt{x^{2}+z^{2}}\right)dzdx}dy\nonumber \\
 & \leq & \lim_{\rho\rightarrow\infty}\sqrt{\lambda}\int_{0}^{\triangle}e^{-\rho r_{\rho}^{2}2\int_{0}^{y}\int_{0}^{r_{\rho}^{-\varepsilon}-x}g\left(\sqrt{x^{2}+z^{2}}\right)dzdx}dy\label{eq:second term rectangular analysis step 2}\\
 & \leq & \lim_{\rho\rightarrow\infty}\sqrt{\lambda}\int_{0}^{\triangle}e^{-\rho r_{\rho}^{2}2\int_{0}^{y}\int_{0}^{r_{\rho}^{-\varepsilon}-\triangle}g\left(\sqrt{x^{2}+z^{2}}\right)dzdx}dy\label{eq:second term rectangular analysis step 3}\\
 & \leq & \lim_{\rho\rightarrow\infty}\sqrt{\lambda}\int_{0}^{\triangle}e^{-\rho r_{\rho}^{2}2\int_{0}^{y}\int_{0}^{r_{\rho}^{-\varepsilon}-\triangle}g\left(z+\triangle\right)dzdx}dy\label{eq:second term rectangular analysis step 4}\\
 & = & \lim_{\rho\rightarrow\infty}\sqrt{\lambda}\int_{0}^{\triangle}e^{-\rho r_{\rho}^{2}2y\int_{0}^{r_{\rho}^{-\varepsilon}-\triangle}g\left(z+\triangle\right)dz}dy\label{eq:second term rectangular analysis step 1}\end{eqnarray}
where \eqref{eq:second term rectangular analysis step 2} is obtained
by noting that $r_{\rho}^{-2\varepsilon}-x^{2}\geq\left(r_{\rho}^{-\varepsilon}-x\right)^{2}$
for $r_{\rho}^{-\varepsilon}\geq x$ (note that for $\rho$ sufficiently
large, $r_{\rho}^{-\varepsilon}>\triangle\geq y\geq x$); \eqref{eq:second term rectangular analysis step 3}
is obtained by noting that $x\leq\triangle$ and \eqref{eq:second term rectangular analysis step 4}
is obtained by noting that $y\leq\triangle$ and the non-increasing
property of $g$. 

Let $\rho$ be sufficiently large such that $r_{\rho}^{-\varepsilon}\geq2\triangle$
and also note that $g\left(\triangle\right)>0$. Therefore $\beta\triangleq\int_{0}^{\triangle}g\left(z+\triangle\right)dz$
is a positive constant and $\beta>0$. It then follows from \eqref{eq:second term rectangular analysis step 1}
that \begin{eqnarray}
 &  & \lim_{\rho\rightarrow\infty}\sqrt{\frac{\log\rho+b}{C}}\int_{0}^{\triangle}e^{-\rho r_{\rho}^{2}\int_{L_{y}}g\left(\left\Vert \boldsymbol{x}-\boldsymbol{y}\right\Vert \right)d\boldsymbol{x}}dy\nonumber \\
 & \leq & \lim_{\rho\rightarrow\infty}\sqrt{\frac{\log\rho+b}{C}}\int_{0}^{\triangle}e^{-\rho r_{\rho}^{2}2\beta y}dy\nonumber \\
 & = & \lim_{\rho\rightarrow\infty}\sqrt{\frac{\log\rho+b}{C}}\times\frac{1-e^{-\rho r_{\rho}^{2}2\beta\triangle}}{\rho r_{\rho}^{2}2\beta}\nonumber \\
 & = & \lim_{\rho\rightarrow\infty}\sqrt{\frac{\log\rho+b}{C}}\times\frac{1-e^{-2\beta\triangle\frac{\log\rho+b}{C}}}{2\beta\frac{\log\rho+b}{C}}\nonumber \\
 & = & 0\label{eq:second term rectangular analysis final step}\end{eqnarray}

As a result of \eqref{eq:second term derivation of the hald disk step 2}
and \eqref{eq:second term rectangular analysis final step}, both
terms on the right hand side of \eqref{eq:second term intermediate step-1}
go to zero and it follows that \[
\lim_{\rho\rightarrow\infty}\rho^{\frac{1}{2}}r_{\rho}^{2}\int_{\ell A_{\rho}}e^{-\rho r_{\rho}^{2}\int_{L_{y}}g\left(\left\Vert \boldsymbol{x}-\boldsymbol{y}\right\Vert \right)d\boldsymbol{x}}d\boldsymbol{y}=0\]
The above equation, together with \eqref{eq:second term intermediate step}
and \eqref{eq:first iterm in the second term step}, leads to the
conclusion that \begin{equation}
4\lim_{\rho\rightarrow\infty}\rho r_{\rho}^{2}\int_{\ell A_{\rho}}e^{-\rho r_{\rho}^{2}\int_{A_{\rho}}g\left(\left\Vert \boldsymbol{x}-\boldsymbol{y}\right\Vert \right)d\boldsymbol{x}}d\boldsymbol{y}=0\label{eq:second term final result}\end{equation}
i.e. the second term in \eqref{eq:asymptotic expected isolated nodes boundary effect}
approaches $0$ as $\rho\rightarrow\infty$.

For the third term in \eqref{eq:asymptotic expected isolated nodes boundary effect},
it can be shown that

\begin{eqnarray}
 &  & 4\lim_{\rho\rightarrow\infty}\rho r_{\rho}^{2}\int_{\angle A_{\frac{1}{r_{\rho}}}}e^{-\rho r_{\rho}^{2}\int_{A_{\frac{1}{r_{\rho}}}}g\left(\left\Vert \boldsymbol{x}-\boldsymbol{y}\right\Vert \right)d\boldsymbol{x}}d\boldsymbol{y}\nonumber \\
 & \leq & 4\lim_{\rho\rightarrow\infty}\rho r_{\rho}^{2}\int_{\angle A_{\frac{1}{r_{\rho}}}}e^{-\rho r_{\rho}^{2}\int_{A_{\frac{1}{r_{\rho}}}\cap D\left(\boldsymbol{y},r_{\rho}^{-\varepsilon}\right)}g\left(\left\Vert \boldsymbol{x}-\boldsymbol{y}\right\Vert \right)d\boldsymbol{x}}d\boldsymbol{y}\nonumber \\
 & \leq & 4\lim_{\rho\rightarrow\infty}\rho r_{\rho}^{2}\int_{\angle A_{\frac{1}{r_{\rho}}}}e^{-\frac{1}{4}\rho r_{\rho}^{2}\int_{D\left(\boldsymbol{y},r_{\rho}^{-\varepsilon}\right)}g\left(\left\Vert \boldsymbol{x}-\boldsymbol{y}\right\Vert \right)d\boldsymbol{x}}d\boldsymbol{y}\nonumber \\
 & = & 4\lim_{\rho\rightarrow\infty}\left(r_{\rho}^{-\varepsilon}\right)^{2}r_{\rho}^{2}\rho e^{-\frac{1}{4}\rho r_{\rho}^{2}\int_{D\left(\boldsymbol{y},r_{\rho}^{-\varepsilon}\right)}g\left(\left\Vert \boldsymbol{x}-\boldsymbol{y}\right\Vert \right)d\boldsymbol{x}}\nonumber \\
 & = & 4\lim_{\rho\rightarrow\infty}r_{\rho}^{2-2\varepsilon}\rho e^{-\frac{1}{4}\rho r_{\rho}^{2}\left(C-\int_{\Re^{2}\backslash D\left(\boldsymbol{y},r_{\rho}^{-\varepsilon}\right)}g\left(\left\Vert \boldsymbol{x}-\boldsymbol{y}\right\Vert \right)d\boldsymbol{x}\right)}\nonumber \\
 & = & 4\lim_{\rho\rightarrow\infty}\left(\frac{\log\rho+b}{C\rho}\right)^{1-\varepsilon}\rho e^{-\frac{1}{4}\left(\log\rho+b\right)}\label{eq:third term intermediate result}\\
 & = & 4C^{-1+\varepsilon}e^{-\frac{1}{4}b}\lim_{\rho\rightarrow\infty}\frac{\left(\log\rho+b\right)^{1-\varepsilon}}{\rho^{\frac{1}{4}-\varepsilon}}\nonumber \\
 & = & 0\label{eq:third term final result}\end{eqnarray}
where the second step results by noting that for any $\boldsymbol{y}\in\angle A_{\frac{1}{r_{\rho}}}$,
$A_{\rho}\cap D\left(\boldsymbol{y},r_{\rho}^{-\varepsilon}\right)$
covers at least one quarter of $D\left(\boldsymbol{y},r_{\rho}^{-\varepsilon}\right)$,
\eqref{eq:first term intermediate step 1} is used in reaching \eqref{eq:third term intermediate result},
and $\varepsilon<\frac{1}{4}$ is used in the final step.

As a result of \eqref{eq:asymptotic expected isolated nodes boundary effect},
\eqref{eq:Analysis on the first term - final result}, \eqref{eq:second term final result}
and \eqref{eq:third term final result}:\begin{equation}
\lim_{\rho\rightarrow\infty}E\left(W\right)=e^{-b}\label{eq:expected number of isolated nodes square final}\end{equation}

\section*{Appendix II: Proof of Lemma \ref{lem:Expected Isolated nodes torus}}

The torus that is commonly discussed in random geometric graph theory
is essentially the same as a square except that the distance between
two points on a torus is defined by their\emph{ toroidal distance},
instead of Euclidean distance. Thus a pair of nodes in $\mathcal{G}^{T}\left(\mathcal{X}_{\frac{\log\rho+b}{C}},g,A_{\frac{1}{r_{\rho}}}^{T}\right)$,
located at $\boldsymbol{x}_{1}$ and $\boldsymbol{x}_{2}$ respectively,
are directly connected with probability $g\left(\left\Vert \boldsymbol{x}_{1}-\boldsymbol{x}_{2}\right\Vert ^{T}\right)$
where $\left\Vert \boldsymbol{x}_{1}-\boldsymbol{x}_{2}\right\Vert ^{T}$
denotes the \emph{toroidal distance} between the two nodes. For a
unit torus $A^{T}=\left[-\frac{1}{2},\frac{1}{2}\right]^{2}$, the
toroidal distance is given by \cite[p. 13]{Penrose03Random}:\begin{equation}
\left\Vert \boldsymbol{x}_{1}-\boldsymbol{x}_{2}\right\Vert ^{T}\triangleq\min\left\{ \left\Vert \boldsymbol{x}_{1}+\boldsymbol{z}-\boldsymbol{x}_{2}\right\Vert :\boldsymbol{z}\in\mathbb{Z}^{2}\right\} \label{eq:definition of toroidal distance in a unit torus}\end{equation}
The toroidal distance between points on a torus of any other size
can be computed analogously. 
\begin{remrk}
The use of toroidal distance allows nodes located near the boundary
to have the same number of connections \emph{probabilistically} as
a node located near the center. Therefore it allows the removal of
the boundary effect that is present in a square. The consideration
of a torus implies that there is no need to consider special cases
occurring near the boundary of the region and that events inside the
region do not depend on the particular location inside the region.
This often simplifies the analysis. 
\end{remrk}
From now on, whenever the difference between a torus and a square
affects the parameter being discussed, we use superscript $^{T}$
to mark the parameter in a torus. 

We note the following relation between toroidal distance and Euclidean
distance on a square area centered at the origin:\begin{eqnarray}
\left\Vert \boldsymbol{x}_{1}-\boldsymbol{x}_{2}\right\Vert ^{T} & \leq & \left\Vert \boldsymbol{x}_{1}-\boldsymbol{x}_{2}\right\Vert \label{eq:property of toroidal distance 1}\\
\left\Vert \boldsymbol{x}\right\Vert ^{T} & = & \left\Vert \boldsymbol{x}\right\Vert \label{eq:property of toroidal distance 2}\end{eqnarray}
which will be used in the later analysis.

It can then be shown that for an arbitrary node in $\mathcal{G}^{T}\left(\mathcal{X}_{\frac{\log\rho+b}{C}},g,A_{\frac{1}{r_{\rho}}}^{T}\right)$
at location $\boldsymbol{y}$, the probability that it is isolated
is given by:

\begin{eqnarray*}
\Pr\left(I_{\boldsymbol{y}}^{T}=1\right) & = & e^{-\int_{A_{\frac{1}{r_{\rho}}}^{T}}\frac{\log\rho+b}{C}g\left(\left\Vert \boldsymbol{x}-\boldsymbol{y}\right\Vert ^{T}\right)d\boldsymbol{x}}\\
 & = & e^{-\int_{A_{\frac{1}{r_{\rho}}}^{T}}\frac{\log\rho+b}{C}g\left(\left\Vert \boldsymbol{x}\right\Vert ^{T}\right)d\boldsymbol{x}}\\
 & = & e^{-\int_{A_{\frac{1}{r_{\rho}}}}\frac{\log\rho+b}{C}g\left(\left\Vert \boldsymbol{x}\right\Vert \right)d\boldsymbol{x}}\end{eqnarray*}
where in the second step, the property of a torus that the probability
that an arbitrary node at location $\boldsymbol{y}$ is isolated is
equal to the probability that a node at the origin is isolated is
used; in the third step \eqref{eq:property of toroidal distance 2}
is used.

Thus the expected number of isolated nodes in $\mathcal{G}^{T}\left(\mathcal{X}_{\frac{\log\rho+b}{C}},g,A_{\frac{1}{r_{\rho}}}^{T}\right)$
is given by

\begin{eqnarray}
 &  & E\left(W^{T}\right)\nonumber \\
 & = & \int_{A_{\frac{1}{r_{\rho}}}}\frac{\log\rho+b}{C}e^{-\int_{A_{\frac{1}{r_{\rho}}}}\frac{\log\rho+b}{C}g\left(\left\Vert \boldsymbol{x}\right\Vert \right)d\boldsymbol{x}}d\boldsymbol{y}\label{eq:expected number of isolated nodes in a torus finite}\\
 & = & \frac{1}{r_{\rho}^{2}}\frac{\log\rho+b}{C}e^{-\int_{A_{\frac{1}{r_{\rho}}}}\frac{\log\rho+b}{C}g\left(\left\Vert \boldsymbol{x}\right\Vert \right)d\boldsymbol{x}}\\
 & = & \rho e^{-\int_{A_{\frac{1}{r_{\rho}}}}\frac{\log\rho+b}{C}g\left(\left\Vert \boldsymbol{x}\right\Vert \right)d\boldsymbol{x}}\end{eqnarray}

First it can be shown using \eqref{eq:definition of C} that for $g$
satisfying \eqref{eq:Condition on g(x) requirement 2}

\begin{eqnarray}
 &  & \lim_{\rho\rightarrow\infty}\rho e^{-\int_{D\left(\boldsymbol{0},r_{\rho}^{-\varepsilon}\right)}\frac{\log\rho+b}{C}g\left(\left\Vert \boldsymbol{x}\right\Vert \right)d\boldsymbol{x}}\nonumber \\
 & = & \lim_{\rho\rightarrow\infty}\rho e^{-\frac{\log\rho+b}{C}\left(C-\int_{\Re^{2}\backslash D\left(\boldsymbol{0},r_{\rho}^{-\varepsilon}\right)}g\left(\left\Vert \boldsymbol{x}\right\Vert \right)d\boldsymbol{x}\right)}\nonumber \\
 & = & e^{-b}\lim_{\rho\rightarrow\infty}e^{\frac{\log\rho+b}{C}\int_{r_{\rho}^{-\varepsilon}}^{\infty}2\pi xg\left(x\right)dx}\nonumber \\
 & = & e^{-b}\label{eq:an analysis of the trunction effect}\end{eqnarray}
where $D\left(\boldsymbol{0},x\right)$ denotes a disk centered at
the origin and with a radius $x$, $\varepsilon$ is a small positive
constant, and the last step results because \begin{eqnarray}
 &  & \lim_{\rho\rightarrow\infty}\frac{\int_{r_{\rho}^{-\varepsilon}}^{\infty}2\pi xg\left(x\right)dx}{\frac{1}{\log\rho+b}}\nonumber \\
 & = & \lim_{\rho\rightarrow\infty}\frac{\pi\varepsilon r_{\rho}^{-\varepsilon}g\left(r_{\rho}^{-\varepsilon}\right)r_{\rho}^{-\varepsilon-2}\frac{\log\rho+b-1}{C\rho^{2}}}{\frac{1}{\rho\left(\log\rho+b\right)^{2}}}\label{eq:little's rule trunction}\\
 & = & \lim_{\rho\rightarrow\infty}\pi\varepsilon\left(\log\rho+b\right)^{2}r_{\rho}^{-2\varepsilon}o_{\rho}\left(\frac{1}{r_{\rho}^{-2\varepsilon}\log^{2}\left(r_{\rho}^{-2\varepsilon}\right)}\right)\nonumber \\
 & = & 0\nonumber \end{eqnarray}
where L'Hôpital's rule is used in reaching \eqref{eq:little's rule trunction}
and in the third step $g\left(x\right)=o_{x}\left(\frac{1}{x^{2}\log^{2}x}\right)$
is used. Note that by definition of $C$ in \eqref{eq:definition of C},
\begin{equation}
\rho e^{-\int_{\Re^{2}}\frac{\log\rho+b}{C}g\left(\left\Vert \boldsymbol{x}\right\Vert \right)d\boldsymbol{x}}=e^{-b}\label{eq:expected number of isolated nodes area in R2}\end{equation}
and \begin{eqnarray}
 &  & \rho e^{-\int_{\Re^{2}}\frac{\log\rho+b}{C}g\left(\left\Vert \boldsymbol{x}\right\Vert \right)d\boldsymbol{x}}\nonumber \\
 & \leq & \rho e^{-\int_{A_{\frac{1}{r_{\rho}}}}\frac{\log\rho+b}{C}g\left(\left\Vert \boldsymbol{x}\right\Vert \right)d\boldsymbol{x}}\nonumber \\
 & \leq & \rho e^{-\int_{D\left(\boldsymbol{0},r_{\rho}^{-\varepsilon}\right)}\frac{\log\rho+b}{C}g\left(\left\Vert \boldsymbol{x}\right\Vert \right)d\boldsymbol{x}}\label{eq:expected number of isolated nodes inequalities}\end{eqnarray}

As a result of \eqref{eq:expected number of isolated nodes in a torus finite},
\eqref{eq:an analysis of the trunction effect}, \eqref{eq:expected number of isolated nodes area in R2}
and \eqref{eq:expected number of isolated nodes inequalities}\begin{eqnarray}
\lim_{\rho\rightarrow\infty}E\left(W^{T}\right) & = & e^{-b}\label{eq:Expected number of isolated nodes asymptotic}\end{eqnarray}

\section*{Appendix III Proof of Theorem \ref{thm:varnish of finite components}}

In this Appendix, we give a proof of Theorem \ref{thm:varnish of finite components}.

For convenience, let $\lambda$ be the node density in $\mathcal{G}\left(\mathcal{X}_{\frac{\log\rho+b}{C}},g,A_{\frac{1}{r_{\rho}}}\right)$
where $\lambda\triangleq\rho r_{\rho}^{2}=\frac{\log\rho+b}{C}$.
Using the above notations, $\mathcal{G}\left(\mathcal{X}_{\frac{\log\rho+b}{C}},g,A_{\frac{1}{r_{\rho}}}\right)$
can be written as $\mathcal{G}\left(\mathcal{X}_{\lambda},g,A_{\frac{1}{r_{\rho}}}\right)$. 

Note that for any finite $\rho$ the total number of nodes in $\mathcal{G}\left(\mathcal{X}_{\lambda},g,A_{\frac{1}{r_{\rho}}}\right)$,
hence the total number of components in $\mathcal{G}\left(\mathcal{X}_{\lambda},g,A_{\frac{1}{r_{\rho}}}\right)$,
is almost surely finite. Denote by $\xi_{k}$ the (random) number
of components of order $k$ in an instance of $\mathcal{G}\left(\mathcal{X}_{\lambda},g,A_{\frac{1}{r_{\rho}}}\right)$.
It then suffices to show that for an arbitrarily large positive integer
$M$:

\begin{equation}
\lim_{\rho\rightarrow\infty}\Pr\left(\sum_{k=2}^{M}\xi_{k}=0\right)=1\label{eq:form of theorm varnishing of finite components in math form}\end{equation}

The following symbols and notations are used in this appendix:

Denote by $g_{1}\left(\boldsymbol{x}_{1},\boldsymbol{x}_{2},\ldots,\boldsymbol{x}_{k}\right)$
the probability that a set of $k$ nodes at non-random positions $\boldsymbol{x}_{1}$,
$\boldsymbol{x}_{2}$, $\ldots$, $\boldsymbol{x}_{k}\in A_{\frac{1}{r_{\rho}}}$
forms a connected component.

Denote by $g_{2}\left(\boldsymbol{y};\boldsymbol{x}_{1},\boldsymbol{x}_{2},\ldots,\boldsymbol{x}_{k}\right)$
the probability that a node at non-random position $\boldsymbol{y}$
is connected to at least one node in $\left\{ \boldsymbol{x}_{1},\boldsymbol{x}_{2},\ldots,\boldsymbol{x}_{k}\right\} $.
It can be shown that

\begin{equation}
g_{2}\left(\boldsymbol{y};\boldsymbol{x}_{1},\boldsymbol{x}_{2},\ldots,\boldsymbol{x}_{k}\right)=1-\prod_{i=1}^{k}\left(1-g\left(\left\Vert \boldsymbol{y}-\boldsymbol{x}_{i}\right\Vert \right)\right)\label{eq:node at non-random position connected - nonisolated}\end{equation}
and \begin{equation}
g_{2}\left(\boldsymbol{y};\boldsymbol{x}_{1},\boldsymbol{x}_{2},\ldots,\boldsymbol{x}_{k}\right)\geq g_{2}\left(\boldsymbol{y};\boldsymbol{x}_{1},\boldsymbol{x}_{2},\ldots,\boldsymbol{x}_{i}\right)\text{ for }1\leq i\leq k\label{eq:inequality for g_2}\end{equation}
As an easy consequence of the union bound, \begin{equation}
g_{2}\left(\boldsymbol{y};\boldsymbol{x}_{1},\boldsymbol{x}_{2},\ldots,\boldsymbol{x}_{k}\right)\leq\sum_{i=1}^{k}g\left(\left\Vert \boldsymbol{y}-\boldsymbol{x}_{i}\right\Vert \right)\label{eq:inequality for g_2 the union bound}\end{equation}

Using the monotonicity and positive integral properties of $g$ in
\eqref{eq:conditions on g(x) - non-increasing} and \eqref{eq:conditions on g(x) - integral boundness},
it can be shown that there exists a positive constant $r$ such that
$g\left(r^{-}\right)\left(1-g\left(r^{+}\right)\right)>0$ where $g\left(r^{-}\right)\triangleq\lim_{x\rightarrow r^{-}}g\left(x\right)$
and $g\left(r^{+}\right)\triangleq\lim_{x\rightarrow r^{+}}g\left(x\right)$.
If $g$ is a continuous function, then $g\left(r^{-}\right)=g\left(r^{+}\right)$;
if $g$ is a discontinuous function, e.g. a unit disk connection model,
by choosing $r$ to be the transmission range, $g\left(r^{-}\right)\left(1-g\left(r^{+}\right)\right)=1$.
For convenience in notations, we use $\beta$ for $g\left(r^{-}\right)\left(1-g\left(r^{+}\right)\right)$,
i.e. \begin{equation}
\beta\triangleq g\left(r^{-}\right)\left(1-g\left(r^{+}\right)\right)\label{eq:definition of beta}\end{equation}

Denote by $\partial A_{\frac{1}{r_{\rho}}}$ the border of $A_{\frac{1}{r_{\rho}}}$.
Denote by $\ell A_{\frac{1}{r_{\rho}}}\subset A_{\frac{1}{r_{\rho}}}$
a rectangular area of size $\left(\frac{1}{r_{\rho}}-2r\right)\times r$
along one side of the border of $A_{\frac{1}{r_{\rho}}}$, within
a distance $r$ of the border and away from the four corners of $A_{\frac{1}{r_{\rho}}}$
by at least $r$. There are four such areas in $A_{\frac{1}{r_{\rho}}}$.
Denote by $\angle A_{\frac{1}{r_{\rho}}}\subset A_{\frac{1}{r_{\rho}}}$
a square area of size $r\times r$ located at a corner of $A_{\frac{1}{r_{\rho}}}$.
There are four such corner squares in $A_{\frac{1}{r_{\rho}}}$. Denote
by $B_{d}\left(A_{\frac{1}{r_{\rho}}}\right)\subset A_{\frac{1}{r_{\rho}}}$
a boundary area within a distance $d$ of the border of $A_{\frac{1}{r_{\rho}}}$.
Note the difference of the definitions of those symbols from those
used Appendix I and particularly Fig. \ref{fig:An-Illustration-of-the boundary areas}. 

Let $D\left(\boldsymbol{x},d\right)\subset\Re^{2}$ represents a disk
centered at $\boldsymbol{x}\in A_{\frac{1}{r_{\rho}}}$ and with a
radius $d$. 

We first establish some preliminary results that will be used in the
proof.
\begin{lemma}
\label{lem:expected number of components of size k}In $\mathcal{G}\left(\mathcal{X}_{\lambda},g,A_{\frac{1}{r_{\rho}}}\right)$,
the expected number of components of order $k$ is given by \begin{equation}
E\left(\xi_{k}\right)=\frac{\lambda^{k}}{k!}\int_{\left(A_{\frac{1}{r_{\rho}}}\right)^{k}}g_{1}\left(\boldsymbol{x}_{1},\boldsymbol{x}_{2},\ldots,\boldsymbol{x}_{k}\right)e^{-\lambda\int_{A_{\frac{1}{r_{\rho}}}}g_{2}\left(\boldsymbol{y};\boldsymbol{x}_{1},\boldsymbol{x}_{2},\ldots,\boldsymbol{x}_{k}\right)d\boldsymbol{y}}d\left(\boldsymbol{x}_{1}\cdots\boldsymbol{x}_{k}\right)\label{eq:expected number of component}\end{equation}
\end{lemma}
\begin{proof}
It can be shown that for any $n\geq k$:\begin{equation}
E\left(\xi_{k}\left|\left|\mathcal{X}_{\lambda}\right|=n\right.\right)=\frac{\left(\begin{array}{c}
n\\
k\end{array}\right)}{\left(A_{\frac{1}{r_{\rho}}}\right)^{n}}\int_{\left(A_{\frac{1}{r_{\rho}}}\right)^{n}}g_{1}\left(\boldsymbol{x}_{1},\boldsymbol{x}_{2},\ldots,\boldsymbol{x}_{k}\right)\prod_{i=k+1}^{n}\left(1-g_{2}\left(\boldsymbol{x}_{i};\boldsymbol{x}_{1},\boldsymbol{x}_{2},\ldots,\boldsymbol{x}_{k}\right)\right)d\left(\boldsymbol{x}_{1}\cdots\boldsymbol{x}_{n}\right)\label{eq:Existence of finite component uniform}\end{equation}
In \eqref{eq:Existence of finite component uniform}, $\left(\begin{array}{c}
n\\
k\end{array}\right)$ is the number of distinct sets of $k$ nodes drawn from a total of
$n$ nodes and the rest term represents the probability of the event
that a \emph{randomly chosen} set of $k$ nodes forms a component
of order $k$. From \eqref{eq:Existence of finite component uniform},
it follows that \begin{eqnarray*}
 &  & E\left(\xi_{k}\right)\\
 & = & \sum_{n=k}^{\infty}E\left(\xi_{k}\left|\left|\mathcal{X}_{\lambda}\right|=n\right.\right)\frac{\left(\lambda A_{\frac{1}{r_{\rho}}}\right)^{n}}{n!}e^{-\lambda A_{\frac{1}{r_{\rho}}}}\\
 & = & \sum_{n=k}^{\infty}\frac{\left(\lambda A_{\frac{1}{r_{\rho}}}\right)^{n}}{n!}e^{-\lambda A_{\frac{1}{r_{\rho}}}}\frac{\left(\begin{array}{c}
n\\
k\end{array}\right)}{\left(A_{\frac{1}{r_{\rho}}}\right)^{n}}\int_{\left(A_{\frac{1}{r_{\rho}}}\right)^{n}}g_{1}\left(\boldsymbol{x}_{1},\boldsymbol{x}_{2},\ldots,\boldsymbol{x}_{k}\right)\prod_{i=k+1}^{n}\left(1-g_{2}\left(\boldsymbol{x}_{i};\boldsymbol{x}_{1},\boldsymbol{x}_{2},\ldots,\boldsymbol{x}_{k}\right)\right)d\left(\boldsymbol{x}_{1}\cdots\boldsymbol{x}_{n}\right)\\
 & = & \sum_{n=k}^{\infty}\frac{\lambda^{n}}{n!}e^{-\lambda A_{\frac{1}{r_{\rho}}}}\left(\begin{array}{c}
n\\
k\end{array}\right)\int_{\left(A_{\frac{1}{r_{\rho}}}\right)^{k}}g_{1}\left(\boldsymbol{x}_{1},\boldsymbol{x}_{2},\ldots,\boldsymbol{x}_{k}\right)\left(\int_{A_{\frac{1}{r_{\rho}}}}1-g_{2}\left(\boldsymbol{y};\boldsymbol{x}_{1},\boldsymbol{x}_{2},\ldots,\boldsymbol{x}_{k}\right)d\boldsymbol{y}\right)^{n-k}d\left(\boldsymbol{x}_{1}\cdots\boldsymbol{x}_{k}\right)\\
 & = & \int_{\left(A_{\frac{1}{r_{\rho}}}\right)^{k}}g_{1}\left(\boldsymbol{x}_{1},\boldsymbol{x}_{2},\ldots,\boldsymbol{x}_{k}\right)\left(\sum_{n=k}^{\infty}\frac{\lambda^{n}}{n!}e^{-\lambda A_{\rho}}\left(\begin{array}{c}
n\\
k\end{array}\right)\left(\int_{A_{\frac{1}{r_{\rho}}}}1-g_{2}\left(\boldsymbol{y};\boldsymbol{x}_{1},\boldsymbol{x}_{2},\ldots,\boldsymbol{x}_{k}\right)d\boldsymbol{y}\right)^{n-k}\right)d\left(\boldsymbol{x}_{1}\cdots\boldsymbol{x}_{k}\right)\\
 & = & \frac{\lambda^{k}}{k!}\int_{\left(A_{\frac{1}{r_{\rho}}}\right)^{k}}g_{1}\left(\boldsymbol{x}_{1},\boldsymbol{x}_{2},\ldots,\boldsymbol{x}_{k}\right)\left(\sum_{n=k}^{\infty}\frac{\left(\lambda\left(\int_{A_{\frac{1}{r_{\rho}}}}1-g_{2}\left(\boldsymbol{y};\boldsymbol{x}_{1},\boldsymbol{x}_{2},\ldots,\boldsymbol{x}_{k}\right)d\boldsymbol{y}\right)\right)^{n-k}}{\left(n-k\right)!}e^{-\lambda A_{\frac{1}{r_{\rho}}}}\right)d\left(\boldsymbol{x}_{1}\cdots\boldsymbol{x}_{k}\right)\\
 & = & \frac{\lambda^{k}}{k!}\int_{\left(A_{\frac{1}{r_{\rho}}}\right)^{k}}g_{1}\left(\boldsymbol{x}_{1},\boldsymbol{x}_{2},\ldots,\boldsymbol{x}_{k}\right)e^{-\lambda\int_{A_{\frac{1}{r_{\rho}}}}g_{2}\left(\boldsymbol{y};\boldsymbol{x}_{1},\boldsymbol{x}_{2},\ldots,\boldsymbol{x}_{k}\right)d\boldsymbol{y}}d\left(\boldsymbol{x}_{1}\cdots\boldsymbol{x}_{k}\right)\end{eqnarray*}

\end{proof}
A similar technique as that used in the proof of Proposition 6.2 in
\cite{Meester96Continuum}, originally due to Penrose \cite{Penrose91On},
was used in the proof of Lemma \ref{lem:expected number of components of size k}
.

The following lemma is also used in the analysis of $E\left(\xi_{k}\right)$.
\begin{lemma}
\label{lem:sufficient and necessary condition ordering}A sufficient
and necessary condition for a given set of nodes to form a single
connected component is that there exists an ordering of the nodes,
which can start from any node in the set, such that each node appearing
later in the order is connected to at least one node appearing earlier
in the order.
\end{lemma}
The proof is trivial and can be omitted. 

Lemma \ref{lem:sufficient and necessary condition ordering} must
have been proved in the literature as it forms the basis of a widely
used algorithm to test network connectivity. However we are unable
to find it.

Using Lemma \ref{lem:sufficient and necessary condition ordering},
the following result can be established:
\begin{lemma}
\label{lem:inequality on g_1}Let $\Gamma_{k}$ denote the set $\left\{ 1,\ldots,k\right\} $.
The function $g_{1}\left(\boldsymbol{x}_{1},\boldsymbol{x}_{2},\ldots,\boldsymbol{x}_{k}\right)$
satisfies the following inequality\begin{eqnarray*}
 &  & g_{1}\left(\boldsymbol{x}_{1},\boldsymbol{x}_{2},\ldots,\boldsymbol{x}_{k}\right)\\
 & \leq & \sum_{i_{2}\in\Gamma_{k}\backslash\left\{ 1\right\} ,i_{3}\in\Gamma_{k}\backslash\left\{ 1,i_{2}\right\} ,\cdots,i_{k}\in\Gamma_{k}\backslash\left\{ 1,i_{2},\ldots,i_{k-1}\right\} }g_{2}\left(\boldsymbol{x}_{i_{2}};\boldsymbol{x}_{1}\right)g_{2}\left(\boldsymbol{x}_{i_{3}};\boldsymbol{x}_{1},\boldsymbol{x}_{i_{2}}\right)\cdots g_{2}\left(\boldsymbol{x}_{i_{k}};\boldsymbol{x}_{1},\boldsymbol{x}_{i_{2}},\ldots,\boldsymbol{x}_{i_{k-1}}\right)\end{eqnarray*}
\end{lemma}
\begin{proof}
Without loss of generality, we assume that such ordering described
in Lemma \ref{lem:sufficient and necessary condition ordering} starts
from $\boldsymbol{x}_{1}\in\left\{ \boldsymbol{x}_{1},\boldsymbol{x}_{2},\ldots,\boldsymbol{x}_{k}\right\} $.
Denote by $\xi_{\left(1,i_{2},\ldots,i_{k}\right)}$ the event that
$\left(\boldsymbol{x}_{1},\boldsymbol{x}_{i_{2}},\ldots,\boldsymbol{x}_{i_{k}}\right)$
is one of such an ordering described in Lemma \ref{lem:sufficient and necessary condition ordering},
where $i_{2}\in\Gamma_{k}\backslash\left\{ 1\right\} ,i_{3}\in\Gamma_{k}\backslash\left\{ 1,i_{2}\right\} ,\cdots,i_{k}\in\Gamma_{k}\backslash\left\{ 1,i_{2},\ldots,i_{k-1}\right\} $.
Using Lemma \ref{lem:sufficient and necessary condition ordering},
it can be shown that

\begin{eqnarray*}
\Pr\left(\xi_{\left(1,i_{2},\ldots,i_{k-1}\right)}\right) & = & g_{2}\left(\boldsymbol{x}_{i_{2}};\boldsymbol{x}_{1}\right)g_{2}\left(\boldsymbol{x}_{i_{3}};\boldsymbol{x}_{1},\boldsymbol{x}_{i_{2}}\right)\cdots g_{2}\left(\boldsymbol{x}_{i_{k}};\boldsymbol{x}_{1},\boldsymbol{x}_{i_{2}},\ldots,\boldsymbol{x}_{i_{k-1}}\right)\end{eqnarray*}

Then it follows that \begin{eqnarray*}
g_{1}\left(\boldsymbol{x}_{1},\boldsymbol{x}_{2},\ldots,\boldsymbol{x}_{k}\right) & = & \Pr\left(\cup_{i_{2}\in\Gamma_{k}\backslash\left\{ 1\right\} ,i_{3}\in\Gamma_{k}\backslash\left\{ 1,i_{2}\right\} ,\cdots,i_{k}\in\Gamma_{k}\backslash\left\{ 1,i_{2},\ldots,i_{k-1}\right\} }\xi_{\left(1,i_{2},\ldots,i_{k}\right)}\right)\end{eqnarray*}

As an easy consequence of the above equation and the union bound:

\begin{eqnarray*}
 &  & g_{1}\left(\boldsymbol{x}_{1},\boldsymbol{x}_{2},\ldots,\boldsymbol{x}_{k}\right)\\
 & \leq & \sum_{i_{2}\in\Gamma_{k}\backslash\left\{ 1\right\} ,i_{3}\in\Gamma_{k}\backslash\left\{ 1,i_{2}\right\} ,\cdots,i_{k}\in\Gamma_{k}\backslash\left\{ 1,i_{2},\ldots,i_{k-1}\right\} }g_{2}\left(\boldsymbol{x}_{i_{2}};\boldsymbol{x}_{1}\right)g_{2}\left(\boldsymbol{x}_{i_{3}};\boldsymbol{x}_{1},\boldsymbol{x}_{i_{2}}\right)\cdots g_{2}\left(\boldsymbol{x}_{i_{k}};\boldsymbol{x}_{1},\boldsymbol{x}_{i_{2}},\ldots,\boldsymbol{x}_{i_{k-1}}\right)\end{eqnarray*}

\end{proof}

The following geometric results are also used in the proof of Theorem
\ref{thm:varnish of finite components}.
\begin{lemma}
\label{lem:analysis of the intersectional area without border}Consider
two points $\boldsymbol{x}_{1},\boldsymbol{x}_{2}\in A_{\frac{1}{r_{\rho}}}$
and let $z\triangleq\left\Vert \boldsymbol{x}_{2}-\boldsymbol{x}_{1}\right\Vert $.
For a positive constant $c_{1}=\sqrt{3}r$ and $z\leq r$\[
\left|D\left(\boldsymbol{x}_{1},r\right)\backslash D\left(\boldsymbol{x}_{2},r\right)\right|\geq c_{1}z\]
where $\left|D\left(\boldsymbol{x}_{1},r\right)\backslash D\left(\boldsymbol{x}_{2},r\right)\right|$
denotes the area of $D\left(\boldsymbol{x}_{1},r\right)\backslash D\left(\boldsymbol{x}_{2},r\right)$.\end{lemma}
\begin{proof}
First it can be shown that for $z\geq2r$\[
\left|D\left(\boldsymbol{x}_{1},r\right)\backslash D\left(\boldsymbol{x}_{2},r\right)\right|=\pi r^{2}\]
and for $z<2r$\begin{eqnarray*}
 &  & f\left(z\right)\\
 & \triangleq & \left|D\left(\boldsymbol{x}_{2},r\right)\backslash D\left(\boldsymbol{x}_{1},r\right)\right|\\
 & = & \pi r^{2}-2r^{2}\arcsin\left(\sqrt{1-\frac{z^{2}}{4r^{2}}}\right)+zr\sqrt{1-\frac{z^{2}}{4r^{2}}}\end{eqnarray*}
Further, it can be shown that\[
\frac{df\left(z\right)}{dz}=2r\sqrt{1-\frac{z^{2}}{4r^{2}}}\]
Therefore $f\left(z\right)$ is an increasing function of $z$ for
$z<2r$ and $\frac{df\left(z\right)}{dz}\geq\sqrt{3}r$ for $z\leq r$.
It then follows from $f\left(0\right)=0$ that $f\left(z\right)\geq\sqrt{3}rz$
for $z\leq r$. \end{proof}
\begin{lemma}
\label{lem:analysis of the intersectional area boundary case}Consider
two points $\boldsymbol{x}_{1}\in\ell A_{\frac{1}{r_{\rho}}}$ and
$\boldsymbol{x}_{2}\in A_{\frac{1}{r_{\rho}}}\cap D\left(\boldsymbol{x}_{1},r\right)$
and let $z\triangleq\left\Vert \boldsymbol{x}_{2}-\boldsymbol{x}_{1}\right\Vert $.
When $\gamma\left(\boldsymbol{x}_{2}\right)\leq\gamma\left(\boldsymbol{x}_{1}\right)$,
\[
\left|A_{\frac{1}{r_{\rho}}}\cap D\left(\boldsymbol{x}_{1},r\right)\backslash D\left(\boldsymbol{x}_{2},r\right)\right|\geq\frac{c_{1}}{2}z\]
When $\gamma\left(\boldsymbol{x}_{2}\right)>\gamma\left(\boldsymbol{x}_{1}\right)$,
for any positive constant $c_{2}$, there exists a positive constant
$z_{0}<r$ such that for all $z\leq z_{0}$

\[
\left|A_{\frac{1}{r_{\rho}}}\cap D\left(\boldsymbol{x}_{1},r\right)\backslash D\left(\boldsymbol{x}_{2},r\right)\right|\geq\left(r-c_{2}\right)z-r\times\left|\gamma\left(\boldsymbol{x}_{2}\right)-\gamma\left(\boldsymbol{x}_{1}\right)\right|\]
where $\gamma\left(\boldsymbol{x}_{1}\right)$ ($\gamma\left(\boldsymbol{x}_{2}\right)$)
represents the shortest Euclidean distance between $\boldsymbol{x}_{1}$
($\boldsymbol{x}_{2}$) and a border of $A_{\frac{1}{r_{\rho}}}$
that is adjacent to $\ell A_{\frac{1}{r_{\rho}}}$ (i.e. $\partial A_{\frac{1}{r_{\rho}}}\cap\ell A_{\frac{1}{r_{\rho}}}$,
see Fig. \ref{fig:An-Illustration-of-the-intersectional-area} for
an illustration of $\gamma\left(\boldsymbol{x}_{2}\right)$ where
$\gamma\left(\boldsymbol{x}_{1}\right)=0$ in the figure).\end{lemma}
\begin{proof}
The first part of the lemma can be easily proved by noting that when
$\gamma\left(\boldsymbol{x}_{2}\right)\leq\gamma\left(\boldsymbol{x}_{1}\right)$
\[
\left|A_{\frac{1}{r_{\rho}}}\cap D\left(\boldsymbol{x}_{1},r\right)\backslash D\left(\boldsymbol{x}_{2},r\right)\right|\geq\frac{1}{2}\left|D\left(\boldsymbol{x}_{1},r\right)\backslash D\left(\boldsymbol{x}_{2},r\right)\right|\]
and the lemma can then be proved using Lemma \ref{lem:analysis of the intersectional area without border}. 

Now let us focus on the situation when $\gamma\left(\boldsymbol{x}_{2}\right)>\gamma\left(\boldsymbol{x}_{1}\right)$.
It can be easily shown (see also Fig. \ref{fig:An-Illustration-of-the-intersectional-area})
that when changing the value of $\gamma\left(\boldsymbol{x}_{1}\right)$
while keeping $\boldsymbol{x}_{2}-\boldsymbol{x}_{1}$ \emph{fixed}
(i.e. $\boldsymbol{x}_{2}$ has the same displacement as $\boldsymbol{x}_{1}$),
$\left|A_{\frac{1}{r_{\rho}}}\cap D\left(\boldsymbol{x}_{1},r\right)\backslash D\left(\boldsymbol{x}_{2},r\right)\right|$
is minimized as $\gamma\left(\boldsymbol{x}_{1}\right)=0$ (i.e. $\boldsymbol{x}_{1}\in\ell A_{\frac{1}{r_{\rho}}}\cap\partial A_{\frac{1}{r_{\rho}}}$)
and $\left(r-c_{2}\right)z-r\times\left|\gamma\left(\boldsymbol{x}_{2}\right)-\gamma\left(\boldsymbol{x}_{1}\right)\right|$
remains constant. Therefore we focus on the worst case when $\boldsymbol{x}_{1}\in\ell A_{\frac{1}{r_{\rho}}}\cap\partial A_{\frac{1}{r_{\rho}}}$.
When $\boldsymbol{x}_{1}\in\ell A_{\frac{1}{r_{\rho}}}\cap\partial A_{\frac{1}{r_{\rho}}}$,
$\left|\gamma\left(\boldsymbol{x}_{2}\right)-\gamma\left(\boldsymbol{x}_{1}\right)\right|=\gamma\left(\boldsymbol{x}_{2}\right)$.

Fig. \ref{fig:An-Illustration-of-the-intersectional-area} shows $A_{\frac{1}{r_{\rho}}}\cap D\left(\boldsymbol{x}_{1},r\right)\backslash D\left(\boldsymbol{x}_{2},r\right)$
for $\boldsymbol{x}_{1}\in\ell A_{\frac{1}{r_{\rho}}}\cap\partial A_{\frac{1}{r_{\rho}}}$
and $\boldsymbol{x}_{2}\in A_{\frac{1}{r_{\rho}}}\cap D\left(\boldsymbol{x}_{1},r\right)$.
It can be shown that under the above conditions for $\boldsymbol{x}_{1}$
and $\boldsymbol{x}_{2}$ (see Fig. \ref{fig:An-Illustration-of-the-intersectional-area}
for definitions of $\alpha$ and $A_{1}$ and some detailed but straightfoward
geometric analysis omitted in the following equation) 

\begin{eqnarray*}
 &  & \left|A_{\rho}\cap D\left(\boldsymbol{x}_{1},r\right)\backslash D\left(\boldsymbol{x}_{2},r\right)\right|\\
 & \geq & \left|A_{1}\right|\\
 & = & h\left(z,\alpha\right)\\
 & \triangleq & \frac{\pi r^{2}}{4}-r^{2}\arccos\frac{z}{2r}+\frac{1}{2}zr\sqrt{1-\frac{z^{2}}{4r^{2}}}+\frac{r^{2}}{2}\arccos\frac{z\cos\alpha}{r}-\frac{1}{2}zr\sqrt{1-\frac{z^{2}}{r^{2}}\cos^{2}\alpha}\cos\alpha+\frac{1}{2}z^{2}\sin\alpha\cos\alpha\end{eqnarray*}
Note that $h\left(0,\alpha\right)=0$, \[
\left.\frac{\partial h\left(z,\alpha\right)}{\partial z}\right|_{z=0}=r\left(1-\cos\alpha\right)\]
and $\cos\alpha=\frac{\gamma\left(\boldsymbol{x}_{2}\right)}{z}$.
Therefore \begin{eqnarray*}
 &  & \lim_{z\rightarrow0^{+}}\frac{h\left(z,\alpha\right)-h\left(0,\alpha\right)}{z}=r\left(1-\cos\alpha\right)\end{eqnarray*}
i.e. for a given positive constant $c_{2}$, there exists $z_{\alpha}>0$
depending on $\alpha$ such that for all $0\leq z\leq z_{\alpha}$\[
h\left(z,\alpha\right)\geq\left(r\left(1-\cos\alpha\right)-c_{2}\right)z\]
The proof is complete by choosing $z_{0}=\min_{0\leq\alpha\leq\frac{\pi}{2}}z_{\alpha}$
and using $\cos\alpha=\frac{\gamma\left(\boldsymbol{x}_{2}\right)}{z}$.

\begin{figure}
\begin{centering}
\includegraphics[width=0.4\columnwidth]{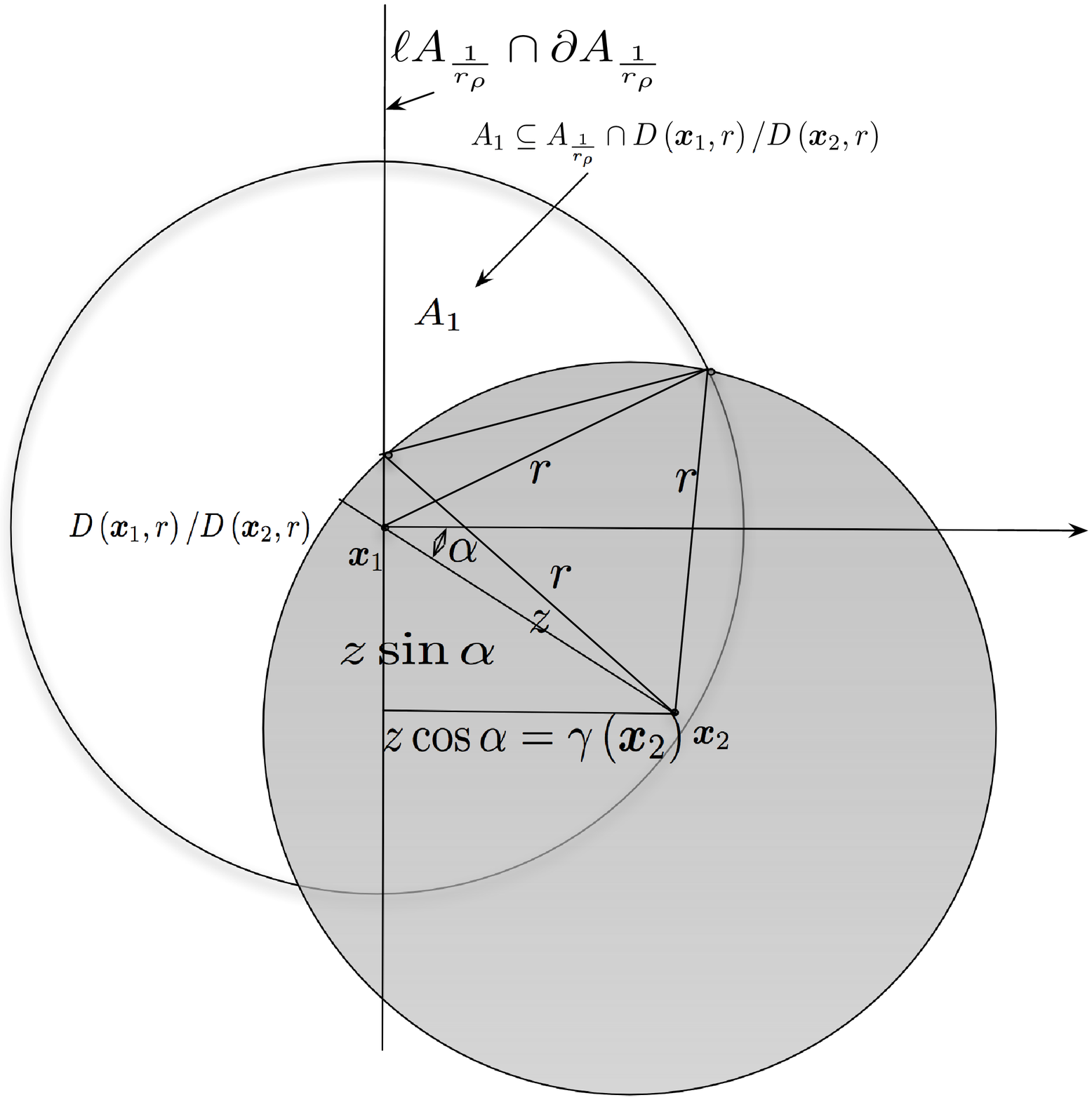}
\par\end{centering}

\caption{\label{fig:An-Illustration-of-the-intersectional-area}An illustration
of $\left|A_{\frac{1}{r_{\rho}}}\cap D\left(\boldsymbol{x}_{2},r\right)/D\left(\boldsymbol{x}_{1},r\right)\right|$
for $\boldsymbol{x}_{1}\in\ell A_{\frac{1}{r_{\rho}}}\cap\partial A_{\frac{1}{r_{\rho}}}$
and $\boldsymbol{x}_{2}\in A_{\frac{1}{r_{\rho}}}\cap D\left(\boldsymbol{x}_{1},r\right)$.
Note that $A_{1}$ is the upper part of $A_{\rho}\cap D\left(\boldsymbol{x}_{1},r\right)\backslash D\left(\boldsymbol{x}_{2},r\right)$
above the line connecting $\boldsymbol{x}_{1}$ and $\boldsymbol{x}_{2}$.
Depending on the relative positions of $\boldsymbol{x}_{1}$ and $\boldsymbol{x}_{2}$,
$A_{\rho}\cap D\left(\boldsymbol{x}_{1},r\right)\backslash D\left(\boldsymbol{x}_{2},r\right)$
may also contain a non-empty region below the line connecting $\boldsymbol{x}_{1}$
and $\boldsymbol{x}_{2}$.}

\end{figure}

\end{proof}
On the basis of the above preliminary results, we are now ready to
start the proof of Theorem \ref{thm:varnish of finite components}.

Let $\delta$ be a positive constant and $\delta\leq\frac{r}{2}$.
First, as a consequence of Lemma \ref{lem:expected number of components of size k},
it can be shown that \begin{eqnarray}
 &  & E\left(\xi_{k}\right)\nonumber \\
 & = & \frac{\lambda^{k}}{k!}\int_{\left(A_{\frac{1}{r_{\rho}}}\right)^{k-1}}\int_{A_{\frac{1}{r_{\rho}}}}g_{1}\left(\boldsymbol{x}_{1},\boldsymbol{x}_{2},\ldots,\boldsymbol{x}_{k}\right)e^{-\lambda\int_{A_{\frac{1}{r_{\rho}}}}g_{2}\left(\boldsymbol{y};\boldsymbol{x}_{1},\boldsymbol{x}_{2},\ldots,\boldsymbol{x}_{k}\right)d\boldsymbol{y}}d\boldsymbol{x}_{1}d\left(\boldsymbol{x}_{2}\cdots\boldsymbol{x}_{k}\right)\nonumber \\
 & = & \frac{\lambda^{k}}{k!}\int_{A_{\frac{1}{r_{\rho}}}}\int_{\left(A_{\frac{1}{r_{\rho}}}\right)^{k-1}\cap\left(D\left(\boldsymbol{x}_{1},\delta\right)\right)^{k-1}}g_{1}\left(\boldsymbol{x}_{1},\boldsymbol{x}_{2},\ldots,\boldsymbol{x}_{k}\right)e^{-\lambda\int_{A_{\frac{1}{r_{\rho}}}}g_{2}\left(\boldsymbol{y};\boldsymbol{x}_{1},\boldsymbol{x}_{2},\ldots,\boldsymbol{x}_{k}\right)d\boldsymbol{y}}d\left(\boldsymbol{x}_{2}\cdots\boldsymbol{x}_{k}\right)d\boldsymbol{x}_{1}\nonumber \\
 & + & \frac{\lambda^{k}}{k!}\int_{A_{\frac{1}{r_{\rho}}}}\int_{\left(A_{\frac{1}{r_{\rho}}}\right)^{k-1}\backslash\left(D\left(\boldsymbol{x}_{1},\delta\right)\right)^{k-1}}g_{1}\left(\boldsymbol{x}_{1},\boldsymbol{x}_{2},\ldots,\boldsymbol{x}_{k}\right)e^{-\lambda\int_{A_{\frac{1}{r_{\rho}}}}g_{2}\left(\boldsymbol{y};\boldsymbol{x}_{1},\boldsymbol{x}_{2},\ldots,\boldsymbol{x}_{k}\right)d\boldsymbol{y}}d\left(\boldsymbol{x}_{2}\cdots\boldsymbol{x}_{k}\right)d\boldsymbol{x}_{1}\label{eq:analysis on the expected number of components size k step 1}\end{eqnarray}

Denote by $E\left(\xi_{k,1}\right)$ and $E\left(\xi_{k,2}\right)$
the two summands in \eqref{eq:analysis on the expected number of components size k step 1}
respectively. In the following analysis, we will show that by choosing
$\delta$ to be sufficiently small, $\lim_{\rho\rightarrow\infty}\sum_{k=2}^{\infty}E\left(\xi_{k.1}\right)=0$
and $\lim_{\rho\rightarrow\infty}\sum_{k=2}^{M}E\left(\xi_{k,2}\right)=0$.
$\sum_{k=2}^{\infty}E\left(\xi_{k.1}\right)$ has the meaning of being
the expected total number of components of finite orders $\infty>k>1$,
where all other nodes of the component are located within a $\delta$
neighbourhood of a randomly designated node (i.e. $\boldsymbol{x}_{1}$
in \eqref{eq:analysis on the expected number of components size k step 1}).
$\lim_{\rho\rightarrow\infty}\sum_{k=2}^{\infty}E\left(\xi_{k.1}\right)=0$
implies $\lim_{\rho\rightarrow\infty}\sum_{k=2}^{M}E\left(\xi_{k.1}\right)=0$.
$\sum_{k=2}^{M}E\left(\xi_{k,2}\right)=0$ has the meaning of being
the expected total number of components of finite orders $M\geq k>1$
where at least one of the nodes forming the component is located outside
a $\delta$ neighbourhood of a randomly designated node (i.e. $\boldsymbol{x}_{1}$
in \eqref{eq:analysis on the expected number of components size k step 1})
in the component.

\subsection*{An analysis of the first term in \eqref{eq:analysis on the expected number of components size k step 1}}

Denote by $D_{\delta}^{i}\subset\left(A_{\frac{1}{r_{\rho}}}\right)^{k-1}$
the set $\left\{ \left(\boldsymbol{x}_{2},\ldots,\boldsymbol{x}_{k}\right)\in\left(A_{\frac{1}{r_{\rho}}}\right)^{k-1}\cap\left(D\left(\boldsymbol{x}_{1},\delta\right)\right)^{k-1}:\left\Vert \boldsymbol{x}_{i}-\boldsymbol{x}_{1}\right\Vert \geq\max_{j\in\left\{ 2,\ldots k\right\} ,j\neq i}\left\Vert \boldsymbol{x}_{j}-\boldsymbol{x}_{1}\right\Vert \right\} $,
$i\in\left\{ 2,\ldots k\right\} $. Using \eqref{eq:inequality for g_2}
and the definition of $D_{\delta}^{i}$, it can be shown that

\begin{eqnarray}
 &  & E\left(\xi_{k,1}\right)\nonumber \\
 & \triangleq & \frac{\lambda^{k}}{k!}\int_{A_{\rho}}\int_{\left(A_{\frac{1}{r_{\rho}}}\right)^{k-1}\cap\left(D\left(\boldsymbol{x}_{1},\delta\right)\right)^{k-1}}g_{1}\left(\boldsymbol{x}_{1},\boldsymbol{x}_{2},\ldots,\boldsymbol{x}_{k}\right)e^{-\lambda\int_{A_{\frac{1}{r_{\rho}}}}g_{2}\left(\boldsymbol{y};\boldsymbol{x}_{1},\boldsymbol{x}_{2},\ldots,\boldsymbol{x}_{k}\right)d\boldsymbol{y}}d\left(\boldsymbol{x}_{2}\cdots\boldsymbol{x}_{k}\right)d\boldsymbol{x}_{1}\nonumber \\
 & = & \sum_{i=2}^{k}\frac{\lambda^{k}}{k!}\int_{A_{\frac{1}{r_{\rho}}}}\int_{D_{\delta}^{i}}g_{1}\left(\boldsymbol{x}_{1},\boldsymbol{x}_{2},\ldots,\boldsymbol{x}_{k}\right)e^{-\lambda\int_{A_{\frac{1}{r_{\rho}}}}g_{2}\left(\boldsymbol{y};\boldsymbol{x}_{1},\boldsymbol{x}_{2},\ldots,\boldsymbol{x}_{k}\right)d\boldsymbol{y}}d\left(\boldsymbol{x}_{2}\cdots\boldsymbol{x}_{k}\right)d\boldsymbol{x}_{1}\nonumber \\
 & \leq & \frac{\lambda^{k}}{\left(k-2\right)!k}\int_{A_{\frac{1}{r_{\rho}}}}\int_{D_{\delta}^{2}}g_{1}\left(\boldsymbol{x}_{1},\boldsymbol{x}_{2},\ldots,\boldsymbol{x}_{k}\right)e^{-\lambda\int_{A_{\frac{1}{r_{\rho}}}}g_{2}\left(\boldsymbol{y};\boldsymbol{x}_{1},\boldsymbol{x}_{2}\right)d\boldsymbol{y}}d\left(\boldsymbol{x}_{2}\cdots\boldsymbol{x}_{k}\right)d\boldsymbol{x}_{1}\nonumber \\
 & \leq & \frac{\lambda^{k}}{\left(k-2\right)!k}\int_{A_{\frac{1}{r_{\rho}}}}\int_{D_{\delta}^{2}}e^{-\lambda\int_{A_{\frac{1}{r_{\rho}}}}g_{2}\left(\boldsymbol{y};\boldsymbol{x}_{1},\boldsymbol{x}_{2}\right)d\boldsymbol{y}}d\left(\boldsymbol{x}_{2}\cdots\boldsymbol{x}_{k}\right)d\boldsymbol{x}_{1}\nonumber \\
 & \leq & \frac{\lambda^{k}}{\left(k-2\right)!k}\int_{A_{\frac{1}{r_{\rho}}}}\int_{A_{\frac{1}{r_{\rho}}}\cap D\left(\boldsymbol{x}_{1},\delta\right)}\left(\pi\left\Vert \boldsymbol{x}_{2}-\boldsymbol{x}_{1}\right\Vert ^{2}\right)^{k-2}e^{-\lambda\int_{A_{\frac{1}{r_{\rho}}}}g_{2}\left(\boldsymbol{y};\boldsymbol{x}_{1},\boldsymbol{x}_{2}\right)d\boldsymbol{y}}d\boldsymbol{x}_{2}d\boldsymbol{x}_{1}\label{eq:analysis on the expected number of components size k step 2 finite area}\end{eqnarray}
As a result of the following inequality: \begin{eqnarray*}
 &  & \sum_{k=2}^{\infty}\frac{\lambda^{k}\left(\pi\left\Vert \boldsymbol{x}_{2}-\boldsymbol{x}_{1}\right\Vert ^{2}\right)^{k-2}}{\left(k-2\right)!k}\\
 & = & \lambda^{2}\left(\sum_{k=0}^{\infty}\frac{\lambda^{k}\left(\pi\left\Vert \boldsymbol{x}_{2}-\boldsymbol{x}_{1}\right\Vert ^{2}\right)^{k}}{k!\left(k+2\right)}\right)\\
 & \leq & \lambda^{2}\left(\sum_{k=0}^{\infty}\frac{\lambda^{k}\left(\pi\left\Vert \boldsymbol{x}_{2}-\boldsymbol{x}_{1}\right\Vert ^{2}\right)^{k}}{k!}e^{-\lambda\pi\left\Vert \boldsymbol{x}_{2}-\boldsymbol{x}_{1}\right\Vert ^{2}}\right)e^{\lambda\pi\left\Vert \boldsymbol{x}_{2}-\boldsymbol{x}_{1}\right\Vert ^{2}}\\
 & = & \lambda^{2}e^{\lambda\pi\left\Vert \boldsymbol{x}_{2}-\boldsymbol{x}_{1}\right\Vert ^{2}}\end{eqnarray*}
it follows from \eqref{eq:analysis on the expected number of components size k step 2 finite area}
that \begin{eqnarray}
 &  & \sum_{k=2}^{\infty}E\left(\xi_{k,1}\right)\nonumber \\
 & \leq & \lambda^{2}\int_{A_{\frac{1}{r_{\rho}}}}\int_{A_{\frac{1}{r_{\rho}}}\cap D\left(\boldsymbol{x}_{1},\delta\right)}e^{-\lambda\left(\int_{A_{\frac{1}{r_{\rho}}}}g_{2}\left(\boldsymbol{y};\boldsymbol{x}_{1},\boldsymbol{x}_{2}\right)d\boldsymbol{y}-\pi\left\Vert \boldsymbol{x}_{2}-\boldsymbol{x}_{1}\right\Vert ^{2}\right)}d\boldsymbol{x}_{2}d\boldsymbol{x}_{1}\nonumber \\
 & = & \lambda^{2}\int_{A_{\frac{1}{r_{\rho}}}}\int_{A_{\frac{1}{r_{\rho}}}\cap D\left(\boldsymbol{x}_{1},\delta\right)}e^{-\lambda\left(\int_{A_{\frac{1}{r_{\rho}}}}g\left(\left\Vert \boldsymbol{y}-\boldsymbol{x}_{2}\right\Vert \right)d\boldsymbol{y}+\int_{A_{\frac{1}{r_{\rho}}}}g\left(\left\Vert \boldsymbol{y}-\boldsymbol{x}_{1}\right\Vert \right)\left(1-g\left(\left\Vert \boldsymbol{y}-\boldsymbol{x}_{2}\right\Vert \right)\right)d\boldsymbol{y}-\pi\left\Vert \boldsymbol{x}_{2}-\boldsymbol{x}_{1}\right\Vert ^{2}\right)}d\boldsymbol{x}_{2}d\boldsymbol{x}_{1}\nonumber \\
 & \leq & \lambda^{2}\int_{A_{\frac{1}{r_{\rho}}}}\int_{A_{\frac{1}{r_{\rho}}}\cap D\left(\boldsymbol{x}_{1},\delta\right)}e^{-\lambda\left(\int_{A_{\frac{1}{r_{\rho}}}}g\left(\left\Vert \boldsymbol{y}-\boldsymbol{x}_{2}\right\Vert \right)d\boldsymbol{y}+\int_{A_{\frac{1}{r_{\rho}}}\cap D\left(\boldsymbol{x}_{1},r\right)\backslash D\left(\boldsymbol{x}_{2},r\right)}g\left(\left\Vert \boldsymbol{y}-\boldsymbol{x}_{1}\right\Vert \right)\left(1-g\left(\left\Vert \boldsymbol{y}-\boldsymbol{x}_{2}\right\Vert \right)\right)d\boldsymbol{y}-\pi\left\Vert \boldsymbol{x}_{2}-\boldsymbol{x}_{1}\right\Vert ^{2}\right)}d\boldsymbol{x}_{2}d\boldsymbol{x}_{1}\label{eq:number of finite components in finite area step 1}\\
 & \leq & \lambda^{2}\int_{A_{\frac{1}{r_{\rho}}}}\int_{A_{\frac{1}{r_{\rho}}}\cap D\left(\boldsymbol{x}_{1},\delta\right)}e^{-\lambda\left(\int_{A_{\frac{1}{r_{\rho}}}}g\left(\left\Vert \boldsymbol{y}-\boldsymbol{x}_{2}\right\Vert \right)d\boldsymbol{y}+g\left(r^{-}\right)\left(1-g\left(r^{+}\right)\right)\left|A_{\frac{1}{r_{\rho}}}\cap D\left(\boldsymbol{x}_{1},r\right)\backslash D\left(\boldsymbol{x}_{2},r\right)\right|-\pi\left\Vert \boldsymbol{x}_{2}-\boldsymbol{x}_{1}\right\Vert ^{2}\right)}d\boldsymbol{x}_{2}d\boldsymbol{x}_{1}\label{eq:number of finite components in finite area step 2}\end{eqnarray}
where in \eqref{eq:number of finite components in finite area step 1}
the parameter $r>0$ is chosen such that $g\left(r^{-}\right)\left(1-g\left(r^{+}\right)\right)>0$.
For convenience, use $\beta$ for $g\left(r^{-}\right)\left(1-g\left(r^{+}\right)\right)$
as defined in \eqref{eq:definition of beta}. It follows from \eqref{eq:number of finite components in finite area step 2}
that

\begin{eqnarray}
 &  & \sum_{k=2}^{\infty}E\left(\xi_{k,1}\right)\nonumber \\
 & \leq & \lambda^{2}\int_{B_{r}\left(A_{\frac{1}{r_{\rho}}}\right)}\int_{A_{\frac{1}{r_{\rho}}}\cap D\left(\boldsymbol{x}_{1},\delta\right)}e^{-\lambda\left(\int_{A_{\frac{1}{r_{\rho}}}}g\left(\left\Vert \boldsymbol{y}-\boldsymbol{x}_{2}\right\Vert \right)d\boldsymbol{y}+\beta\left|A_{\frac{1}{r_{\rho}}}\cap D\left(\boldsymbol{x}_{1},r\right)\backslash D\left(\boldsymbol{x}_{2},r\right)\right|-\pi\left\Vert \boldsymbol{x}_{2}-\boldsymbol{x}_{1}\right\Vert ^{2}\right)}d\boldsymbol{x}_{2}d\boldsymbol{x}_{1}\nonumber \\
 & + & \lambda^{2}\int_{A_{\frac{1}{r_{\rho}}}\backslash B_{r}\left(A_{\frac{1}{r_{\rho}}}\right)}\int_{A_{\frac{1}{r_{\rho}}}\cap D\left(\boldsymbol{x}_{1},\delta\right)}e^{-\lambda\left(\int_{A_{\frac{1}{r_{\rho}}}}g\left(\left\Vert \boldsymbol{y}-\boldsymbol{x}_{2}\right\Vert \right)d\boldsymbol{y}+\beta\left|A_{\frac{1}{r_{\rho}}}\cap D\left(\boldsymbol{x}_{1},r\right)\backslash D\left(\boldsymbol{x}_{2},r\right)\right|-\pi\left\Vert \boldsymbol{x}_{2}-\boldsymbol{x}_{1}\right\Vert ^{2}\right)}d\boldsymbol{x}_{2}d\boldsymbol{x}_{1}\label{eq:number of finite components in finite area step 3}\end{eqnarray}

For the first summand in the above equation, it can be shown that
\begin{eqnarray}
 &  & \lambda^{2}\int_{B_{r}\left(A_{\frac{1}{r_{\rho}}}\right)}\int_{A_{\frac{1}{r_{\rho}}}\cap D\left(\boldsymbol{x}_{1},\delta\right)}e^{-\lambda\left(\int_{A_{\frac{1}{r_{\rho}}}}g\left(\left\Vert \boldsymbol{y}-\boldsymbol{x}_{2}\right\Vert \right)d\boldsymbol{y}+\beta\left|A_{\frac{1}{r_{\rho}}}\cap D\left(\boldsymbol{x}_{1},r\right)\backslash D\left(\boldsymbol{x}_{2},r\right)\right|-\pi\left\Vert \boldsymbol{x}_{2}-\boldsymbol{x}_{1}\right\Vert ^{2}\right)}d\boldsymbol{x}_{2}d\boldsymbol{x}_{1}\nonumber \\
 & = & 4\lambda^{2}\int_{\ell A_{\frac{1}{r_{\rho}}}}\int_{A_{\frac{1}{r_{\rho}}}\cap D\left(\boldsymbol{x}_{1},\delta\right)}e^{-\lambda\left(\int_{A_{\frac{1}{r_{\rho}}}}g\left(\left\Vert \boldsymbol{y}-\boldsymbol{x}_{2}\right\Vert \right)d\boldsymbol{y}+\beta\left|A_{\frac{1}{r_{\rho}}}\cap D\left(\boldsymbol{x}_{1},r\right)\backslash D\left(\boldsymbol{x}_{2},r\right)\right|-\pi\left\Vert \boldsymbol{x}_{2}-\boldsymbol{x}_{1}\right\Vert ^{2}\right)}d\boldsymbol{x}_{2}d\boldsymbol{x}_{1}\nonumber \\
 & + & 4\lambda^{2}\int_{\angle A_{\frac{1}{r_{\rho}}}}\int_{A_{\frac{1}{r_{\rho}}}\cap D\left(\boldsymbol{x}_{1},\delta\right)}e^{-\lambda\left(\int_{A_{\frac{1}{r_{\rho}}}}g\left(\left\Vert \boldsymbol{y}-\boldsymbol{x}_{2}\right\Vert \right)d\boldsymbol{y}+\beta\left|A_{\frac{1}{r_{\rho}}}\cap D\left(\boldsymbol{x}_{1},r\right)\backslash D\left(\boldsymbol{x}_{2},r\right)\right|-\pi\left\Vert \boldsymbol{x}_{2}-\boldsymbol{x}_{1}\right\Vert ^{2}\right)}d\boldsymbol{x}_{2}d\boldsymbol{x}_{1}\label{eq:number of finite components in finite area step 4}\end{eqnarray}
Denote by $\gamma\left(\boldsymbol{x}\right)$ the shortest Euclidean
distance between a point $\boldsymbol{x}\in\ell_{r+\delta}A_{\frac{1}{r_{\rho}}}$
and a border of $A_{\frac{1}{r_{\rho}}}$ adjacent to $\ell_{r+\delta}A_{\frac{1}{r_{\rho}}}$
(i.e. $\partial A_{\frac{1}{r_{\rho}}}\cap\ell_{r+\delta}A_{\frac{1}{r_{\rho}}}$),
where $\ell_{r+\delta}A_{\frac{1}{r_{\rho}}}$ denotes a boundary
rectangular area of size $\left(r_{\rho}^{-1}-2\left(r-\delta\right)\right)\times\left(r+\delta\right)$
within $r+\delta$ of the border of $A_{\frac{1}{r_{\rho}}}$ and
away from the corners of $A_{\frac{1}{r_{\rho}}}$ by at least $r-\delta$.
Denote by $B_{\gamma\left(\boldsymbol{x}_{2}\right)}\left(A_{\frac{1}{r_{\rho}}}\right)\subset A_{\rho}$
a boundary area $\left\{ \boldsymbol{x}\in\ell_{r+\delta}A_{\frac{1}{r_{\rho}}}:\gamma\left(\boldsymbol{x}\right)\leq\gamma\left(\boldsymbol{x}_{2}\right)\right\} $,
denote by $R\left(\boldsymbol{x}_{2},2r\right)$ a rectangular area
of size $2r\times\gamma\left(\boldsymbol{x}_{2}\right)$ located between
$\boldsymbol{x}_{2}$ and $\partial A_{\frac{1}{r_{\rho}}}$ with
$\boldsymbol{x}_{2}$ at the center of one side of $R\left(\boldsymbol{x}_{2},2r\right)$.
See Fig. \ref{fig:An-illustration-of-areas-finite-component-analysis}
for an illustration of the areas defined above.

\begin{figure}
\begin{centering}
\includegraphics[width=0.3\columnwidth]{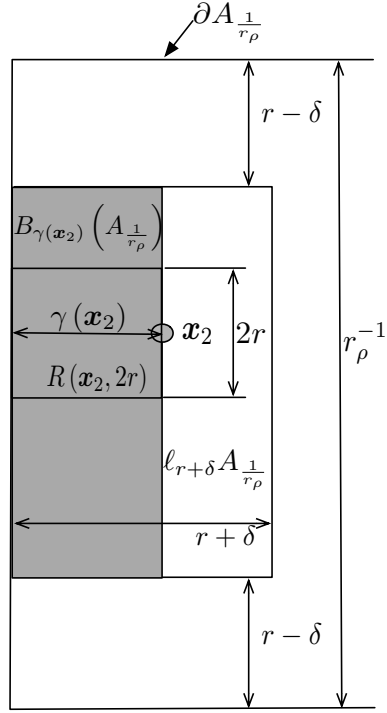}
\par\end{centering}

\caption{An illustration of the areas $\ell_{r+\delta}A_{\frac{1}{r_{\rho}}}$,
$B_{\gamma\left(\boldsymbol{x}_{2}\right)}\left(A_{\frac{1}{r_{\rho}}}\right)$
and $R\left(\boldsymbol{x}_{2},2r\right)$. The shaded area is $B_{\gamma\left(\boldsymbol{x}_{2}\right)}\left(A_{\frac{1}{r_{\rho}}}\right)$.\label{fig:An-illustration-of-areas-finite-component-analysis}}

\end{figure}

First we evaluate the term: $\int_{B_{\gamma\left(\boldsymbol{x}_{2}\right)}\left(A_{\frac{1}{r_{\rho}}}\right)}g\left(\left\Vert \boldsymbol{y}-\boldsymbol{x}_{2}\right\Vert \right)d\boldsymbol{y}+\beta\left|A_{\frac{1}{r_{\rho}}}\cap D\left(\boldsymbol{x}_{1},r\right)\backslash D\left(\boldsymbol{x}_{2},r\right)\right|$.
It can be shown that for $\boldsymbol{x}_{1}\in\ell A_{\frac{1}{r_{\rho}}}$
and $\boldsymbol{x}_{2}\in A_{\frac{1}{r_{\rho}}}\cap D\left(\boldsymbol{x}_{1},\delta\right)$,
when $\gamma\left(\boldsymbol{x}_{2}\right)\geq\gamma\left(\boldsymbol{x}_{1}\right)$:
\begin{eqnarray}
 &  & \int_{B_{\gamma\left(\boldsymbol{x}_{2}\right)}\left(A_{\frac{1}{r_{\rho}}}\right)}g\left(\left\Vert \boldsymbol{y}-\boldsymbol{x}_{2}\right\Vert \right)d\boldsymbol{y}\nonumber \\
 & \geq & \int_{B_{\gamma\left(\boldsymbol{x}_{2}\right)}\left(A_{\frac{1}{r_{\rho}}}\right)\cap R\left(\boldsymbol{x}_{2},2r\right)}g\left(\left\Vert \boldsymbol{y}-\boldsymbol{x}_{2}\right\Vert \right)d\boldsymbol{y}\nonumber \\
 & \geq & \int_{0}^{r}\int_{0}^{\gamma\left(\boldsymbol{x}_{2}\right)}g\left(\sqrt{x^{2}+y^{2}}\right)dxdy\nonumber \\
 & \geq & \int_{0}^{r}\int_{0}^{\left|\gamma\left(\boldsymbol{x}_{2}\right)-\gamma\left(\boldsymbol{x}_{1}\right)\right|}g\left(\sqrt{x^{2}+y^{2}}\right)dxdy\nonumber \\
 & \geq & \int_{0}^{r}\int_{0}^{\left|\gamma\left(\boldsymbol{x}_{2}\right)-\gamma\left(\boldsymbol{x}_{1}\right)\right|}g\left(\sqrt{\left(\frac{r}{2}\right)^{2}+y^{2}}\right)dxdy\label{eq:intermediate step on a bound step 1}\\
 & = & c_{3}\left|\gamma\left(\boldsymbol{x}_{2}\right)-\gamma\left(\boldsymbol{x}_{1}\right)\right|\nonumber \end{eqnarray}
where in \eqref{eq:intermediate step on a bound step 1}, the non-increasing
monotonicity condition on $g$ and that $\left|\gamma\left(\boldsymbol{x}_{2}\right)-\gamma\left(\boldsymbol{x}_{1}\right)\right|\leq\left\Vert \boldsymbol{x}_{2}-\boldsymbol{x}_{1}\right\Vert \leq\delta\leq\frac{r}{2}$
is used, and $c_{3}\triangleq\int_{0}^{r}g\left(\sqrt{\left(\frac{r}{2}\right)^{2}+y^{2}}\right)dy$.
Since $g\left(r^{-}\right)\left(1-g\left(r^{+}\right)\right)>0$,
it follows from the non-increasing monotonicity condition on $g$
that $g\left(\frac{r}{2}\right)>0$ and $ $$c_{3}$ is a positive
constant, i.e. $c_{3}>0$.

Choose $c_{2}$ to be sufficiently small such that $c_{4}\triangleq\frac{c_{3}}{\beta}-c_{2}>0$
and choose $\delta$ to be sufficiently small such that $\delta\leq z_{0}$.
Using \eqref{eq:intermediate step on a bound step 1} and Lemma \ref{lem:analysis of the intersectional area boundary case},
it follows that

\begin{eqnarray*}
 &  & \int_{B_{\gamma\left(\boldsymbol{x}_{2}\right)}\left(A_{\frac{1}{r_{\rho}}}\right)}g\left(\left\Vert \boldsymbol{y}-\boldsymbol{x}_{2}\right\Vert \right)d\boldsymbol{y}+\beta\left|A_{\frac{1}{r_{\rho}}}\cap D\left(\boldsymbol{x}_{1},r\right)\backslash D\left(\boldsymbol{x}_{2},r\right)\right|\\
 & \geq & c_{3}\left|\gamma\left(\boldsymbol{x}_{2}\right)-\gamma\left(\boldsymbol{x}_{1}\right)\right|+\beta\left(\left(r-c_{2}\right)\left\Vert \boldsymbol{x}_{2}-\boldsymbol{x}_{1}\right\Vert -r\times\left|\gamma\left(\boldsymbol{x}_{2}\right)-\gamma\left(\boldsymbol{x}_{1}\right)\right|\right)\\
 & = & \beta\left(\left(r-c_{2}\right)-\left(r-\frac{c_{3}}{\beta}\right)\times\frac{\left|\gamma\left(\boldsymbol{x}_{2}\right)-\gamma\left(\boldsymbol{x}_{1}\right)\right|}{\left\Vert \boldsymbol{x}_{2}-\boldsymbol{x}_{1}\right\Vert }\right)\left\Vert \boldsymbol{x}_{2}-\boldsymbol{x}_{1}\right\Vert \end{eqnarray*}
Note that $\frac{\left|\gamma\left(\boldsymbol{x}_{2}\right)-\gamma\left(\boldsymbol{x}_{1}\right)\right|}{\left\Vert \boldsymbol{x}_{2}-\boldsymbol{x}_{1}\right\Vert }\leq1$,
therefore\[
\left(r-c_{2}\right)-\left(r-\frac{c_{3}}{\beta}\right)\times\frac{\left|\gamma\left(\boldsymbol{x}_{2}\right)-\gamma\left(\boldsymbol{x}_{1}\right)\right|}{\left\Vert \boldsymbol{x}_{2}-\boldsymbol{x}_{1}\right\Vert }\geq\frac{c_{3}}{\beta}-c_{2}\]
 \begin{eqnarray}
 &  & \int_{B_{\gamma\left(\boldsymbol{x}_{2}\right)}\left(A_{\frac{1}{r_{\rho}}}\right)}g\left(\left\Vert \boldsymbol{y}-\boldsymbol{x}_{2}\right\Vert \right)d\boldsymbol{y}+\beta\left|A_{\frac{1}{r_{\rho}}}\cap D\left(\boldsymbol{x}_{1},r\right)\backslash D\left(\boldsymbol{x}_{2},r\right)\right|\nonumber \\
 & \geq & \beta c_{4}\left\Vert \boldsymbol{x}_{2}-\boldsymbol{x}_{1}\right\Vert \label{eq:inequalitys when x_2>x_1}\end{eqnarray}

When $\gamma\left(\boldsymbol{x}_{2}\right)<\gamma\left(\boldsymbol{x}_{1}\right)$,
using Lemma \ref{lem:analysis of the intersectional area boundary case}

\begin{eqnarray}
 &  & \int_{B_{\gamma\left(\boldsymbol{x}_{2}\right)}\left(A_{\frac{1}{r_{\rho}}}\right)}g\left(\left\Vert \boldsymbol{y}-\boldsymbol{x}_{2}\right\Vert \right)d\boldsymbol{y}+\beta\left|A_{\frac{1}{r_{\rho}}}\cap D\left(\boldsymbol{x}_{1},r\right)\backslash D\left(\boldsymbol{x}_{2},r\right)\right|\nonumber \\
 & \geq & \beta\left|A_{\frac{1}{r_{\rho}}}\cap D\left(\boldsymbol{x}_{1},r\right)\backslash D\left(\boldsymbol{x}_{2},r\right)\right|\nonumber \\
 & \geq & \beta\frac{c_{1}}{2}\left\Vert \boldsymbol{x}_{2}-\boldsymbol{x}_{1}\right\Vert \label{eq:inequalitys when x_1>x_2}\end{eqnarray}

Let $c_{5}\triangleq\min\left\{ \frac{c_{1}}{2},c_{4}\right\} $.
It follows from \eqref{eq:inequalitys when x_2>x_1} and \eqref{eq:inequalitys when x_1>x_2}
that \begin{eqnarray*}
 &  & \int_{B_{\gamma\left(\boldsymbol{x}_{2}\right)}\left(A_{\frac{1}{r_{\rho}}}\right)}g\left(\left\Vert \boldsymbol{y}-\boldsymbol{x}_{2}\right\Vert \right)d\boldsymbol{y}+\beta\left|A_{\frac{1}{r_{\rho}}}\cap D\left(\boldsymbol{x}_{1},r\right)\backslash D\left(\boldsymbol{x}_{2},r\right)\right|\\
 & \geq & \beta c_{5}\left\Vert \boldsymbol{x}_{2}-\boldsymbol{x}_{1}\right\Vert \end{eqnarray*}

Choose $\delta$ to be sufficiently small such that $\pi\delta\leq\frac{1}{2}c_{5}\beta$
and also $\delta\leq z_{0}$. Note also that for $\boldsymbol{x}_{2}\in A_{\frac{1}{r_{\rho}}}\cap D\left(\boldsymbol{x}_{1},\delta\right)$,
$\left\Vert \boldsymbol{x}_{2}-\boldsymbol{x}_{1}\right\Vert \leq\delta$.
Then it follows that \begin{eqnarray*}
 &  & 4\lambda^{2}\int_{\ell A_{\frac{1}{r_{\rho}}}}\int_{A_{\frac{1}{r_{\rho}}}\cap D\left(\boldsymbol{x}_{1},\delta\right)}e^{-\lambda\left(\int_{A_{\frac{1}{r_{\rho}}}}g\left(\left\Vert \boldsymbol{y}-\boldsymbol{x}_{2}\right\Vert \right)d\boldsymbol{y}+\beta\left|A_{\frac{1}{r_{\rho}}}\cap D\left(\boldsymbol{x}_{1},r\right)\backslash D\left(\boldsymbol{x}_{2},r\right)\right|-\pi\left\Vert \boldsymbol{x}_{2}-\boldsymbol{x}_{1}\right\Vert ^{2}\right)}d\boldsymbol{x}_{2}d\boldsymbol{x}_{1}\\
 & \leq & 4\lambda^{2}\int_{\ell A_{\frac{1}{r_{\rho}}}}\int_{A_{\frac{1}{r_{\rho}}}\cap D\left(\boldsymbol{x}_{1},\delta\right)}e^{-\lambda\left(\int_{A_{\frac{1}{r_{\rho}}}\backslash B_{\gamma\left(\boldsymbol{x}_{2}\right)}\left(A_{\frac{1}{r_{\rho}}}\right)}g\left(\left\Vert \boldsymbol{y}-\boldsymbol{x}_{2}\right\Vert \right)d\boldsymbol{y}+\frac{1}{2}c_{5}\beta\left\Vert \boldsymbol{x}_{2}-\boldsymbol{x}_{1}\right\Vert \right)}d\boldsymbol{x}_{2}d\boldsymbol{x}_{1}\\
 & \leq & 4\lambda^{2}\int_{\ell_{r+\delta}A_{\frac{1}{r_{\rho}}}}\int_{A_{\frac{1}{r_{\rho}}}\cap D\left(\boldsymbol{x}_{2},\delta\right)}e^{-\lambda\left(\int_{A_{\frac{1}{r_{\rho}}}\backslash B_{\gamma\left(\boldsymbol{x}_{2}\right)}\left(A_{\frac{1}{r_{\rho}}}\right)}g\left(\left\Vert \boldsymbol{y}-\boldsymbol{x}_{2}\right\Vert \right)d\boldsymbol{y}+\frac{1}{2}c_{5}\beta\left\Vert \boldsymbol{x}_{2}-\boldsymbol{x}_{1}\right\Vert \right)}d\boldsymbol{x}_{1}d\boldsymbol{x}_{2}\\
 & \leq & 4\lambda^{2}\int_{0}^{\delta}e^{-\lambda\frac{1}{2}c_{5}\beta x}2\pi xdx\int_{\ell_{r+\delta}A_{\frac{1}{r_{\rho}}}}e^{-\lambda\int_{A_{\frac{1}{r_{\rho}}}\backslash B_{\gamma\left(\boldsymbol{x}_{2}\right)}\left(A_{\frac{1}{r_{\rho}}}\right)}g\left(\left\Vert \boldsymbol{y}-\boldsymbol{x}_{2}\right\Vert \right)d\boldsymbol{y}}d\boldsymbol{x}_{2}\\
 & = & 32\pi\frac{1-e^{-\lambda\frac{1}{2}c_{5}\beta\delta}\left(1+\lambda\frac{1}{2}c_{5}\beta\delta\right)}{\left(c_{5}\beta\right)^{2}}\int_{\ell_{r+\delta}A_{\frac{1}{r_{\rho}}}}e^{-\lambda\int_{A_{\frac{1}{r_{\rho}}}\backslash B_{\gamma\left(\boldsymbol{x}_{2}\right)}\left(A_{\frac{1}{r_{\rho}}}\right)}g\left(\left\Vert \boldsymbol{y}-\boldsymbol{x}_{2}\right\Vert \right)d\boldsymbol{y}}d\boldsymbol{x}_{2}\end{eqnarray*}
 We further divide $\ell_{r+\delta}A_{\frac{1}{r_{\rho}}}$ into two
parts: one rectangular area of size $\left(r_{\rho}^{-1}-2r_{\rho}^{-\varepsilon}\right)\times\left(r+\delta\right)$
in the center of $\ell_{r+\delta}A_{\frac{1}{r_{\rho}}}$, denoted
by $\ell_{r+\delta}^{1}A_{\frac{1}{r_{\rho}}}$, and the other area
$\ell_{r+\delta}^{2}A_{\rho}=\ell_{r+\delta}A_{\frac{1}{r_{\rho}}}\backslash\ell_{r+\delta}^{1}A_{\frac{1}{r_{\rho}}}$
. It can be shown that

\begin{eqnarray}
 &  & \lim_{\rho\rightarrow\infty}\int_{\ell_{r+\delta}^{1}A_{\frac{1}{r_{\rho}}}}e^{-\lambda\int_{A_{\frac{1}{r_{\rho}}}\backslash B_{\gamma\left(\boldsymbol{x}_{2}\right)}\left(A_{\frac{1}{r_{\rho}}}\right)}g\left(\left\Vert \boldsymbol{y}-\boldsymbol{x}_{2}\right\Vert \right)d\boldsymbol{y}}d\boldsymbol{x}_{2}\nonumber \\
 & \leq & \lim_{\rho\rightarrow\infty}\int_{\ell_{r+\delta}^{1}A_{\frac{1}{r_{\rho}}}}e^{-\lambda\int_{A_{\frac{1}{r_{\rho}}}\backslash B_{\gamma\left(\boldsymbol{x}_{2}\right)}\left(A_{\frac{1}{r_{\rho}}}\right)\cap D\left(\boldsymbol{x}_{2},r_{\rho}^{-\varepsilon}\right)}g\left(\left\Vert \boldsymbol{y}-\boldsymbol{x}_{2}\right\Vert \right)d\boldsymbol{y}}d\boldsymbol{x}_{2}\nonumber \\
 & = & \lim_{\rho\rightarrow\infty}\left(r_{\rho}^{-1}-2r_{\rho}^{-\varepsilon}\right)\times\left(r+\delta\right)e^{-\frac{1}{2}\lambda\int_{D\left(\boldsymbol{0},r_{\rho}^{-\varepsilon}\right)}g\left(\left\Vert \boldsymbol{y}\right\Vert \right)d\boldsymbol{y}}\nonumber \\
 & = & 0\label{eq:convergence of boundary item to zero along the border}\end{eqnarray}
where the last step results due to \eqref{eq:first iterm in the second term step},
which showed that for $g$ satisfying both \eqref{eq:conditions on g(x) - non-increasing}
and \eqref{eq:Condition on g(x) requirement 2} $\lim_{\rho\rightarrow\infty}\rho^{\frac{1}{2}}e^{-\frac{1}{2}\rho r_{\rho}^{2}\int_{D\left(\boldsymbol{0},r_{\rho}^{-\varepsilon}\right)}g\left(\left\Vert \boldsymbol{x}\right\Vert \right)d\boldsymbol{x}}=e^{-\frac{b}{2}}$.
Then the result follows easily from the definition of $r_{\rho}$
in \eqref{eq:definition of r_rho}. Note that the result in \eqref{eq:convergence of boundary item to zero along the border}
cannot be obtained for $g$ satisfying \eqref{eq:conditions on g(x) - non-increasing}
and \eqref{eq:conditions on g(x) - integral boundness} only. 

Using similar steps that resulted in \eqref{eq:third term final result},
it can be shown that

\[
\lim_{\rho\rightarrow\infty}\int_{\ell_{r+\delta}^{2}A_{\frac{1}{r_{\rho}}}}e^{-\lambda\int_{A_{\frac{1}{r_{\rho}}}\backslash B_{\gamma\left(\boldsymbol{x}_{2}\right)}\left(A_{\frac{1}{r_{\rho}}}\right)}g\left(\left\Vert \boldsymbol{y}-\boldsymbol{x}_{2}\right\Vert \right)d\boldsymbol{y}}d\boldsymbol{x}_{2}=0\]
The above equation, together with \eqref{eq:convergence of boundary item to zero along the border},
allows us to conclude that the first term in \eqref{eq:number of finite components in finite area step 4}
converges to $0$ as $\rho\rightarrow\infty$:\begin{equation}
\lim_{\rho\rightarrow\infty}4\lambda^{2}\int_{\ell A_{\frac{1}{r_{\rho}}}}\int_{A_{\frac{1}{r_{\rho}}}\cap D\left(\boldsymbol{x}_{1},\delta\right)}e^{-\lambda\left(\int_{A_{\frac{1}{r_{\rho}}}}g\left(\left\Vert \boldsymbol{y}-\boldsymbol{x}_{2}\right\Vert \right)d\boldsymbol{y}+\beta\left|A_{\frac{1}{r_{\rho}}}\cap D\left(\boldsymbol{x}_{1},r\right)\backslash D\left(\boldsymbol{x}_{2},r\right)\right|-\pi\left\Vert \boldsymbol{x}_{2}-\boldsymbol{x}_{1}\right\Vert ^{2}\right)}d\boldsymbol{x}_{2}d\boldsymbol{x}_{1}=0\label{eq:convergence of boundary item to zero border case}\end{equation}

Now let us consider the second term in \eqref{eq:number of finite components in finite area step 4}.
First it can be shown that\[
\int_{A_{\frac{1}{r_{\rho}}}}g\left(\left\Vert \boldsymbol{y}-\boldsymbol{x}_{2}\right\Vert \right)d\boldsymbol{y}\geq\int_{A_{\frac{1}{r_{\rho}}}\cap D\left(\boldsymbol{x}_{2},r_{\rho}^{-\varepsilon}\right)}g\left(\left\Vert \boldsymbol{y}-\boldsymbol{x}_{2}\right\Vert \right)d\boldsymbol{y}\]
 and for any $\boldsymbol{x}_{2}\in A_{\frac{1}{r_{\rho}}}$, $A_{\frac{1}{r_{\rho}}}\cap D\left(\boldsymbol{x}_{2},r_{\rho}^{-\varepsilon}\right)$
contains at least one quarter of $D\left(\boldsymbol{x}_{2},r_{\rho}^{-\varepsilon}\right)$.
Further, since

\[
\lim_{\rho\rightarrow\infty}\int_{D\left(\boldsymbol{x}_{2},r_{\rho}^{-\varepsilon}\right)}g\left(\left\Vert \boldsymbol{y}-\boldsymbol{x}_{2}\right\Vert \right)d\boldsymbol{y}=C\]
there exists a $\rho_{0}$ such that for $\rho\geq\rho_{0}$ and any
positive constant $\gamma<1$ \[
\int_{D\left(\boldsymbol{x}_{2},r_{\rho}^{-\varepsilon}\right)}g\left(\left\Vert \boldsymbol{y}-\boldsymbol{x}_{2}\right\Vert \right)d\boldsymbol{y}\geq\gamma C\]
As a result of the above discussions, it can be shown that for sufficiently
large $\rho\geq\rho_{0}$ \begin{eqnarray}
 &  & 4\lambda^{2}\int_{\angle A_{\frac{1}{r_{\rho}}}}\int_{A_{\frac{1}{r_{\rho}}}\cap D\left(\boldsymbol{x}_{1},\delta\right)}e^{-\lambda\left(\int_{A_{\frac{1}{r_{\rho}}}}g\left(\left\Vert \boldsymbol{y}-\boldsymbol{x}_{2}\right\Vert \right)d\boldsymbol{y}+\beta\left|A_{\frac{1}{r_{\rho}}}\cap D\left(\boldsymbol{x}_{1},r\right)\backslash D\left(\boldsymbol{x}_{2},r\right)\right|-\pi\left\Vert \boldsymbol{x}_{2}-\boldsymbol{x}_{1}\right\Vert ^{2}\right)}d\boldsymbol{x}_{2}d\boldsymbol{x}_{1}\nonumber \\
 & \leq & 4\lambda^{2}\int_{\angle A_{\frac{1}{r_{\rho}}}}\int_{A_{\frac{1}{r_{\rho}}}\cap D\left(\boldsymbol{x}_{1},\delta\right)}e^{-\lambda\left(\frac{1}{4}\gamma C-\pi\left\Vert \boldsymbol{x}_{2}-\boldsymbol{x}_{1}\right\Vert ^{2}\right)}d\boldsymbol{x}_{2}d\boldsymbol{x}_{1}\nonumber \\
 & \leq & 4\lambda^{2}\pi\delta^{2}r^{2}e^{-\lambda\left(\frac{1}{4}\gamma C-\pi\delta^{2}\right)}\label{eq:corner term boundary case}\end{eqnarray}
where by choosing $\delta<\frac{1}{4\pi}\gamma C$, the above equation
can be easily shown as converging to $0$ as $\rho\rightarrow\infty$. 

In summary, using \eqref{eq:number of finite components in finite area step 4},
\eqref{eq:convergence of boundary item to zero border case} and \eqref{eq:corner term boundary case},
it can be shown that for $\delta<\min\left\{ \frac{1}{4\pi}\gamma C,\frac{r}{2},\frac{1}{2\pi}c_{5}\beta,z_{0}\right\} $,
for the first term in \eqref{eq:number of finite components in finite area step 3},
we have \begin{equation}
\lim_{\rho\rightarrow\infty}\lambda^{2}\int_{B_{r}\left(A_{\frac{1}{r_{\rho}}}\right)}\int_{A_{\frac{1}{r_{\rho}}}\cap D\left(\boldsymbol{x}_{1},\delta\right)}e^{-\lambda\left(\int_{A_{\frac{1}{r_{\rho}}}}g\left(\left\Vert \boldsymbol{y}-\boldsymbol{x}_{2}\right\Vert \right)d\boldsymbol{y}+\beta\left|A_{\frac{1}{r_{\rho}}}\cap D\left(\boldsymbol{x}_{1},r\right)\backslash D\left(\boldsymbol{x}_{2},r\right)\right|-\pi\left\Vert \boldsymbol{x}_{2}-\boldsymbol{x}_{1}\right\Vert ^{2}\right)}d\boldsymbol{x}_{2}d\boldsymbol{x}_{1}=0\label{eq:final result boundary case finite area}\end{equation}

For the second term in \eqref{eq:number of finite components in finite area step 3},
using Lemma \ref{lem:analysis of the intersectional area without border},
it can be shown that \begin{eqnarray*}
 &  & \lambda^{2}\int_{A_{\frac{1}{r_{\rho}}}\backslash B_{r}\left(A_{\frac{1}{r_{\rho}}}\right)}\int_{A_{\frac{1}{r_{\rho}}}\cap D\left(\boldsymbol{x}_{1},\delta\right)}e^{-\lambda\left(\int_{A_{\frac{1}{r_{\rho}}}}g\left(\left\Vert \boldsymbol{y}-\boldsymbol{x}_{2}\right\Vert \right)d\boldsymbol{y}+\beta\left|A_{\frac{1}{r_{\rho}}}\cap D\left(\boldsymbol{x}_{1},r\right)\backslash D\left(\boldsymbol{x}_{2},r\right)\right|-\pi\left\Vert \boldsymbol{x}_{2}-\boldsymbol{x}_{1}\right\Vert ^{2}\right)}d\boldsymbol{x}_{2}d\boldsymbol{x}_{1}\\
 & = & \lambda^{2}\int_{A_{\frac{1}{r_{\rho}}}\backslash B_{r}\left(A_{\frac{1}{r_{\rho}}}\right)}\int_{A_{\frac{1}{r_{\rho}}}\cap D\left(\boldsymbol{x}_{1},\delta\right)}e^{-\lambda\left(\int_{A_{\frac{1}{r_{\rho}}}}g\left(\left\Vert \boldsymbol{y}-\boldsymbol{x}_{2}\right\Vert \right)d\boldsymbol{y}+\beta\left|D\left(\boldsymbol{x}_{1},r\right)\backslash D\left(\boldsymbol{x}_{2},r\right)\right|-\pi\left\Vert \boldsymbol{x}_{2}-\boldsymbol{x}_{1}\right\Vert ^{2}\right)}d\boldsymbol{x}_{2}d\boldsymbol{x}_{1}\\
 & \leq & \lambda^{2}\int_{A_{\frac{1}{r_{\rho}}}\backslash B_{r}\left(A_{\frac{1}{r_{\rho}}}\right)}\int_{A_{\frac{1}{r_{\rho}}}\cap D\left(\boldsymbol{x}_{1},\delta\right)}e^{-\lambda\left(\int_{A_{\frac{1}{r_{\rho}}}}g\left(\left\Vert \boldsymbol{y}-\boldsymbol{x}_{2}\right\Vert \right)d\boldsymbol{y}+\left(\beta\sqrt{3}r-\pi\delta\right)\left\Vert \boldsymbol{x}_{2}-\boldsymbol{x}_{1}\right\Vert \right)}d\boldsymbol{x}_{2}d\boldsymbol{x}_{1}\\
 & \leq & \lambda^{2}\int_{A_{\frac{1}{r_{\rho}}}}\int_{A_{\frac{1}{r_{\rho}}}\cap D\left(\boldsymbol{x}_{1},\delta\right)}e^{-\lambda\left(\int_{A_{\frac{1}{r_{\rho}}}}g\left(\left\Vert \boldsymbol{y}-\boldsymbol{x}_{2}\right\Vert \right)d\boldsymbol{y}+\left(\beta\sqrt{3}r-\pi\delta\right)\left\Vert \boldsymbol{x}_{2}-\boldsymbol{x}_{1}\right\Vert \right)}d\boldsymbol{x}_{2}d\boldsymbol{x}_{1}\\
 & = & \lambda^{2}\int_{A_{\frac{1}{r_{\rho}}}}\int_{A_{\frac{1}{r_{\rho}}}\cap D\left(\boldsymbol{x}_{2},\delta\right)}e^{-\lambda\left(\beta\sqrt{3}r-\pi\delta\right)\left\Vert \boldsymbol{x}_{2}-\boldsymbol{x}_{1}\right\Vert }d\boldsymbol{x}_{1}e^{-\lambda\int_{A_{\frac{1}{r_{\rho}}}}g\left(\left\Vert \boldsymbol{y}-\boldsymbol{x}_{2}\right\Vert \right)d\boldsymbol{y}}d\boldsymbol{x}_{2}\\
 & \leq & \frac{1-e^{-\lambda\left(\beta\sqrt{3}r-\pi\delta\right)\delta}\left(1+\lambda\left(\beta\sqrt{3}r-\pi\delta\right)\delta\right)}{\lambda\left(\beta\sqrt{3}r-\pi\delta\right)^{2}}\lambda\int_{A_{\rho}}e^{-\lambda\int_{A_{\frac{1}{r_{\rho}}}}g\left(\left\Vert \boldsymbol{y}-\boldsymbol{x}_{2}\right\Vert \right)d\boldsymbol{y}}d\boldsymbol{x}_{2}\end{eqnarray*}

In Theorem \ref{thm:expected isolated nodes asymptotically infinite square},
we have established that\begin{eqnarray*}
 &  & \lim_{\rho\rightarrow\infty}\lambda\int_{A_{\frac{1}{r_{\rho}}}}e^{-\lambda\int_{A_{\frac{1}{r_{\rho}}}}g\left(\left\Vert \boldsymbol{y}-\boldsymbol{x}_{2}\right\Vert \right)d\boldsymbol{y}}d\boldsymbol{x}_{2}=e^{-b}\end{eqnarray*}
Therefore it follows straightforwardly that for $\delta<\beta\sqrt{3}r/\pi$\begin{equation}
\lim_{\rho\rightarrow\infty}\lambda^{2}\int_{A_{\frac{1}{r_{\rho}}}\backslash B_{r}\left(A_{\frac{1}{r_{\rho}}}\right)}\int_{A_{\frac{1}{r_{\rho}}}\cap D\left(\boldsymbol{x}_{1},\delta\right)}e^{-\lambda\left(\int_{A_{\frac{1}{r_{\rho}}}}g\left(\left\Vert \boldsymbol{y}-\boldsymbol{x}_{2}\right\Vert \right)d\boldsymbol{y}+\beta\left|A_{\frac{1}{r_{\rho}}}\cap D\left(\boldsymbol{x}_{1},r\right)\backslash D\left(\boldsymbol{x}_{2},r\right)\right|-\pi\left\Vert \boldsymbol{x}_{2}-\boldsymbol{x}_{1}\right\Vert ^{2}\right)}d\boldsymbol{x}_{2}d\boldsymbol{x}_{1}=0\label{eq:final results non-boundary case finite area}\end{equation}
Using \eqref{eq:number of finite components in finite area step 2},
\eqref{eq:number of finite components in finite area step 3}, \eqref{eq:final result boundary case finite area}
and \eqref{eq:final results non-boundary case finite area}, we are
able to conclude that by chosing $\delta$ to be a positive constant
such that \[
\delta<\min\left\{ \frac{1}{4\pi}\gamma C,\frac{r}{2},\frac{1}{2\pi}c_{5}\beta,\beta\sqrt{3}r/\pi,z_{0}\right\} \]
 \begin{equation}
\lim_{\rho\rightarrow\infty}\sum_{k=2}^{\infty}E\left(\xi_{k,1}\right)=0\label{eq:number of finite components in finite area final result}\end{equation}

\subsection*{An analysis of the second term in \eqref{eq:analysis on the expected number of components size k step 1}}

Now let us consider the second term in \eqref{eq:analysis on the expected number of components size k step 1},
i.e. \begin{eqnarray*}
 &  & E\left(\xi_{k,2}\right)\\
 & = & \frac{\lambda^{k}}{k!}\int_{A_{\frac{1}{r_{\rho}}}}\int_{\left(A_{\frac{1}{r_{\rho}}}\right)^{k-1}\backslash\left(D\left(\boldsymbol{x}_{1},\delta\right)\right)^{k-1}}g_{1}\left(\boldsymbol{x}_{1},\boldsymbol{x}_{2},\ldots,\boldsymbol{x}_{k}\right)e^{-\lambda\int_{A_{\frac{1}{r_{\rho}}}}g_{2}\left(\boldsymbol{y};\boldsymbol{x}_{1},\boldsymbol{x}_{2},\ldots,\boldsymbol{x}_{k}\right)d\boldsymbol{y}}d\left(\boldsymbol{x}_{2}\cdots\boldsymbol{x}_{k}\right)d\boldsymbol{x}_{1}\end{eqnarray*}

For $\left(\boldsymbol{x}_{2}\cdots\boldsymbol{x}_{k}\right)\in\left(A_{\frac{1}{r_{\rho}}}\right)^{k-1}\backslash\left(D\left(\boldsymbol{x}_{1},\delta\right)\right)^{k-1}$,
there is one node in $\left\{ \boldsymbol{x}_{2}\cdots\boldsymbol{x}_{k}\right\} $
outside a Euclidean distance $\delta$ of $\boldsymbol{x}_{1}$ and
belongs to $A_{\frac{1}{r_{\rho}}}\backslash D\left(\boldsymbol{x}_{1},\delta\right)$.
Without losing generality, assume that node is $\boldsymbol{x}_{j}\in A_{\frac{1}{r_{\rho}}}\backslash D\left(\boldsymbol{x}_{1},\delta\right)$,
where $j\in\Gamma_{k}/\left\{ 1\right\} $. 

Let $\angle A_{\frac{1}{r_{\rho}}}\subset A_{\frac{1}{r_{\rho}}}$
be a square area of size $r\times r$ located at a corner of $A_{\frac{1}{r_{\rho}}}$
as defined in the beginning of this Appendix and let $\overline{\angle A_{\frac{1}{r_{\rho}}}}\subset A_{\frac{1}{r_{\rho}}}$
be an area in $A_{\frac{1}{r_{\rho}}}$ excluding the four corner
squares $\angle A_{\frac{1}{r_{\rho}}}$. It is straightforward from
the proofs of Lemmas \ref{lem:analysis of the intersectional area boundary case}
and \ref{lem:analysis of the intersectional area without border}
that for $\boldsymbol{x}_{1}\in\overline{\angle A_{\frac{1}{r_{\rho}}}}$
and $\boldsymbol{x}_{j}\in A_{\frac{1}{r_{\rho}}}\backslash D\left(\boldsymbol{x}_{1},\delta\right)$,
i.e. $\left\Vert \boldsymbol{x}_{j}-\boldsymbol{x}_{1}\right\Vert \geq\delta$,
there exists a positive constant $c_{6}>0$, depending on $\delta$,
such that \[
\left|A_{\frac{1}{r_{\rho}}}\cap D\left(\boldsymbol{x}_{1},r\right)\backslash D\left(\boldsymbol{x}_{2},r\right)\right|\geq c_{6}\]
Using the above inequality and \eqref{eq:inequality for g_2}, it
follows that

\begin{eqnarray*}
 &  & \int_{A_{\frac{1}{r_{\rho}}}}g_{2}\left(\boldsymbol{y};\boldsymbol{x}_{1},\boldsymbol{x}_{2},\ldots,\boldsymbol{x}_{k}\right)d\boldsymbol{y}\\
 & \geq & \int_{A_{\frac{1}{r_{\rho}}}}g_{2}\left(\boldsymbol{y};\boldsymbol{x}_{1},\boldsymbol{x}_{j}\right)d\boldsymbol{y}\\
 & = & \int_{A_{\frac{1}{r_{\rho}}}}g\left(\left\Vert \boldsymbol{y}-\boldsymbol{x}_{j}\right\Vert \right)+g\left(\left\Vert \boldsymbol{y}-\boldsymbol{x}_{1}\right\Vert \right)\left(1-g\left(\left\Vert \boldsymbol{y}-\boldsymbol{x}_{j}\right\Vert \right)\right)d\boldsymbol{y}\\
 & \geq & \int_{A_{\frac{1}{r_{\rho}}}}g\left(\left\Vert \boldsymbol{y}-\boldsymbol{x}_{j}\right\Vert \right)d\boldsymbol{y}+\beta\left|A_{\frac{1}{r_{\rho}}}\cap D\left(\boldsymbol{x}_{1},r\right)\backslash D\left(\boldsymbol{x}_{j},r\right)\right|\\
 & \geq & \int_{A_{\frac{1}{r_{\rho}}}}g\left(\left\Vert \boldsymbol{y}-\boldsymbol{x}_{j}\right\Vert \right)d\boldsymbol{y}+\beta c_{6}\end{eqnarray*}

Therefore \begin{eqnarray}
 &  & \frac{\lambda^{k}}{k!}\int_{\overline{\angle A_{\frac{1}{r_{\rho}}}}}\int_{\left(A_{\frac{1}{r_{\rho}}}\right)^{k-1}\backslash\left(D\left(\boldsymbol{x}_{1},\delta\right)\right)^{k-1}}g_{1}\left(\boldsymbol{x}_{1},\boldsymbol{x}_{2},\ldots,\boldsymbol{x}_{k}\right)e^{-\lambda\int_{A_{\frac{1}{r_{\rho}}}}g_{2}\left(\boldsymbol{y};\boldsymbol{x}_{1},\boldsymbol{x}_{2},\ldots,\boldsymbol{x}_{k}\right)d\boldsymbol{y}}d\left(\boldsymbol{x}_{2}\cdots\boldsymbol{x}_{k}\right)d\boldsymbol{x}_{1}\nonumber \\
 & \leq & \frac{\lambda^{k}}{k!}\int_{\overline{\angle A_{\frac{1}{r_{\rho}}}}}\int_{\left(A_{\frac{1}{r_{\rho}}}\right)^{k-1}\backslash\left(D\left(\boldsymbol{x}_{1},\delta\right)\right)^{k-1}}g_{1}\left(\boldsymbol{x}_{1},\boldsymbol{x}_{2},\ldots,\boldsymbol{x}_{k}\right)e^{-\lambda\int_{A_{\frac{1}{r_{\rho}}}}g\left(\left\Vert \boldsymbol{y}-\boldsymbol{x}_{j}\right\Vert \right)d\boldsymbol{y}-\lambda\beta c_{6}}d\left(\boldsymbol{x}_{2}\cdots\boldsymbol{x}_{k}\right)d\boldsymbol{x}_{1}\nonumber \\
 & \leq & \frac{\lambda^{k}}{k!}\int_{\left(A_{\frac{1}{r_{\rho}}}\right)^{k}}g_{1}\left(\boldsymbol{x}_{1},\boldsymbol{x}_{2},\ldots,\boldsymbol{x}_{k}\right)e^{-\lambda\int_{A_{\frac{1}{r_{\rho}}}}g\left(\left\Vert \boldsymbol{y}-\boldsymbol{x}_{j}\right\Vert \right)d\boldsymbol{y}-\lambda\beta c_{6}}d\left(\boldsymbol{x}_{1}\boldsymbol{x}_{2}\cdots\boldsymbol{x}_{k}\right)\nonumber \\
 & = & \frac{\lambda^{k}}{k!}\int_{\left(A_{\frac{1}{r_{\rho}}}\right)^{k}}g_{1}\left(\boldsymbol{x}_{1},\boldsymbol{x}_{2},\ldots,\boldsymbol{x}_{k}\right)e^{-\lambda\int_{A_{\frac{1}{r_{\rho}}}}g\left(\left\Vert \boldsymbol{y}-\boldsymbol{x}_{1}\right\Vert \right)d\boldsymbol{y}-\lambda\beta c_{6}}d\left(\boldsymbol{x}_{1}\boldsymbol{x}_{2}\cdots\boldsymbol{x}_{k}\right)\label{eq:number of finite components non-border step 1}\end{eqnarray}
where a re-numbering of the nodes occurred in the last step of the
above equation. First using Lemma \ref{lem:inequality on g_1}, and
then using \eqref{eq:inequality for g_2 the union bound} and the
inequality that $\int_{A_{\frac{1}{r_{\rho}}}}g\left(\left\Vert \boldsymbol{x}_{j}-\boldsymbol{x}_{i}\right\Vert \right)d\boldsymbol{x}_{i}\leq C$,
it can be shown that\begin{eqnarray}
 &  & \frac{\lambda^{k}}{k!}\int_{\left(A_{\frac{1}{r_{\rho}}}\right)^{k}}g_{1}\left(\boldsymbol{x}_{1},\boldsymbol{x}_{2},\ldots,\boldsymbol{x}_{k}\right)e^{-\lambda\int_{A_{\frac{1}{r_{\rho}}}}g\left(\left\Vert \boldsymbol{y}-\boldsymbol{x}_{1}\right\Vert \right)d\boldsymbol{y}-\lambda\beta c_{6}}d\left(\boldsymbol{x}_{1}\boldsymbol{x}_{2}\cdots\boldsymbol{x}_{k}\right)\nonumber \\
 & \leq & \frac{\lambda^{k}}{k!}\int_{\left(A_{\frac{1}{r_{\rho}}}\right)^{k}}\sum_{i_{2}\in\Gamma_{k}\backslash\left\{ 1\right\} ,\cdots,i_{k}\in\Gamma_{k}\backslash\left\{ 1,i_{2},\ldots,i_{k-1}\right\} }g_{2}\left(\boldsymbol{x}_{i_{2}};\boldsymbol{x}_{1}\right)\cdots g_{2}\left(\boldsymbol{x}_{i_{k}};\boldsymbol{x}_{1},\boldsymbol{x}_{i_{2}},\ldots,\boldsymbol{x}_{i_{k-1}}\right)\nonumber \\
 & \times & e^{-\lambda\int_{A_{\frac{1}{r_{\rho}}}}g\left(\left\Vert \boldsymbol{y}-\boldsymbol{x}_{1}\right\Vert \right)d\boldsymbol{y}-\lambda\beta c_{6}}d\left(\boldsymbol{x}_{i_{k}}\cdots\boldsymbol{x}_{i_{2}}\boldsymbol{x}_{1}\right)\nonumber \\
 & \leq & \frac{\lambda^{k}C^{k-1}}{k!}\left(k-1\right)!\left(k-1\right)!\int_{A_{\frac{1}{r_{\rho}}}}e^{-\lambda\int_{A_{\frac{1}{r_{\rho}}}}g\left(\left\Vert \boldsymbol{y}-\boldsymbol{x}_{1}\right\Vert \right)d\boldsymbol{y}-\lambda\beta c_{6}}d\boldsymbol{x}_{1}\nonumber \\
 & = & \frac{\left(k-1\right)!}{k}e^{-\frac{b\beta c_{6}}{C}}\times\frac{\left(\log\rho+b\right)^{k-1}}{\rho^{\frac{\beta c_{6}}{C}}}\times\lambda\int_{A_{\frac{1}{r_{\rho}}}}e^{-\lambda\int_{A_{\rho}}g\left(\left\Vert \boldsymbol{y}-\boldsymbol{x}_{1}\right\Vert \right)d\boldsymbol{y}}d\boldsymbol{x}_{1}\label{eq:number of finite components non-border step 2}\end{eqnarray}

Using Theorem \ref{thm:expected isolated nodes asymptotically infinite square},
\eqref{eq:number of finite components non-border step 1} and \eqref{eq:number of finite components non-border step 2},
it follows that

\begin{equation}
\lim_{\rho\rightarrow\infty}\frac{\lambda^{k}}{k!}\int_{\overline{\angle A_{\frac{1}{r_{\rho}}}}}\int_{\left(A_{\frac{1}{r_{\rho}}}\right)^{k-1}\backslash\left(D\left(\boldsymbol{x}_{1},\delta\right)\right)^{k-1}}g_{1}\left(\boldsymbol{x}_{1},\boldsymbol{x}_{2},\ldots,\boldsymbol{x}_{k}\right)e^{-\lambda\int_{A_{\frac{1}{r_{\rho}}}}g_{2}\left(\boldsymbol{y};\boldsymbol{x}_{1},\boldsymbol{x}_{2},\ldots,\boldsymbol{x}_{k}\right)d\boldsymbol{y}}d\left(\boldsymbol{x}_{2}\cdots\boldsymbol{x}_{k}\right)d\boldsymbol{x}_{1}=0\label{eq:number of finite components in infinite area - non-border case}\end{equation}

Using similar steps as those leading to \eqref{eq:number of finite components non-border step 2},
it can be shown that

\begin{eqnarray*}
 &  & \frac{\lambda^{k}}{k!}\int_{\angle A_{\frac{1}{r_{\rho}}}}\int_{\left(A_{\frac{1}{r_{\rho}}}\right)^{k-1}\backslash\left(D\left(\boldsymbol{x}_{1},\delta\right)\right)^{k-1}}g_{1}\left(\boldsymbol{x}_{1},\boldsymbol{x}_{2},\ldots,\boldsymbol{x}_{k}\right)e^{-\lambda\int_{A_{\frac{1}{r_{\rho}}}}g_{2}\left(\boldsymbol{y};\boldsymbol{x}_{1},\boldsymbol{x}_{2},\ldots,\boldsymbol{x}_{k}\right)d\boldsymbol{y}}d\left(\boldsymbol{x}_{2}\cdots\boldsymbol{x}_{k}\right)d\boldsymbol{x}_{1}\\
 & \leq & \frac{\lambda^{k}}{k!}\int_{\angle A_{\frac{1}{r_{\rho}}}}\int_{\left(A_{\frac{1}{r_{\rho}}}\right)^{k-1}}g_{1}\left(\boldsymbol{x}_{1},\boldsymbol{x}_{2},\ldots,\boldsymbol{x}_{k}\right)e^{-\lambda\int_{A_{\frac{1}{r_{\rho}}}}g_{2}\left(\left\Vert \boldsymbol{y}-\boldsymbol{x}_{1}\right\Vert \right)d\boldsymbol{y}}d\left(\boldsymbol{x}_{2}\cdots\boldsymbol{x}_{k}\right)d\boldsymbol{x}_{1}\\
 & \leq & \frac{\lambda^{k}C^{k-1}}{k}\left(k-1\right)!\int_{\angle A_{\frac{1}{r_{\rho}}}}e^{-\lambda\int_{A_{\frac{1}{r_{\rho}}}}g\left(\left\Vert \boldsymbol{y}-\boldsymbol{x}_{1}\right\Vert \right)d\boldsymbol{y}}d\boldsymbol{x}_{1}\end{eqnarray*}

Using similar steps that resulted in \eqref{eq:corner term boundary case},
it can be shown that

\begin{eqnarray}
 &  & \lim_{\rho\rightarrow\infty}\frac{\lambda^{k}C^{k-1}}{k}\left(k-1\right)!\int_{\angle A_{\frac{1}{r_{\rho}}}}e^{-\lambda\int_{A_{\frac{1}{r_{\rho}}}}g\left(\left\Vert \boldsymbol{y}-\boldsymbol{x}_{1}\right\Vert \right)d\boldsymbol{y}-\lambda\beta c_{6}}d\boldsymbol{x}_{1}\nonumber \\
 & \leq & \lim_{\rho\rightarrow\infty}\frac{\lambda^{k}C^{k-1}}{k}\left(k-1\right)!\delta^{2}e^{-\frac{1}{4}\lambda\gamma C}\nonumber \\
 & = & 0\label{eq:number of finite components in infinite area - corner case}\end{eqnarray}

The combination of \eqref{eq:number of finite components in infinite area - non-border case}
and \eqref{eq:number of finite components in infinite area - corner case}
allows us to conclude that \[
\lim_{\rho\rightarrow\infty}E\left(\xi_{k,2}\right)=0\]

It follows that for any fixed but arbitrarily large integer $M$ \begin{equation}
\lim_{\rho\rightarrow\infty}\sum_{k=2}^{M}E\left(\xi_{k,2}\right)=0\label{eq:number of finite components in infinite area final result}\end{equation}

Finally from \eqref{eq:number of finite components in finite area final result}
and \eqref{eq:number of finite components in infinite area final result},
we conclude that

\[
\lim_{\rho\rightarrow\infty}\left(\sum_{k=2}^{M}E\left(\xi_{k}\right)=0\right)=0\]
Noting that $\xi_{k}$ is a non-negative integer, therefore\[
\lim_{\rho\rightarrow\infty}\Pr\left(\sum_{k=2}^{M}\xi_{k}=0\right)=1\]

\section*{Acknowledgment}

The authors would like to thank Prof P. R. Kumar of University of
Illinois, Urbana-Champaign for his comments on an earlier version
of this paper.

\bibliographystyle{IEEEtran}
\bibliography{/Users/gmao/mgq/localtexmf/bibtex/BIB/SensorNetwork}

\begin{biography}
{Guoqiang Mao} (S'98\textendash{}M'02\textendash{}SM\textquoteright{}08)
received PhD in telecommunications engineering in 2002 from Edith
Cowan University. He joined the School of Electrical and Information
Engineering, the University of Sydney in December 2002 where he is
a Senior Lecturer now. His research interests include wireless localization
techniques, wireless multihop networks, graph theory and its application
in networking, and network performance analysis. He is a Senior Member
of IEEE and an Associate Editor of IEEE Transactions on Vehicular
Technology.
\end{biography}

\begin{biography}
{Brian D.O. Anderson} (S\textquoteright{}62\textendash{}M\textquoteright{}66\textendash{}SM\textquoteright{}74\textendash{}F\textquoteright{}75\textendash{}
LF\textquoteright{}07) was born in Sydney, Australia, and educated
at Sydney University in mathematics and electrical engineering, with
PhD in electrical engineering from Stanford University in 1966. He
is a Distinguished Professor at the Australian National University
and Distinguished Researcher in National ICT Australia. His awards
include the IEEE Control Systems Award of 1997, the 2001 IEEE James
H Mulligan, Jr Education Medal, and the Bode Prize of the IEEE Control
System Society in 1992, as well as several IEEE best paper prizes.
He is a Fellow of the Australian Academy of Science, the Australian
Academy of Technological Sciences and Engineering, the Royal Society,
and a foreign associate of the National Academy of Engineering. His
current research interests are in distributed control, sensor networks
and econometric modelling. 
\end{biography}

\end{document}